\documentclass[11pt,reqno]{amsart}
\numberwithin{equation}{section}
\usepackage{amsmath,amsthm}
\usepackage{mathptmx}  
\usepackage{amssymb}
\usepackage{multicol}
\DeclareFontFamily{U}{BOONDOX-calo}{\skewchar\font=45 }
\DeclareFontShape{U}{BOONDOX-calo}{m}{n}{
  <-> s*[1.05] BOONDOX-r-calo}{}
\DeclareFontShape{U}{BOONDOX-calo}{b}{n}{
  <-> s*[1.05] BOONDOX-b-calo}{}
\DeclareMathAlphabet{\mathcalboondox}{U}{BOONDOX-calo}{m}{n}
\SetMathAlphabet{\mathcalboondox}{bold}{U}{BOONDOX-calo}{b}{n}
\DeclareMathAlphabet{\mathbcalboondox}{U}{BOONDOX-calo}{b}{n}

\usepackage{mathtools}
\usepackage{xfrac}

\newcommand{\mcb}[1]{{\mathcalboondox #1}}
\usepackage{amssymb}
\usepackage{amsmath}
\usepackage{amsthm}
\usepackage{diagbox}
\usepackage{booktabs}
\usepackage{enumitem}
\usepackage[sans]{dsfont} 
\usepackage{bbm}
\usepackage{changebar}
\usepackage[colorlinks]{hyperref}
\usepackage{a4wide}
\usepackage[usenames,dvipsnames,svgnames,table]{xcolor}
\usepackage{xcolor}

\usepackage{tikz}
\usetikzlibrary{arrows,shapes,automata,petri,positioning,calc}
\tikzset{
    place/.style={
        circle,
        thick,
        draw=black,
        fill=gray!50,
        minimum size=20mm,
    },
        state/.style={
        circle,
        thick,
        draw=blue!75,
        fill=blue!20,
        minimum size=20mm,
    },
}

\tikzset{
    cross/.pic = {
    \draw[rotate = 45] (-0.2,0) -- (0.2,0);
    \draw[rotate = 45] (0,-0.2) -- (0, 0.2);
    }
}

\usepackage{ulem}
\usepackage{setspace}

\newtheorem{thm}{Theorem}[section]
\newtheorem{lem}[thm]{Lemma}
\newtheorem{cor}[thm]{Corollary}

\newtheorem{prop}[thm]{Proposition}

\newtheorem{definition}[thm]{Definition}

\newtheorem{rem}[thm]{Remark}

\newcommand\ve{\varepsilon}

\newcommand{\bb}[1]{{\mathbb #1}}

\usepackage{comment}

\title[{Hydrodynamic behavior of long-range symmetric exclusion with a slow barrier}
]{Hydrodynamic behavior of  long-range symmetric\\ exclusion with a slow barrier: superdiffusive regime}
\author{Pedro Cardoso, Patr\'icia   Gon\c calves, Byron Jim\'enez-Oviedo}

\begin{document}
\begin{abstract}
We analyse the hydrodynamical behavior of the long jumps   symmetric exclusion process
in the presence of a slow barrier. The jump rates are  given by a symmetric  transition probability $p(\cdot)$ with infinite variance. When jumps occur from $\mathbb{Z}_{-}^{*}$ to $\mathbb N$ the rates are slowed down by a factor $\alpha n^{-\beta}$ (with $\alpha>0$ and $\beta\geq 0$). We obtain  several partial differential equations given in terms of the regional fractional Laplacian on $\mathbb R^*$ and with different boundary conditions. Surprisingly, in opposition to the diffusive regime, we get different regimes depending on whether $\alpha=1$ (all bonds with the same rate) or $\alpha\neq 1$. 
\end{abstract}
\maketitle

\section{Introduction}

In the field of Statistical Mechanics, it is common to derive the macroscopic properties of some fluids from the microscopic interactions of its molecules. This procedure is usually done by using interacting particle systems (IPS), which were introduced in the mathematics community  in \cite{spitzer}, and modelling the particles' movement under these systems, the evolution of each particle is assumed to be stochastic. Typically, in many problems, there is a very large number of particles, placed on certain sites of a lattice, that evolve according to some stochastic rule and whose dynamics conserves one or more quantities; and the goal, then consists of studying the temporal evolution of the conserved quantities.  One of the most classical IPS is the exclusion process, where particles obey an \textit{exclusion rule} that allows at most one particle per site. A particularly interesting macroscopic characterization of the exclusion process  is its \textit{hydrodynamic limit}, where one derives one or several PDEs that describe the space/time evolution of a physical quantity (for instance, the density of particles). 

The hydrodynamic limit of the symmetric exclusion process has been studied in a variety of different settings. In \cite{kipnis1998scaling}, we see the case when particles perform nearest-neighbor jumps; in \cite{tertuaihp, franco2015phase}, the case in which particles perform nearest-neighbor jumps, but the  jump rate at the bond $\{-1,0\}$ is slowed down with respect to the jump rate in all the other bonds; and in \cite{jara2009hydrodynamic}, the case when particles perform long jumps according to the transition probability given in \eqref{eq:trans_prob} with $\gamma<2$, i.e. with infinite variance, was considered. Recently, in \cite{byrondif} and \cite{casodif}, the case of the exclusion process with long  jumps given by the transition probability described in \eqref{eq:trans_prob} with $\gamma>2$, so that it has finite variance, was studied. In all the previous models when the transition probability has finite variance, the hydrodynamic limit can be derived by speeding up the process in the diffusive time scale $tn^2$  and  the hydrodynamic limit is described by the heat equation, given in terms  of the usual Laplacian,  which is a local operator. However,  when the transition probability has an infinite variance, i.e. when $\gamma<2$, the hydrodynamic limit has a different behavior, as is seen in \cite{jara2009hydrodynamic}, \cite{byronsdif} and \cite{stefano}. In all the aforementioned  articles, the hydrodynamic limit was obtained by speeding up the process in the super-diffusive time scale $tn^{\gamma}$, and the hydrodynamic equation is given in terms of the fractional Laplacian or the regional fractional Laplacian but on a bounded domain.

In this article, we are mainly inspired by putting together the two works  \cite{jara2009hydrodynamic} and \cite{franco2015phase}. Our  model is also an exclusion process where particles move on a infinite lattice as in \cite{jara2009hydrodynamic}, but we introduce slow bonds connecting negative to positive sites, with a crossing rate given  by  $~\frac{\alpha}{2 n^{\beta}}p(\cdot)$, where $\alpha >0$ and $\beta \geq 0$, see Figure \ref{fig:dynamics}. The presence of these slow bonds will hinder the transport of mass between $(-\infty,0)$ and $[0, \infty)$ and this will have an impact  at the  macroscopic level as we will see below.

\begin{figure}[h]
    \centering
    \includegraphics{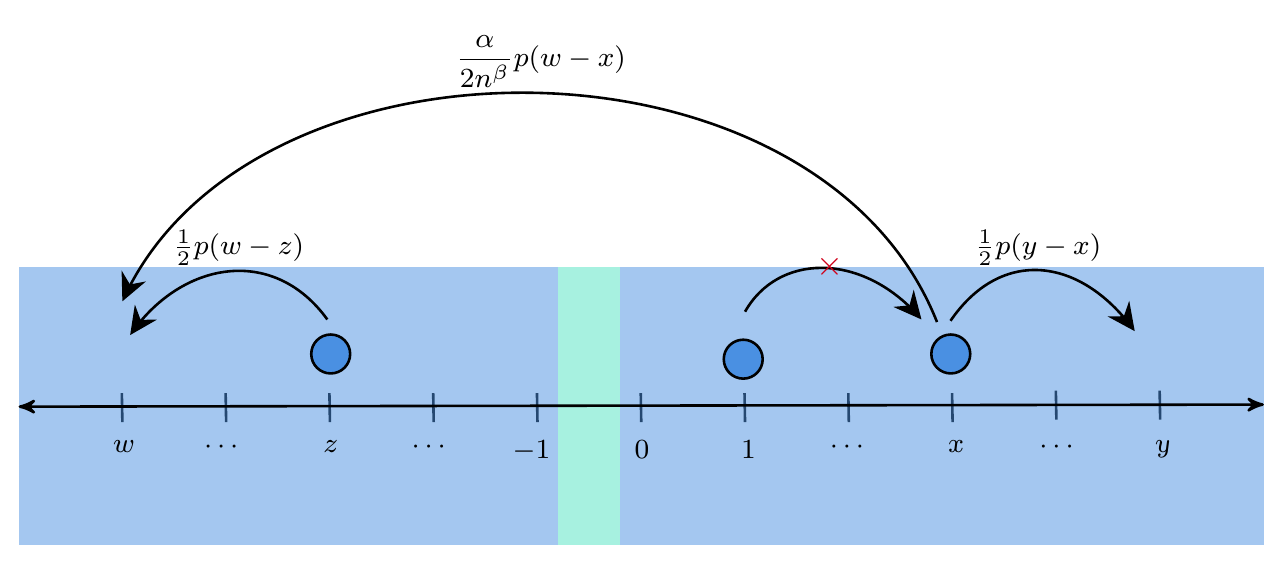}\caption{Dynamics of the long-jumps symmetric exclusion  with slow bonds.}\label{fig:dynamics}
\end{figure}

 In the companion article \cite{CGJ2},  we analysed the case of $p(\cdot)$ with finite variance, and by speeding up the process in the diffusive time scale, we derived the heat equation with different boundary conditions depending on the regime of $\beta$. More precisely, when $\beta>1$ we obtained Neumann boundary conditions; when $\beta=1$ we obtained Robin boundary conditions and when $\beta<1$ we do not see any boundary effect at the macroscopic level. These results are in accordance with what was obtained in \cite{tertuaihp,franco2015phase} for the nearest-neighbor case. 

In this article we analyse the case of $p(\cdot)$ with an infinite variance. Since the jump rate from negative sites to positive sites in the lattice is slowed down with respect to all the other rates, by speeding up the process in the super-diffusive time scale we get an  hydrodynamic equation
given in terms of  the (regional) fractional Laplacian, as described in \eqref{deflapfrac} and \eqref{deflapfracreg}, which is a nonlocal operator and is not as classical, in the PDE literature, as the usual Laplacian.

Our method of proof is entropy method of \cite{GPV}. It consists of showing tightness of the sequence of empirical measures associated to the density; and then characterizing uniquely the limiting point by showing that it is  a Dirac measure on the trajectory of measures absolutely continuous with respect to the Lebesgue measure, whose density is a weak solution to the hydrodynamic equation. If the uniqueness of the weak solution is proved, then the convergence of the whole sequence to its limiting point follows. 

We observe that since  the slowed jump rates are given by  $\frac{\alpha}{2 n^{\beta}}p(\cdot)$, we obtain  Neumann boundary conditions or Robin boundary conditions, depending on the value of $\beta$ and $\gamma$. More precisely, if $(\alpha, \beta)=(1,0)$ we trivially recover the setting of \cite{jara2009hydrodynamic}, where the fractional  heat equation   without boundary conditions was derived; and the same behavior is obtained if we have "few" slow bonds, since their macroscopic effect is negligible in the limit. This is a sort of a geometric condition which is stated more precisely in \eqref{defsigmas}. On the other hand, the majority of the difficulties in our work are related to  the more interesting case where all the bounds that connect $\mathbb{Z}_{-}^{*}$ to $\mathbb{N}$ are slow, creating thus a slow barrier from negative sites to positive sites in lattice.  In this setting, we obtain  hydrodynamic equations which have similarities to those obtained in  \cite{stefano}. The hydrodynamic equations are written in terms of the regional fractional Laplacian restricted to $\mathbb R^*$ and the boundary conditions are given in terms of fractional derivatives as stated in \cite{reflected} and \cite{Guan2006}. We highlight that the appearance of the regional fractional Laplacian is not only restricted to the slow case, but also when $\beta=0$ but $\alpha \neq 1$, where we find a superposition of the fractional Laplacian and of the regional fractional Laplacian, which is rather surprising. In this regime, the rates from negative to positive sites differ with respect to all other rates by only a constant, i.e. there is no difference with respect to the size of the system $n$. Even so, the impact of this change at the macroscopic level is to really change the nature of the PDE from the fractional heat equation (i.e. when $\beta=0$ and $\alpha=1$) to a PDE given by a mixture of the fractional Laplacian and the regional fractional Laplacian on $\mathbb R^*$. 

Let us now  discuss the slow case, i.e. $\beta >0$. If  $\gamma \in (0,1]$ or $\gamma \in (1,2)$ and $\beta > \gamma-1$, we get fractional Neumann boundary conditions; while if $\gamma \in (1,2)$ and $\beta = \gamma-1$, we get fractional Robin boundary conditions.  Finally, if $\gamma \in (1,2)$ and $\beta \in (0, \gamma-1)$, we have a  hydrodynamic equation also written in terms of the regional fractional Laplacian. In all the case, but the last, we were able to prove the uniqueness of our weak solutions, which, as we mentioned above, is a crucial ingredient to prove our results. In the last case,  unfortunately, we were not able to prove the uniqueness of weak solutions, despite our best efforts. This equation has similarities  to the one obtained in \cite{casodif} for $\beta \in (0,1)$, but since we do not have a fractional version of Proposition B.3 in \cite{casodif}, we were not able to apply Oleinik's trick to prove such uniqueness result.  Nevertheless, with our arguments, we not only prove the existence of weak solutions to that equation, but we also characterize them in a very reasonable way, by showing that despite the equation being given in terms of the regional fractional Laplacian on $\mathbb R^*$, the solution is continuous at the origin and has some other properties, as described in Remark \ref{rem:cons}.

Now we  discuss the main difficulties in the development of our proofs. In the microscopic setting, in the critical case $\gamma \in (1,2)$ and $\beta=\gamma-1$, a technical replacement lemma  is needed in order to obtain fractional Robin boundary conditions. In order to prove the required  two-blocks estimate, we make use of a moving particle lemma, reminiscent of the work developed  in \cite{stefano}, but the arguments involved are quite technical due to the presence of slow bonds. Regarding the macroscopic setting, our main difficulty is to deal with PDEs given in terms of  a regional fractional Laplacian defined on unbounded domains. This operator is not  discussed in the literature as the regional fractional Laplacian defined on bounded domains (as was obtained in \cite{stefano}). Therefore, we did not find results regarding the uniqueness of weak solutions of any of the  PDEs that we obtain, and we had to derive them by ourselves. Moreover, we did not find results ensuring that the regional fractional Laplacian defined on the semi-line  is well-defined on our set of test functions, nor that we can apply the integration by parts formula as stated in \cite{reflected} (which deals with the regional fractional Laplacian on bounded intervals). Again,  we had to derive these proofs to fit our  setting. Finally, since we cover all the range of the parameters, the passage from the microscopic to the macroscopic level uses several results which need very precise estimates that  had to be carefully done as, for example,  Proposition \ref{lemconvneum}.

Finally, the critical  case $\gamma=2$ (for which the transition probability has an infinite variance) is not considered in this article, because, when speeding up the process in the $\frac{n^2}{\log(n)}$ time scale,  at the macroscopic level, we  obtain the usual heat equation instead of a fractional diffusion equation. This is similar to what has been done in \cite{stefano_patricia}. We leave this for future work.

Here follows an outline of this article. In Section \ref{sec:model} we introduce our models, we define all the notions of weak solutions that we derive and we state our main result, the hydrodynamic limit. In Section \ref{sectight} we prove tightness of the sequence of empirical measures. In Section \ref{seccharac} we characterize the limiting point by showing that it is concentrated on a Dirac measure of a trajectory of measures which are absolutely continuous with respect to the Lebesgue measure, whose density is a weak solution of the corresponding hydrodynamic equation. In Section \ref{secenerest} we show some properties of the weak solutions, and Section \ref{secheurwithout} is devoted to the proof of several estimates which are needed from the microscopic system in order to  characterize the solutions. We complement the article with two appendices: Appendix \ref{secuseres} is concerned with all the convergences from the discrete system to the macroscopic quantities, and Appendix \ref{secuniq} is devoted to exploration of tools from analysis to deal with fractional Sobolev spaces and the fractional operators that we deal with, and also with the proof of uniqueness of our weak solutions. 

\section{The model and hydrodynamics}\label{sec:model_hydro}

\subsection{Long-range exclusion with slow bonds in $\mathbb{Z}$}
\label{sec:model}

Before describing our model we describe some sets that are used along the article as $\mathbb{Z}:=\{\ldots, -1, 0, 1, \ldots \}$, $\mathbb{Z}_{-}^{*}:=\{-1,-2,\ldots \}$, $\mathbb{N}:=\{0,1,2,\ldots,\}$,  $\mathbb{R}_{-}^{*}:= (- \infty, 0)$, $\mathbb{R}_{+}^{*}:= (0, \infty)$ , $\mathbb{R}_{+}:=[0, \infty)$ and $\mathbb{R}^{*} = \mathbb{R}_{-}^{*} \cup \mathbb{R}_{+}^{*}=\mathbb{R}\setminus\{0\}$.  Now we describe the dynamics of our  model. All the exchanges of particles  occur at  all bonds $\{x,y \in \mathbb{Z}: x \neq y\}$ and we make the identification of $\{x,y\}$ and $\{y, x\}$. We denote the set of all bonds by $\mcb B$ and the elements of the lattice are called \textit{sites} and are denoted by Latin letters such as $x,y,z$. 

The exclusion process has state space given by  $\Omega:=\{0,1\}^{\mathbb{Z}}$ and the elements of $\Omega$ are  configurations and we denote them  by Greek letters such as $\eta, \xi$. Given  a configuration $\eta$ and a site $x$, we denote the number of particles at $x$ by $\eta(x)$.
Given a bond $\{x,y\}$, a particle can only move between $x$ and $y$ if $\eta(x) \neq \eta(y)$ and in this case, $\eta(x)$ and $\eta(y)$ exchange their values and  produce a new configuration $\eta^{x,y} \in \Omega$ given by:   
\begin{equation*}
\eta^{x,y}(z) = \eta(y)\mathbbm{1}_{z=x}+\eta(x)\mathbbm{1}_{z=y}+ \eta(z)\mathbbm{1}_{z\neq x,y}.
\end{equation*}
The exchange between the sites  $x$ and $y$ occurs with probability $p(x-y)$, where $p: \mathbb{Z} \rightarrow \mathbb{R}$ is given by
\begin{align}\label{eq:trans_prob}
p(z) =\mathbbm{1}_{z \neq 0}c_{\gamma}|z|^{-\gamma-1}  
\end{align}
where $ \gamma \in (0,2)$ and $c_{\gamma}$ is a normalizing constant. Observe that the variance of $p(\cdot)$ is infinite since $\gamma<2$. When $\gamma\in(0,1)$ the first moment  of $p(\cdot)$ is infinite and for $\gamma  \in (1,2)$ we denote it by
$m: = \sum_{z \in \mathbb{N}} z p(z)  < \infty.
$
Now we  consider a set of slow bonds 
\begin{equation*}
    \mcb S \subset \mcb S_0:=\big\{\{x,y\} \in \mcb B: x <0, y \geq 0 \big\}
\end{equation*}
and  the complement of $\mcb S$ with respect to $\mcb B$ will be denoted by $\mcb F$, the set of  fast bonds. Let $n$ be a positive integer and we fix the parameters $\alpha >0$  and $\beta \geq 0$. Given a configuration $\eta$, the rate of the transition at the bond $\{x,y\}$  is denoted by $\xi_{x,y}^n(\eta)$ and it is given by 
\begin{align*}
\xi_{x,y}^n(\eta) =
\begin{cases}
\alpha n^{-\beta}  [\eta(x)(1-\eta(y)) +\eta(y)(1-\eta(x))], \{x,y\} \in \mcb S, \\
 [\eta(x)(1-\eta(y)) + \eta(y)(1-\eta(x))], \{x,y\} \in \mcb F.
 \end{cases}
\end{align*} 
Note that from the previous definition since  $\beta\geq 0$ all the bonds in $\mcb S$ are called \textit{slow bonds}. In the particular case $(\alpha, \beta)=(1,0)$ all the bonds produce the same rate and we say by convention that $\mcb S= \varnothing$.

A function $f: \Omega \rightarrow \mathbb{R}$ is  said to be {local} if there exists a finite set $\Lambda \subset \mathbb{Z}$ such that $f$ is determined by  the value of $ \eta(x)$ for $ x \in \Lambda$.  The infinitesimal generator $\mcb {L}_{n}$ of our process  is defined on local functions $f:\Omega\to\mathbb{R}$ by 
\begin{align*}
\mcb{L}_{n}f(\eta) :=& \frac{1}{2}  \sum_{\{x, y\} \in \mcb B} p(x-y) \xi^{n}_{x,y}(\eta)[f(\eta^{x,y})-f(\eta)].
\end{align*}
Let $a \in (0,1)$. The Bernoulli product measure  in $\Omega$ is denoted by  $\nu_{a}$ and has marginals given by $\nu_{a}\{\eta \in \Omega: \eta(x)=1\} = a, \forall x \in \mathbb{Z}$. Under this measure, the random variables $\{ \eta(x): x \in \mathbb{Z} \}$ are independent and identically distributed with Bernoulli distribution of parameter  $a$. A simple computation, based on the symmetry of $p(\cdot)$, shows that the measure $\nu_a$ is reversible with respect to $\mcb{L}_{n}$.

We will say that we have a \textit{thin barrier} blocking the movement between $\mathbb{Z}_{-}^{*}$ and $\mathbb{N}$ if
\begin{equation} \label{defsigmas}
 \exists \delta \in [0,1] \cap (\gamma-1, \infty): \sum_{ \{x,y  \} \in \mcb S} |y-x|^{\delta} p(y-x) < \infty.
\end{equation}
On the other hand, we will say that we have a \textit{ thick barrier} if $\mcb S=\mcb S_0$.
An example of a model with a thin barrier can be obtained by taking $\mcb S:=\{(x,0): x\in\mathbb Z\}$ and for $\gamma\in(1,2)$ we choose $\delta =1$, for $\gamma=1$ we choose $\delta=\frac{1}{2}$ and   when $\gamma\in (0,1)$ we choose $\delta =0$.

From here on, we fix $T>0$ and a finite time horizon $[0,T]$. Moreover, we denote $\eta_{t}^{n}(x):= \eta_{t n^{\gamma}}(x)$, so that $\eta_t^n$ has  infinitesimal generator $n^{\gamma} \mcb{L}_{n}$. Let  $\mcb D([0,T],\Omega)$ be the space of c\`adl\`ag trajectories (right-continuous and with left limits everywhere) $f:[0,T] \rightarrow \Omega$ endowed  with the Skorohod topology. 

Let $\mcb {M}^+(\mathbb{R})$ be the space of non-negative Radon measures on $\mathbb{R}$ equipped with the weak topology. For   $\eta \in \Omega$, the empirical measure associated to the density is denoted by  $\pi^{n}(\eta,du)$ and it is defined by 
\begin{equation*}\label{MedEmp}
\pi^{n}(\eta, du):=\dfrac{1}{n}\sum _{x }\eta(x)\delta_{\frac{x}{n}}\left( du\right) \in \mcb{M}^+(\mathbb{R}).
 \end{equation*}
Above  $\delta_{b}$ is a Dirac measure on $b \in \mathbb{R}$. For  $G: \mathbb{R} \rightarrow \mathbb{R}$,  $\langle \pi^n, G \rangle$ denotes the integral of $G$ with respect to $\pi^n(\eta,du)$.

Let $\mcb g: \mathbb{R} \rightarrow [0,1]$ be a measurable function. We say that a sequence $(\mu_n)_{n \geq 1}$ of probability measures on $\Omega$  is  {associated to the profile $\mcb g$ if  for every function $G \in C_c^0(\mathbb{R})$ and for every $\delta >0$, it holds
\begin{equation*}
\lim_{n \rightarrow \infty} \mu_n \left( \eta \in \Omega: \Big| \langle \pi^n, G \rangle - \int_{\mathbb{R}} G(u) \mcb g(u) du \Big| > \delta \right) =0. 
\end{equation*}

Above $C^0_c(\mathbb R)$ denotes the space of continuous functions with compact support.  For every $n \geq 1$,  $\mathbb{P} _{\mu_{n}}$ is the probability measure on $\mcb D([0,T],\Omega)$ induced by the Markov process $\{\eta_{t}^{n};{t\geq 0}\}$ and by the initial configuration $\eta_0^n$ with distribution $\mu_{n}$ and $\mathbb{E}_{\mu_{n}}$  is  the expectation with respect to $\mathbb{P}_{\mu_{n}}$. Let  $\pi^{n}_{t}(\eta, dq):=\pi^{n}(\eta^{n}_{t}, dq) $. Let  $\mcb D([0,T], \mcb{M}^+(\mathbb{R}))$ be the space of c\`adl\`ag  trajectories $f:[0,T] \rightarrow \mcb{M}^+(\mathbb{R})$ with the Skorohod topology. In particular, $(\pi_{t}^{n})_{0 \leq t \leq T} \in \mcb D([0,T],\mcb M^+(\mathbb{R}))$. Finally, we define $(\mathbb{Q}_{n})_{n \geq 1}$ as the sequence of probability measures on $\mcb D([0,T],\mcb M^+(\mathbb{R}))$ induced by the Markov process $(\pi_{t}^{n})_{0 \leq t \leq T}$ and by the initial configuration $\eta_0^n$ with distribution $\mu_{n}$.

When we have a thick barrier (i.e. $\mcb S= \mcb S_0$) separating $\mathbb{Z}_{-}^{*}$ and $\mathbb{N}$, at the microscopic level we expect to have boundary conditions that mimic a macroscopic blockage of mass between $\mathbb{R}_{-}^{*}$ and $\mathbb{R}_{+}$. Because of this, in some cases it will be convenient to deal with functions which \textit{may be} discontinuous at the origin but have smooth restrictions in $\mathbb{R}_{-}$ and $\mathbb{R}_{+}$.  Before defining the space of test functions, we present the notation for functions depending only on the space variable. For an interval $I$ in $\mathbb{R}$ and a number $r \in  \mathbb{N}$, we denote by $C^{r} \left(  I \right)$ the set of functions defined on $I$ that are $r$ times differentiable. Moreover, $C^{\infty} (I):= \cap_{r=1}^{\infty} C^{r}(I)$. We also consider the set  $C_{c}^{r} ( \mathbb{R} )$ of functions $G \in C^{r}\left( \mathbb{R} \right)$ such that $G$ has a compact support that may include $0$. Finally, we denote by $C_{c0}^{r} ( \mathbb{R} )$ the subset of $C_{c}^{r} ( \mathbb{R} )$ of functions with a compact support which \textit{does not include}   the origin; more exactly, if $G \in C_{c0}^{r} ( \mathbb{R} )$, there exist $0< \bar{b} < b$ such that $G(u)=0$ if $|u| \leq \bar{b}$ or $|u| \geq b$. We say that $G \in C_c^{\infty}(\mathbb{R}^{*})$ if there exist $G_{-},G_{+} \in C_c^{\infty} (\mathbb{R})$ such that $G(u)=G_{-}(u)\mathbbm{1}_{u <0}+G_{+}(u)\mathbbm{1}_{ u \geq 0}$.

\subsection{Fractional Sobolev spaces and Fractional operators}

Regardless of the measure space $X$, we  denote the Lebesgue measure in $X$ by $\mu$ so that $L^2( X):=L^2 (X, d \mu)$ is the space of functions $f: X \rightarrow \mathbb{R}$ such that $\int_X |f|^2 d \mu < \infty$ and its norm is denoted by $\| \cdot \|_{2,X}$. Also, $L^{\infty}( X):=L^{\infty} (X, d \mu)$ is the space of functions $f: X \rightarrow \mathbb{R}$ with finite essential supremum and its norm is denoted by $\| \cdot \|_{\infty}$. 

Following Section 2 of \cite{hitchhiker}, we define the fractional Sobolev spaces where our solutions belong to. Given an open interval $I$, the Sobolev space $ H^{\frac{\gamma}{2}}(I)$ is the set of functions $f \in L^2(I)$ such that 
\begin{align*}
\iint_{I^2}  \frac{|f(u)-f(v)|^2}{|u-v|^{\gamma+1}}  du dv < \infty. 
\end{align*}
With an abuse of notation, we will say that $f \in \mcb{H}^{\frac{\gamma}{2}} \left( \mathbb{R}^* \right)$ if  $f_{-}:=f|_{ \mathbb{R}_{-}^{*}} \in  \mcb{H}^{\frac{\gamma}{2}} \left( \mathbb{R}_{-}^{*} \right)$ and $f_{+}:= f|_{ \mathbb{R}_{+}^{*}} \in  \mcb{H}^{\frac{\gamma}{2}} \left( \mathbb{R}_+^{*} \right)$.  Now we  recall Theorem 8.2 of \cite{hitchhiker}:
\begin{prop} \label{holderrep}
Let $\gamma \in (1,2)$ and $f \in \mcb H^{\frac{\gamma}{2}}(I)$. Then there exists (exactly) one function $\tilde{f} \in C^{\frac{\gamma-1}{2}}(\bar{I}\,)$ such that $f = \tilde{f}$ almost everywhere in $I$. Above, $\bar{I}$ denotes the closure of $I$.
\end{prop} 
From last proposition and for $\gamma \in (1,2)$, we know that $f_{-}$ and $f_{+}$ have continuous representatives $\tilde{f}_{-}$ and $\tilde{f}_{+}$ in $(-\infty,0]$ and $[0, \infty)$, respectively. If there exists $a \in \mathbb{R}$ such that $g:=f-a \in \mcb{H}^{\frac \gamma 2} \left( \mathbb{R}^* \right)$,  we denote 
\begin{align*}
f(0^{+}) := \lim_{\varepsilon \rightarrow 0^{+}} \frac{1}{\varepsilon} \int_0^{\varepsilon} f(u) du =a+  \lim_{\varepsilon \rightarrow 0^{+}} \frac{1}{\varepsilon} \int_0^{\varepsilon} \tilde{g}_{+}(u) du =a+  \tilde{g}_{+}(0) ; \\
  f(0^{-}) := \lim_{\varepsilon \rightarrow 0^{+}} \frac{1}{\varepsilon} \int_{-\varepsilon}^{0} f(u) du =a+ \lim_{\varepsilon \rightarrow 0^{+}} \frac{1}{\varepsilon} \int_{-\varepsilon}^{0} \tilde{g}_{-}(u) du =a+ \tilde{g}_{-}(0). 
\end{align*}
Above, $\tilde{g}_{-}$ and $\tilde{g}_{+}$ are the continuous representatives of $g_{-}:=g|_{ \mathbb{R}_{-}^{*}}$ and $g_{+}:= g|_{ \mathbb{R}_{+}^{*}}$ in $(-\infty,0]$ and $[0, \infty)$, respectively.
\begin{rem} \label{Hcont}
 Observe that $\mcb{H}^{\frac{\gamma}{2}}(\mathbb{R}) \subset \mcb{H}^{\frac{\gamma}{2}} \left( \mathbb{R}^* \right)$ since $f|_{ \mathbb{R}_{-}^{*}} \in  \mcb{H}^{\frac{\gamma}{2}} \left( \mathbb{R}_{-}^{*} \right)$ and $f|_{ \mathbb{R}_{+}^{*}} \in  \mcb{H}^{\frac{\gamma}{2}} \left( \mathbb{R}_+^{*} \right)$ for every $f \in \mcb{H}^{\frac{\gamma}{2}}(\mathbb{R})$.
\end{rem}
We note that when $f \in \mcb{H}^{\frac{\gamma}{2}}(\mathbb{R})$, since it has a  unique representative which is continuous, then $\tilde{f}$ coincides with $\tilde{f}_{-}$ and $\tilde{f}_{+}$ on $(-\infty,0]$ and $[0, \infty)$, respectively, and 
$
f(0^{+}) = f_{+}(0) = f(0) = f_{-}(0) =  f(0^{-}).
$
But, when $f \in \mcb{H}^{\frac{\gamma}{2}} \left( \mathbb{R}^* \right)$, we do not have a global representative in $\mathbb{R}$, and therefore it is not guaranteed that $f(0^{+})=f(0^{-})$.

\subsubsection{Introducing the time dependence}

Similarly  to Definition 23.1 of \cite{MR1033497}, we say that $\varrho: [0,T ] \rightarrow L^2(I)$ belongs to   $L^2 \big(0, T ; \mcb{H}^{\frac{\gamma}{2}} (I ) \big)$ if $\varrho(t,\cdot) \in \mcb{H}^{\frac{\gamma}{2}} ( I )$ for a.e. $t\in[0,T]$ and
\begin{align*}
\|\varrho\|^2_{L^2 \big(0, T ; \mcb{H}^{\frac{\gamma}{2}} ( I ) \big)} :=  \int_0^T \|\varrho(t, \cdot) \|^2_{\mcb{H}^{\frac{\gamma}{2}} \left( I \right)} dt < \infty.
\end{align*}
As above, given  $\varrho \in L^2 \big(0, T ; \mcb{H}^{\frac{\gamma}{2}} ( I) \big)$, we say that $\tilde{\varrho}$ is the continuous representative of $\varrho$ if $\tilde{\varrho}(t,\cdot) \in \mcb{H}^{\frac{\gamma}{2}} ( I )$ and $\tilde{\varrho}(t,\cdot) \in C^0( \bar{I} )$ for a.e. $t\in[0,T]$.  This means that  $\varrho = \tilde{\varrho}$ almost everywhere on $[0,T] \times \mathbb{R}$. Moreover, $\tilde{\varrho}$ is unique. We say that $\varrho \in L^2 \left(0, T ; \mcb{H}^{\frac{\gamma}{2}} \left( \mathbb{R}^{*} \right) \right) $ if $\varrho(t,\cdot) \in \mcb{H}^{\frac{\gamma}{2}} (\mathbb{R}^{*} )$ for a.e. $t\in[0,T]$ and we note that from  Remark \ref{Hcont}, it holds  $L^2 \left(0, T ; \mcb{H}^{\frac{\gamma}{2}} \left( \mathbb{R} \right) \right) \subset L^2 \left(0, T ; \mcb{H}^{\frac{\gamma}{2}} \left( \mathbb{R}^{*} \right) \right)$.

Now we define our space of  test functions. As in  \cite{MR1033497}, given a metric space $(N, \| \cdot \|_{N} )$, we say that $P ( [0,T],  N )$ is the space of all polynomials $G: [0,T] \rightarrow N$, i.e., there exists $k \in \mathbb{N}$ such that $G(t)=a_0 + a_1 t + \ldots + a_k t^k$, with $a_j \in N, \forall j=0,1,\ldots,k$, $\forall t \in [0,T]$.
We say that $G \in P \big( [0,T],  C_c^{\infty}(\mathbb{R}^{*}) \big)$ if there exists $k \in \mathbb{N}$ such that $G(t)=a_0 + a_1 t + \ldots + a_k t^k$, with $a_j \in C_c^{\infty}(\mathbb{R}^{*}), \forall j=0,1,\ldots,k$, $\forall t \in [0,T]$. 
 Note  that for every $G \in P \big( [0,T],  C_c^{\infty}(\mathbb{R}^{*}) \big)$, there exist $G_{-},G_{+} \in P \big( [0,T],  C_c^{\infty}(\mathbb{R}) \big)$ such that
$
G (t,u) =\mathbbm{1}_{u\in(-\infty,0)}G_{-} (t,u)+\mathbbm{1}_{u\in [0, \infty)} G_{+} (t,u).
$
It is simple to check  that $ P \big( [0,T],  C_c^{\infty}(\mathbb{R}) \big) \subset P \big( [0,T],  C_c^{\infty}(\mathbb{R}^*) \big)$.
Moreover, given $G \in  P \big( [0,T],  C_c^{\infty}(\mathbb{R} ^{*}) \big)$, there exists $b>0: G(s,u)=0$ when $|u|  \geq b$, for every $s \in [0,T]$. We denote
\begin{align}\label{eq:B_g}
b_G:= \min \big \{b \in \mathbb{N}: G(s,u)=0 \quad\textrm{for}\quad (s,u) \in [0,T] \times \big(  ( - \infty, -b ] \cup [b, \infty) \big) \big \}.
\end{align}
Moreover, if $G \in  P \big( [0,T],  C_{c0}^{\infty}(\mathbb{R}) \big)$, there exists $\bar{b}>0: G(s,u)=0$ when $|u|  \leq \bar{b}$, for every $s \in [0,T]$. We denote
\begin{align}\label{eq:barB_g}
\bar{b}_G= \sup \big \{b \in [ 0, b_G] : G(s,u)=0 \; \text{for} \; (s,u) \in [0,T] \times [-b,b] \big \}.
\end{align}
Finally we will introduce some fractional operators that will be relevant in the statement of the hydrodynamic equations. We begin with the fractional Laplacian.  

\subsubsection{Fractional Laplacian}

Let $G: \mathbb{R} \rightarrow \mathbb{R}$. For every $\varepsilon >0$, we define
\begin{equation} \label{limlapfrac}
[-(- \Delta)^{\frac{\gamma}{2}} G ]_{\varepsilon} (u ) := c_{\gamma}   \int_{- \infty}^{\infty} \mathbbm{1}_{\{ |u - v| \geq \varepsilon \}} \frac{G(v) - G(u)}{ |u-v|^{\gamma+1}} dv, \forall u \in \mathbb{R}. 
\end{equation}
We define the fractional Laplacian $-(- \Delta)^{\frac{\gamma}{2}}$ of exponent $\frac{\gamma}{2}$ on the set of functions $G: \mathbb{R} \rightarrow \mathbb{R}$ such that $$\int_{\mathbb{R}} \frac{|G(u)|}{\left( 1 + |u| \right)^{\gamma+1}} du < \infty$$ by
\begin{equation} \label{deflapfrac}
[-(- \Delta)^{\frac{\gamma}{2}} G](u) =  \lim_{\varepsilon \rightarrow 0^+}  [-(- \Delta)^{\frac{\gamma}{2}} G ]_{\varepsilon} (u ), \forall u \in \mathbb{R},
\end{equation}
provided the limit exists. From Proposition \ref{lapfracglob} we  know  that the fractional Laplacian is well-defined for functions $G \in C_c^2 (\mathbb{R})$.
Now we will introduce the regional fractional Laplacian defined on open intervals of the real line.

\subsubsection{Regional fractional Laplacian}

Let $I$ be an open interval of $\bb R$ and let  $G: \mathbb{R} \rightarrow \mathbb{R}$. For every $\varepsilon >0$, we define $[-(- \Delta)_I^{\frac{\gamma}{2}} G ]_{\varepsilon}: I \rightarrow \mathbb{R}$, on $u\in I$, by
\begin{equation} \label{limlapfracreg2}
[-(- \Delta)_I^{\frac{\gamma}{2}} G ]_{\varepsilon} (u ) := c_{\gamma}   \int_{I} \mathbbm{1}_{\{ |u - v| \geq \varepsilon \}} \frac{G(v) - G(u)}{ |u-v|^{\gamma+1}} dv. 
\end{equation} 
The regional fractional Laplacian $[-(- \Delta)_I^{\frac{\gamma}{2}} G ]: I \rightarrow \mathbb{R}$ of exponent $\frac{\gamma}{2}$ is defined on the set of functions $G: \mathbb{R} \rightarrow \mathbb{R}$ such that $$\int_{I} \frac{|G(u)|}{\left( 1 + |u| \right)^{\gamma+1}} du < \infty,$$ and on $u \in I,$ by
\begin{equation} \label{deflapfracreg}
[-(- \Delta)_{I}^{\frac{\gamma}{2}} G](u) =  \lim_{\varepsilon \rightarrow 0^+}  [-(- \Delta)_{I}^{\frac{\gamma}{2}} G ]_{\varepsilon} (u )
\end{equation}
provided the limit exists. We will be interested when $I=\mathbb{R}_{-}^{*}$ or $I=\mathbb{R}_+^{*}$. In these cases, from Proposition \ref{lapfracreg} we know  that the regional fractional Laplacian on $I$ is well-defined on functions $G \in C_c^2 (\mathbb{R}^{*})$ when $\gamma \in  (1,2)$ and it is well-defined on functions $G \in C_{c0}^2 (\mathbb{R})$ when $\gamma \in (0,1]$. We observe that $-(- \Delta)_I^{\frac{\gamma}{2}} G$ does not depend on the values of $G$ in $\mathbb{R} - I$ and for $I=\mathbb{R}$, the definitions of $[-(- \Delta)_I^{\frac{\gamma}{2}} G ]_{\varepsilon}$ and $[-(- \Delta)_I^{\frac{\gamma}{2}} G ]$ coincide with the definitions of $[-(- \Delta)^{\frac{\gamma}{2}} G ]_{\varepsilon}$ and $[-(- \Delta)^{\frac{\gamma}{2}} G ]$, respectively.   
If $G$ is such that $[-(- \Delta)_{\mathbb{R}_{-}^{*}}^{\frac{\gamma}{2}} G]$ and $[-(- \Delta)_{\mathbb{R}_{+}^{*}}^{\frac{\gamma}{2}} G]$ are well-defined, with an abuse of notation, we define the regional fractional Laplacian on $\mathbb{R}^{*}$ $-(- \Delta)_{\mathbb{R}^{*}}^{\frac{\gamma}{2}} G: \mathbb{R} \rightarrow \mathbb{R}$ by
\begin{equation} \label{lapfracreg3}
[-(- \Delta)_{\mathbb{R}^{*}}^{\frac{\gamma}{2}} G](x)=\mathbbm{1}_{x>0 } [-(- \Delta)_{\mathbb{R}_{+}^{*}}^{\frac{\gamma}{2}} G](x)+\mathbbm{1}_{x<0 } [-(- \Delta)_{\mathbb{R}_{-}^{*}}^{\frac{\gamma}{2}} G](x) 
\end{equation}
Inspired from \cite{Guan2006}, we define   the fractional derivatives
\begin{align*}
D^{\gamma} f(0^+) :=-  \lim_{u \rightarrow 0^{+}} f'(u) u^{2-\gamma} \quad \textrm{and}\quad   D^{\gamma} f(0^-) :=- \lim_{u \rightarrow 0^{-}} f'(-u) u^{2-\gamma},
\end{align*}
when the limits exist. We observe that $D^{\gamma} f(0^+)=D^{\gamma} f(0^-)=0$ when $f \in C_{c}^1(\mathbb{R}^*)$.


\subsection{Hydrodynamic equations}
Now we define the notions of weak solution of the hydrodynamic equations that we obtain. 
\begin{definition}\label{eq:dif}
Let $\mcb g: \mathbb{R} \rightarrow [0,1]$ be a measurable function. We say that $\varrho :[0,T] \times \mathbb{R} \rightarrow [0,1]$ is a weak solution of the fractional diffusion equation in $\mathbb{R}$ with initial condition $\mcb g$ 
\begin{equation} \label{eqhydfracdifreal}
\begin{cases}
\partial_t \rho(t,u) = \big[-(- \Delta)^{\frac{\gamma}{2}} \rho \big] (t,u), (t,u) \in [0,T] \times \mathbb{R}, \\
\varrho(0,u) =\mcb g(u), u \in \mathbb{R} 
\end{cases}
\end{equation}
if the following two conditions hold:
\begin{enumerate}
\item
for every $t \in [0,T]$ and for every $G \in \mcb S_{\textrm {Dif}}:=P \big( [0,T] , C_c^{\infty}(\mathbb{R}) \big)$, it holds  $F_{\textrm{FrDif}}(t, \varrho,G, \mcb g)=0$, where
\begin{align*}
F_{\textrm{FrDif}}(t, \varrho,G, \mcb g):= & \int_{\mathbb{R}} \varrho(t,u) G(t,u) du - \int_{\mathbb{R}} \mcb g(u) G(0,u) du 
-  \int_0^t \int_{\mathbb{R}} \varrho(s,u) \big[ [-(- \Delta)^{\frac{\gamma}{2}}  ] + \partial_s \big] G(s,u) du ds; 
\end{align*}
\item 
there exists $a \in (0,1)$ such that $\bar{ \varrho} \in L^2 \left(0, T ; \mcb{H}^{\frac{\gamma}{2}} ( \mathbb{R} ) \right)$, where $\bar{ \varrho}:=\varrho-a$.
\end{enumerate}
\end{definition}
\begin{definition}
Let  $\kappa \geq 0$ and $\mcb g: \mathbb{R} \rightarrow [0,1]$ a measurable function. We say that $\varrho:[0,T] \times \mathbb{R} \rightarrow [0,1]$ is a weak solution of the fractional diffusion equation in $\mathbb{R}^{*}$ with fractional Robin boundary conditions and initial condition $\mcb g$ 
\begin{equation} \label{eqhydfracdifrob}
\begin{cases}
\partial_t \rho(t,u) = \big[-(- \Delta)_{\mathbb{R}^{*}}^{\frac{\gamma}{2}} \rho \big] (t,u), (t,u) \in [0,T] \times \mathbb{R}^{*}, \\
D^{\gamma} \rho(t,0^+) = D^{\gamma} \rho(t,0^-)= \frac{\kappa}{C_{\gamma}}  [ \rho(t,0^+) - \rho(t,0^{-})], t \in [0,T], \\
\varrho(0,u) =\mcb  g(u), u \in \mathbb{R} 
\end{cases}
\end{equation}
if the following two conditions hold:
\begin{enumerate}
\item
for every $t \in [0,T]$ and for every $G \in \mcb S_{\gamma}$, it holds $F_{\textrm{FrRob}}(t, \varrho,G,\mcb g,\kappa)=0$, where 
\begin{align*}
F_{\textrm{FrRob}}(t, \varrho,G, \mcb g, \kappa):= & \int_{\mathbb{R}} \varrho(t,u) G(t,u) du - \int_{\mathbb{R}} \mcb g(u) G(0,u) du \\-&  \int_0^t \int_{\mathbb{R}} \varrho(s,u) \big[ [-(- \Delta)_{\mathbb{R}^{*}}^{\frac{\gamma}{2}}  ] + \partial_s \big] G(s,u) du ds  \\
+ & \kappa \mathbbm{1}_{\gamma \in (1,2)} \int_0^t  [   \varrho(s,0^{+})  - \varrho(s,0^{-}) ]  [  G(s,0^{+})    -  G(s,0^{-})  ] ds;  
\end{align*}
\item 
there exists $a \in (0,1)$ such that $\bar{ \varrho} \in L^2 \left(0, T ; \mcb{H}^{\frac{\gamma}{2}} ( \mathbb{R}^* ) \right)$, where $\bar{ \varrho}:=\varrho-a$.
\end{enumerate}
Above, $C_{\gamma}$ is the constant produced taking $\alpha=\gamma$ in equation (3.8) of \cite{Guan2006}. Moreover, $\mcb S_{\gamma} = \mcb S_{Rob}:= P \big( [0,T] , C_{c}^{\infty}(\mathbb{R}^{*}) \big)$ if $\gamma \in (1,2)$ and $\mcb S_{\gamma} = \mcb S_{Neu}:= P \big( [0,T] , C_{c0}^{\infty}(\mathbb{R}) \big)$ if $\gamma \in (0,1]$.
\end{definition}

\begin{rem}
Taking  $\kappa=0$ in last definition, we denote $F_{\textrm{FrNeu}}(t, \varrho,G, \mcb g)=F_{\textrm{FrRob}}(t, \varrho,G, \mcb g,0)$ and we say that $\varrho$ is a weak solution to the fractional diffusion equation with fractional Neumann boundary conditions. 
\end{rem}

\begin{definition}
Let $\kappa \in (0,1) \cup (1, \infty)$ and $\mcb g: \mathbb{R} \rightarrow [0,1]$  a measurable function. We say that $\varrho:[0,T] \times \mathbb{R} \rightarrow [0,1]$ is a weak solution of the fractional diffusion equation in $\mathbb{R}^{*}$ with initial condition $\mcb g$
\begin{equation} \label{eqhydfracbetazero}
\begin{cases}
\partial_t \rho(t,u) = \big[ \big( \kappa [-(- \Delta)^{\frac{\gamma}{2}}]  + (1 - \kappa) [-(- \Delta)_{\mathbb{R}^{*}}^{\frac{\gamma}{2}}] \big) \rho  \big] (t,u), (t,u) \in [0,T] \times \mathbb{R}^{*}, \\
D^{\gamma} \rho(t,0^+) = D^{\gamma} \rho(t,0^-), t \in [0,T], \\
\varrho(0,u) =\mcb  g(u), u \in \mathbb{R}, 
\end{cases}
\end{equation}
if the following two conditions hold:
\begin{enumerate}
\item
for every $t \in [0,T]$ and for every $G \in \mcb S_{\gamma}$, it holds $F_{FrDif2}(t, \varrho,G,\mcb g, \kappa)=0$, where 
\begin{align*}
F_{FrDif2}(t, \varrho,G, \mcb g, \kappa):= & \int_{\mathbb{R}} \varrho(t,u) G(t,u) du - \int_{\mathbb{R}} \mcb g(u) G(0,u) du \\
-&  \int_0^t \int_{\mathbb{R}} \varrho(s,u) \big[  \kappa [-(- \Delta)^{\frac{\gamma}{2}}  ] +  (1-\kappa) [-(- \Delta)_{\mathbb{R}^{*}}^{\frac{\gamma}{2}}  ] + \partial_s \big] G(s,u) du ds;  
\end{align*}
\item 
there exists $a \in (0,1)$ such that $\bar{ \varrho} \in L^2 \left(0, T ; \mcb{H}^{\frac{\gamma}{2}} ( \mathbb{R}) \right)$, where $\bar{ \varrho}:=\varrho-a$.
\end{enumerate}
Above, $\mcb S_{\gamma} = \mcb S_{Dif}$ if $\gamma \in (1,2)$, and $\mcb S_{\gamma} = \mcb S_{Neu}$ if $\gamma \in (0,1]$.
\end{definition}

\begin{definition}
Let  $\mcb g: \mathbb{R} \rightarrow [0,1]$ be a measurable function. We say that $\varrho:[0,T] \times \mathbb{R} \rightarrow [0,1]$ is a weak solution of the fractional diffusion equation in $\mathbb{R}^{*}$ with initial condition $\mcb g$ 
\begin{equation} \label{eqhydfracdifnotuniq}
\begin{cases}
\partial_t \rho(t,u) = \big[-(- \Delta)_{\mathbb{R}^{*}}^{\frac{\gamma}{2}} \rho \big] (t,u), (t,u) \in [0,T] \times \mathbb{R}^{*}, \\
D^{\gamma} \rho(t,0^+) = D^{\gamma} \rho(t,0^-) , t \in [0,T], \\
\varrho(0,u) =\mcb  g(u), u \in \mathbb{R} 
\end{cases}
\end{equation}
if the following conditions hold:
\begin{enumerate}
\item
$F_{\textrm{FrRob}}(t, \varrho,G,\mcb g, 0)=0$, for every $t \in [0,T]$, for every $G \in \mcb S_{\textrm {Rob}0}$ where 
\begin{equation}\label{SROB0}
\mcb S_{\textrm {Rob}0}:=\{G \in \mcb S_{\textrm {Rob}}: G_{-}(0)=G_{+}(0) \};
\end{equation}
\item 
there exists $a \in (0,1)$ such that $\bar{ \varrho} \in L^2 \big(0, T ; \mcb{H}^{\frac{\gamma}{2}} ( \mathbb{R}^* ) \big)$,  where $\bar{ \varrho}:=\varrho-a$;
\item
 $\rho(s,0^{+})=\rho(s,0^{-})$, for a.e. $s\in [0,T]$;  
\end{enumerate}
\end{definition}
\begin{rem}\label{rem:cons}
We prove in Proposition \ref{4condition} that from $(2)$ and $(3)$  we get that
\begin{enumerate}
\item  [a)] $\rho$ has a continuous representative $\tilde{\rho}$ such that  $\tilde{\rho}(s, \cdot)$ is $\frac{\gamma-1}{2}$-H\"older continuous, for a.e. $s\in [0,T]$; 
\end{enumerate} 
and in Proposition \ref{5condition} we prove  that from $(2)$ and $a)$ we get that
\begin{enumerate}
\item  [b)] there exists $a \in (0,1)$ such that $\bar{ \varrho} \in L^2 \big(0, T ; \mcb{H}^{\frac{\gamma}{2}-\delta} ( \mathbb{R} ) \big), \forall \delta \in (0, \frac{\gamma-1}{2})$,  where $\bar{ \varrho}:=\varrho-a$.
\end{enumerate} 
 We observe that from $b)$ there exists $(G_k)_{k \geq 1} \subset  \mcb S_{\textrm {Dif}} \subsetneq \mcb S_{\textrm {Rob}0} $ such that for every $\delta \in (0, \frac{\gamma-1}{2})$, it holds
\begin{align*}
lim_{k \rightarrow \infty} \|G_k - \bar{\rho}\|_{L^2 \big(0,T; \mcb{H}^{\frac{\gamma}{2}-\delta} ( \mathbb{R}_{-}^{*} ) \big) } = lim_{k \rightarrow \infty} \|G_k - \bar{\rho}\|_{L^2 \big(0,T; \mcb{H}^{\frac{\gamma}{2}-\delta} ( \mathbb{R}_{+}^{*} ) \big) }=0.
\end{align*}
Now we will state a similar (and very desirable) condition that, unfortunately, we were not able to prove, but that would be enough to prove the uniqueness of weak solutions of \eqref{eqhydfracdifnotuniq}:
\begin{equation} \label{conduniq}
\exists (G_k)_{k \geq 1} \subset \mcb S_{\textrm {Rob}0}: lim_{k \rightarrow \infty} \|G_k - \bar{\rho}\|_{L^2 \big(0,T; \mcb{H}^{\frac{\gamma}{2}} ( \mathbb{R}_{-}^{*} ) \big) } = lim_{k \rightarrow \infty} \|G_k - \bar{\rho}\|_{L^2 \big(0,T; \mcb{H}^{\frac{\gamma}{2}} ( \mathbb{R}_{+}^{*} ) \big) }=0.
\end{equation}
\end{rem}
The uniqueness of weak solutions of \eqref{eqhydfracdifreal}, \eqref{eqhydfracdifrob} and \eqref{eqhydfracbetazero} is proved in Appendix \ref{secuniq}, where we also prove the uniqueness of weak solutions of \eqref{eqhydfracdifnotuniq} assuming \eqref{conduniq}.
\subsection{The main result}

Now we will enunciate the hydrodynamic limit of our model. Hereinafter we write $f(n) \lesssim g(n)$ if there exists a constant $C$ independent of $n$ such that $f(n) \leq C g(n)$ for every $n \geq 1$.   
 \begin{thm} \textbf{(Hydrodynamic Limit)} \label{hydlim}
Let $\mcb g: \mathbb{R} \rightarrow [0,1]$ be a measurable function. Let $(\mu_n)_{n \geq 1}$ be a sequence of probability measures in $\Omega$ associated to the profile $\mcb g$ such that $H( \mu_n | \nu_a) \lesssim  n$, for some $a \in (0,1)$. Then, for any $0 \leq t \leq T$, any $G \in C_c^0(\mathbb{R})$ and any $\delta >0$,
\begin{align*}
\lim_{n \rightarrow \infty} \mathbb{P}_{\mu_{n}} \Big( \eta^n_\cdot \in  \mcb D([0,T],\Omega) : \Big| \langle \pi^n_t, G \rangle - \int_{\mathbb{R}} G(u) \varrho(t,u) du \Big| > \delta \Big) =0, 
\end{align*}
where $\varrho$ is the unique weak solution of
\begin{align*}
\begin{cases}
\eqref{eqhydfracdifreal},  &\quad  \text{if} \;  \eqref{defsigmas}  \; \text{holds}; \\
\eqref{eqhydfracbetazero}, \quad \textrm{with} \quad  \kappa= \alpha &\quad  \text{if} \;  \mcb S = \mcb S_{0} \; \text{and} \;  \beta =0; \\ 
\eqref{eqhydfracdifrob}, \quad \textrm{with} \quad  \kappa= 0 &\quad  \text{if} \; \mcb S = \mcb S_{0}, \;  \gamma \in (0,1] \; \text{and} \;  \beta >0; \\
\eqref{eqhydfracdifrob}, \quad \textrm{with} \quad  \kappa= \alpha m & \quad  \text{if} \; \mcb S = \mcb S_{0}, \;  \gamma \in (1,2) \; \text{and} \;  \beta = \gamma-1; \\
\eqref{eqhydfracdifrob}, \quad \textrm{with} \quad  \kappa= 0 & \quad  \text{if} \; \mcb S = \mcb S_{0}, \;   \gamma \in (1,2) \; \text{and} \;  \beta > \gamma-1.
\end{cases}
\end{align*}
In the regime  $\mcb S = \mcb S_{0}$,  $\gamma \in (1,2)$ and $\beta \in (0, \gamma-1)$, we have that $(\mathbb{Q}_n)_{n \geq 1}$ is tight and all limit points are concentrated on trajectories of the form $\pi_t(du)=\varrho_t(u)du$, where $\varrho$ is  a weak solution of \eqref{eqhydfracdifnotuniq}.
\end{thm}

We observe that we have a \textit{static} behavior when \eqref{defsigmas} holds, since in this case we do not have a sufficiently large number of slow bonds in order to block the passage of particles between $\mathbb{Z}_{-}^{*}$ and $\mathbb{N}$. On the other hand, when $ \mcb S  = \mcb S_0$ we  have a \textit{phase transition} depending on the values of $\beta$ and $\gamma$, similarly to Theorem 4.1 of \cite{stefano}.
Indeed, if $\gamma \in (1,2)$ there are three possibilities: if $\beta \in (0, \gamma-1)$ the profile is almost everywhere  H\"older-continuous; 
 for $\beta > \gamma- 1$ there is no transport of mass between $\mathbb{R}_{-}^{*}$ and $\mathbb{R}_{+}$; and finally, in the critical case $\beta= \gamma- 1$, we have  fractional Robin boundary conditions that depend on the value of $\alpha$. On the other hand, for $\gamma \in (0,1]$ and $\beta >0$, the profile can be approximated by test functions in $\mcb S_{Neu}$ and, macroscopically, we do not have mass flowing through the origin. Most remarkably, even when $\beta=0$, we have a particular phase (both for $\gamma \in (0,1]$ and $\gamma \in (1,2)$), which is quite different from the diffusive behavior (i.e. when $\gamma>2$). Heuristically, this happens because for $\mcb S = \mcb S_0$, the flow of mass through the slow bonds is significantly faster when $\gamma < 2$ compared to $\gamma >2$. Macroscopically, this different behavior can be justified by the presence of the fractional Laplacian, which is a non-local operator.

The rest of this article is devoted to the proof of last theorem. From here on we fix $a\in(0,1)$ such that $H( \mu_n | \nu_a) \lesssim  n$. In Section \ref{sectight}, we prove that the sequence $(\mathbb{Q}_n)_{n \geq 1}$ is tight with respect to the Skorohod topology of $\mcb D \big( [0,T], \mcb{M}^+(\mathbb{R}) \big)$ and therefore it has at least one limiting point $\mathbb{Q}$. We prove that any limit point $\mathbb{Q}$ is concentrated on trajectories that satisfy the first (resp. second) condition of weak solutions of the corresponding hydrodynamic equations from the results of Section  \ref{seccharac} and Section \ref{secheurwithout} (resp. Section \ref{secenerest}). The necessary replacement lemmas are proved in Section \ref{secheurwithout} and the uniqueness of the hydrodynamic equations is explained in Appendix \ref{secuniq}. Finally, we present some auxiliary results in Appendix \ref{secuseres}. The uniqueness of the weak solutions of the corresponding hydrodynamic equation implies in the uniqueness of the limit point $\mathbb{Q}$ and in the convergence of the sequence $(\mathbb{Q}_n)_{n \geq 1}$, leading to the desired result.
\begin{rem}
We comment a bit on statement of the previous theorem in  the case $\gamma \in (1,2)$ and $\beta \in (0, \gamma-1)$.  Since we were only able to prove the uniqueness of weak solutions of \eqref{eqhydfracdifnotuniq}  assuming \eqref{conduniq}, we could not obtain the convergence as in the other regimes.  In case uniqueness can be proved  without assuming \eqref{conduniq}, then  convergence would follow. 
\end{rem}
\section{Tightness} \label{sectight}

This section is devoted to the proof of tightness of the sequence $(\mathbb{Q}_{n})_{ n \geq 1 }$. To that end we invoke  Proposition 4.1.6  of \cite{kipnis1998scaling}, which tells us that it is enough to show that, for every $\varepsilon >0$,
\begin{equation} \label{T1sdif}
\displaystyle \lim _{\delta \rightarrow 0^+} \limsup_{n \rightarrow\infty} \sup_{\tau  \in \mathcal{T}_{T},\bar\tau \leq \delta} {\mathbb{P}}_{\mu _{n}}\Big[\eta_{\cdot}^{n}\in \mcb D ( [0,T], \Omega) :\left\vert \langle\pi^{n}_{\tau+ \bar\tau},G\rangle-\langle\pi^{n}_{\tau},G\rangle\right\vert > \ve \Big]  =0, 
\end{equation}
for any function $G\in C_c^2(\mathbb{R})$. In last display $\mathcal{T}_T$ represents  the set of stopping times bounded by $T$, which means $\tau+\bar{\tau}$ should be read as $\min \{ \tau + \bar{\tau}, T \}$. Actually, we claim that in order to conclude tightness, it is enough to show
last result but for $G \in C_{c0}^2(\mathbb{R})$.
Indeed, given $G \in C_{c}^2(\mathbb{R})$, there exists $(G_k)_{k \geq 1} \subset  C_{c0}^2(\mathbb{R})$ such that $\lim_{k \rightarrow \infty} \|G-G_k \|_{1, \mathbb{R}}=0$ and since the probability in \eqref{T1sdif} is bounded by
\begin{align*}
&  {\mathbb{P}}_{\mu _{n}}\Big[\eta_{\cdot}^{n}\in \mcb D ( [0,T], \Omega) :\left\vert \langle\pi^{n}_{\tau+ \bar\tau},G_k\rangle-\langle\pi^{n}_{\tau},G_k\rangle\right\vert > \frac{\ve}{2}  \Big] \\
+ & {\mathbb{P}}_{\mu _{n}}\Big[\eta_{\cdot}^{n}\in \mcb D ( [0,T], \Omega) :\left\vert \langle\pi^{n}_{\tau+ \bar\tau},G - G_k \rangle-\langle\pi^{n}_{\tau}, G - G_k \rangle\right\vert > \frac{\ve}{2}  \Big],
\end{align*}  
we can conclude the claim.

 From Dynkin's formula, see, for example  Appendix 1 of \cite{kipnis1998scaling},  for every $t \geq 0$ we have that
\begin{equation} \label{defMnt}
\mcb M_{t}^{n}(G) = \langle \pi_{t}^{n},G(t, \cdot) \rangle - \langle \pi_{0}^{n},G (0, \cdot) \rangle  - \int_0^t ( n^{\gamma} \mcb {L}_n + \partial_s) \langle  \pi_{s}^{n},G (s, \cdot) \rangle ds,
\end{equation}
is a martingale for every $G$ sufficiently smooth.  We observe that \eqref{T1sdif} is a direct consequence of  the next result combined with  Markov's inequality.

\begin{prop} \label{tightcond1dif}
For $G \in C_{c0}^2(\mathbb{R})$, it holds
\begin{align*}
\lim_{\delta \rightarrow 0^+} \limsup_{n \rightarrow \infty} \sup_{\tau \in \mathcal{T}_T, \bar{\tau} \leq \delta} \mathbb{E}_{\mu_n} \left[ \Big| \int_{\tau}^{\tau+ \bar{\tau}} n^{\gamma}  \mcb L_{n}\langle \pi_{s}^{n},G\rangle ds \Big| \right] = 0.
\end{align*}
\end{prop}
\begin{prop} \label{tightcond2dif}
Let $G \in \mcb S_{Rob}$ if $\mcb S = \mcb S_0$, $\gamma \in (1,2)$ and $\beta \in [\gamma-1,\infty)$, and let $G \in \mcb S_{Dif}$ otherwise. We have
\begin{align*}
\lim_{\delta \rightarrow 0^+} \limsup_{n \rightarrow \infty} \sup_{\tau \in \mathcal{T}_T, \bar{\tau} \leq \delta} \mathbb{E}_{\mu_n} \left[ \left( \mcb M_{\tau}^{n}(G) -  \mcb M_{\tau+\bar{\tau}}^{n}(G) \right)^2 \right] = 0.
\end{align*}
\end{prop}
\begin{rem}
The statement of the previous result also includes functions $G$ which are time dependent and may be discontinuous at the origin. This general result  is not necessary here but it will be useful ahead. 
\end{rem}
Before presenting the proof of last results we first obtain the action of the generator on the empirical measure and we state it as a proposition since it will be used several times along the article.  Its proof is a simple  a long computation that is left to the reader. 
\begin{prop} \label{gensdif}
For any $G$, it holds   
\begin{align}
  \int_{0}^{t} n^{\gamma} \mcb L_{n}\langle \pi_{s}^{n},G(s, \cdot) \rangle ds   =&  \int_{0}^{t} \frac{1}{n} \sum _{x } n^{\gamma}  \mcb {K}_n G \left(s, \tfrac{x}{n} \right)  \eta_{s}^{n}(x) ds \label{princsdif}   \\
+&  \int_{0}^{t} \frac{n^{\gamma-1}}{2} \Big( 1 - \frac{\alpha}{n^{\beta}} \Big)  \sum_{\{x,y \} \in \mcb S }  [G(s,\tfrac{y}{n}) - G(s,\tfrac{x}{n})] p(y-x)  [\eta_s^n(y)-\eta_s^n(x)] ds \label{extrasdif} \\
=&  \int_{0}^{t}\sum_{ \{x,y\} \in \mcb F} n^{\gamma-1} [G(s, \tfrac{y}{n}) - G(s, \tfrac{x}{n}) ]  p(y-x)  \eta_s^n(x) ds \label{sneuterm} \\
+&  \int_{0}^{t} \alpha n^{\gamma-1-\beta}  \sum_{ \{x,y\} \in \mcb S}  [G(s, \tfrac{y}{n}) - G(s, \tfrac{x}{n}) ]  p(y-x)  \eta_s^n(x) ds, \label{srobterm}
\end{align}
where  
\begin{equation}\label{op_Kn}
\mcb{K}_n G \left( \tfrac{x}{n} \right): = \sum_{y} \left[ G( \tfrac{y}{n}) -G( \tfrac{x}{n}) \right] p(y-x) = \sum_{z } \left[ G( \tfrac{z+x}{n}) -G( \tfrac{x}{n})\right] p(z) .
\end{equation}
\end{prop}
Last identity has been written  in two different ways since we  will use \eqref{princsdif} and \eqref{extrasdif} when \eqref{defsigmas} holds; and \eqref{srobterm} and \eqref{sneuterm} otherwise.

\begin{proof}[Proof of Proposition \ref{tightcond1dif}]

The proof follows if we show that,   $\sup_{s \in [0,T]} |n^{\gamma}  \mcb L_{n}\langle \pi_{s}^{n},G\rangle| \lesssim 1.$

\textbf{I.)} If \eqref{defsigmas} holds, since $|\eta_{s}^{n}(x)| \leq 1$, then from   \eqref{princsdif} and \eqref{extrasdif}, we  have that
\begin{align*}
 |n^{\gamma}  \mcb L_{n}\langle \pi_{s}^{n},G\rangle| \leq& \frac{1}{n} \sum _{x } | n^{\gamma}  \mcb {K}_n G \left( \tfrac{x}{n} \right) -[-(- \Delta)^{\frac{\gamma}{2}}   G]  (\tfrac{x}{n} ) |  + \frac{1}{n} \sum _{x } | [-(- \Delta)^{\frac{\gamma}{2}}   G]  (\tfrac{x}{n} ) |  \\
 +& \frac{n^{\gamma-1}}{2} \Big( 1 - \frac{\alpha}{n^{\beta}} \Big)  \sum_{\{x,y \} \in \mcb S }  |G(s,\tfrac{y}{n}) - G(s,\tfrac{x}{n})| p(y-x).
\end{align*}
Now, from Proposition \ref{convdisc} and Proposition \ref{neum1}, we conclude that
\begin{align*}
|n^{\gamma}  \mcb L_{n}\langle \pi_{s}^{n},G\rangle| \lesssim \int_{\mathbb{R}} | [-(- \Delta)^{\frac{\gamma}{2}}   G]  (u ) | du \lesssim 1,
\end{align*}
and this ends the proof in case \eqref{defsigmas} holds. 
 
\textbf{II.)} If  \eqref{defsigmas} does not  hold, since $|\eta_{s}^{n}(x)| \leq 1$, from \eqref{srobterm} and \eqref{sneuterm}, we can write
\begin{align*}
& |n^{\gamma}  \mcb L_{n}\langle \pi_{s}^{n},G\rangle|  \\
\leq & \frac{1}{n}   \sum_{x}  \Big|n^{\gamma} \sum_{y: \{x,y \} \in \mcb F }  [G(\tfrac{y}{n}) - G(\tfrac{x}{n})] p(y-x) -  [-(- \Delta)_{\mathbb{R}^{*}}^{\frac{\gamma}{2}}   G]  (\tfrac{x}{n} )    \Big|   \\
+ & \alpha \frac{1}{n}   \sum_{x}  \Big|n^{\gamma-\beta} \sum_{y: \{x,y \} \in \mcb S }  [G(\tfrac{y}{n}) - G(\tfrac{x}{n})] p(y-x) - \mathbbm{1}_{\beta=0} \big( [-(- \Delta)^{\frac{\gamma}{2}}   G]  (\tfrac{x}{n} ) - [-(- \Delta)_{\mathbb{R}^{*}}^{\frac{\gamma}{2}}   G]  (\tfrac{x}{n} )   \big) \Big|   \\
 +&  \frac{1}{n} \sum _{x } \Big|  [-(- \Delta)_{\mathbb{R}^{*}}^{\frac{\gamma}{2}}   G]  (\tfrac{x}{n} ) \Big|   + \alpha \mathbbm{1}_{\beta=0}  \frac{1}{n} \sum _{x }\Big | [-(- \Delta)^{\frac{\gamma}{2}}   G]  (\tfrac{x}{n} ) - [-(- \Delta)_{\mathbb{R}^{*}}^{\frac{\gamma}{2}}   G]  (\tfrac{x}{n} ) \Big| . 
\end{align*}
Then from Proposition \ref{lemconvneum} and Proposition \ref{convslow},  we have
\begin{align*}
|n^{\gamma}  \mcb L_{n}\langle \pi_{s}^{n},G\rangle| \lesssim \alpha \int_{\mathbb{R}} | [-(- \Delta)^{\frac{\gamma}{2}}   G]  (u ) | du + (\alpha + 1) \int_{\mathbb{R}} | [-(- \Delta)_{\mathbb{R}^{*}}^{\frac{\gamma}{2}}   G]  (u ) | du \lesssim 1.
\end{align*}
\end{proof}
\begin{proof}[Proof of Proposition \ref{tightcond2dif}]
From Dynkin's formula
\begin{align*}
\mathbb{E}_{\mu_n} \left[ \left( \mcb M_{\tau}^{n}(G) -  \mcb M_{\tau+\bar{\tau}}^{n}(G) \right)^2 \right] =
 \mathbb{E}_{\mu_n} \Big[ \int_{\tau}^{\tau+\bar{\tau}}n^{\gamma} \left(\mcb L_{n}[\langle \pi_{s}^{n},G (s, \cdot) \rangle]^{2}-2\langle \pi_{s}^{n},G(s, \cdot) \rangle \mcb L_{n} \langle \pi_{s}^{n},G (s, \cdot) \rangle\right) ds \Big].
\end{align*}
Simple computations show that for $G\in \mcb S_{Rob}$  (which also includes $G\in \mcb S_{Dif}$),  the integrand function in last display is equal to 
\begin{align*}
& \frac{n^{\gamma-2}}{2}  \sum_{\{w,z \} \in \mcb F } \big[ G\left(s, \tfrac{w}{n}\right) - G\left(s, \tfrac{z}{n}\right) ]{^2}  p(w-z) [\eta(w)-\eta(z)]^2 \\
+&  \frac{\alpha n^{\gamma - 2}}{2n^{\beta}} \sum_{\{w,z \} \in \mcb S } \big[ G\left(s, \tfrac{w}{n}\right) - G\left(s, \tfrac{z}{n}\right) ]^2  p(w-z) [\eta(w)-\eta(z)]^2.
\end{align*}
Since $ |\eta_s^n(z)|\leq 1, \; \forall z \in \mathbb{Z}$,  if we  assume $G \in \mcb S_{{Dif}}$, then last display can be bounded from above by
\begin{align*}
 \frac{\alpha +1}{2} n^{\gamma-2} \sum_{w,z }     \big[G\left(s,\tfrac{z}{n}\right)-G\left(s,\tfrac{w}{n}\right)\big ]^2 p(z-w)
\end{align*}
and from  Proposition \ref{tight2condaux} the proof ends. Finally, if we assume  $G \in \mcb S_{{Rob}}$ and  $\mcb S = \mcb S_0$, $\gamma \in (1,2)$ and $\beta \geq \gamma-1$, then the integrand function can be bounded from above by 
\begin{align*}
&  \frac{1}{2} \sum_{w=0}^{\infty} \sum_{z=0}^{\infty}  \big[ G_{+} \left(s, \tfrac{w}{n}\right) - G_{+}\left(s,\tfrac{z}{n}\right) ]^2  p(w-z)  
 + \frac{1}{2} \sum_{w=-\infty}^{-1} \sum_{z=-\infty}^{-1}  \big[ G_{-} \left(s, \tfrac{w}{n}\right) - G_{-} \left(s,\tfrac{z}{n}\right) ]^2  p(w-z)  \\
+&  \frac{\alpha}{2n} \sum_{\{w,z \} \in \mcb S } \big[ G\left(s, \tfrac{w}{n}\right) - G\left(s, \tfrac{z}{n}\right) ]^2  p(w-z)  \lesssim \max\{ n^{\gamma-2}, n^{-1} \} +  n^{-1} \lesssim \max\{ n^{\gamma-2}, n^{-1} \}. 
\end{align*}
Above we applied Proposition \ref{tight2condaux}  for $G_{+}$ and $G_{-}$; furthermore, we note the sum over $\{w,z \} \in \mcb S$ in  last line is bounded by $m (2 \| G \|_{\infty})^2$, since $\gamma >1$. 
This ends the proof. 
\end{proof}
\section{Characterization of limit points} \label{seccharac}
In this section we characterize the limit point $\mathbb Q$ of the sequence $(\mathbb{Q}_{n})_{ n \geq 1 }$, whose existence is a consequence  of  the results  of last section. We first observe that since we deal with an exclusion process, according to \cite{kipnis1998scaling}, $\mathbb{Q}$ is concentrated on trajectories $\pi_t(du)$ which are absolutely continuous with respect to the Lebesgue measure, that is $\pi_t(du)= \varrho(t,u)du$.
Now we want  to show that $\mathbb{Q}$ is concentrated on trajectories $\varrho$ satisfying the first condition of weak solutions of our hydrodynamic equations. Before doing that, we first show that  the third condition of weak solutions of \eqref{eqhydfracdifnotuniq} holds for $\gamma \in (1,2)$ and $\beta \in [0, \gamma - 1)$. 
\begin{prop} \label{dircond}
Assume $\gamma \in (1,2)$ and $\beta \in [0, \gamma - 1)$. Under the limit point  $\mathbb{Q}$ we have
\begin{align*}
\mathbb{Q} \Big( \pi_{ \cdot} \in \mcb D \big([0,T], \mcb{M}^+(\mathbb{R}) \big):  \int_0^t [\varrho(s,0^{+}) - \varrho(s,0^{-})] ds = 0,  \forall t \in [0,T]   \Big) = 1.
\end{align*}
\end{prop}

To simplify notation in all what follows, we erase $\pi_{\cdot}$ from the sets where we  look at.
\begin{proof}
In order to prove the proposition, it is enough to verify, for any $\delta >0$
\begin{align} \label{3conda}
\mathbb{Q} \Big( \sup_{t \in [0,T]} \Big| \int_0^t [\varrho(s,0^{+}) - \varrho(s,0^{-})] ds \Big| > \delta  \Big) = 0.
\end{align}
For every $\varepsilon>0$ we define two approximations of the identity given by
\begin{equation} \label{aproxiden}
\iota_{\varepsilon}^{0^+}(u):= \frac{1}{\varepsilon} \mathbbm{1}_{(0, \varepsilon]}(u); \;  \quad \textrm{and}\quad \; \iota_{\varepsilon}^{0^-}(u):= \frac{1}{\varepsilon} \mathbbm{1}_{[-\varepsilon,0)}(u).
\end{equation}
 From Proposition \ref{estenergcombarlenfor} and Proposition \ref{holderrep}, there exists $C>0$ such that $\sup_{s \in [0,T]}| \rho(s,u) -\rho(s,v)| \leq C |u-v|^{\frac{\gamma-1}{2}}$ for $u, v$ such that $u v \geq 0$. This leads to
\begin{align*}
& \lim_{\varepsilon \rightarrow 0^{+}} \mathbb{Q} \Big( \sup_{t \in [0,T]} \Big| \int_0^t  \int_0^{\varepsilon} \frac{[\varrho(s,0^{+})-\varrho(s,u)]}{\varepsilon} du ds \Big| > \frac{\delta}{3}  \Big) \\
=& \lim_{\varepsilon \rightarrow 0^{+}} \mathbb{Q} \Big( \sup_{t \in [0,T]} \Big| \int_0^t  \int_{-\varepsilon}^{0} \frac{[\varrho(s,0^{-})-\varrho(s,u)]}{\varepsilon} du ds \Big| > \frac{\delta}{3}  \Big) = 0.
\end{align*}
To get \eqref{3conda} it is enough to prove that
\begin{align*}
\lim_{\varepsilon \rightarrow 0^{+}} \mathbb{Q} \Big( \sup_{t \in [0,T]} \Big| \int_0^t \int_{\mathbb{R}}  \varrho(s,u) [  \iota_{\varepsilon}^{0^+}(u) - \iota_{\varepsilon}^{0^-}(u) ] du ds \Big| > \frac{\delta}{3}  \Big) = 0.
\end{align*}
Since $\iota_{\varepsilon}^{0^+}$ and $\iota_{\varepsilon}^{0^-}$ are not continuous functions, we cannot use Portmanteau's Theorem directly. However, we can approximate these functions by a continuous function  in such a way that the error vanishes. When dealing with the continuous function we apply Portmanteau's theorem and then, by approximating again,  we go back to our original functions. Doing so,    we are reduced to prove that
\begin{align*}
\lim_{\varepsilon \rightarrow 0^{+}} \liminf_{n \rightarrow \infty} \mathbb{Q}_n \Big( \sup_{t \in [0,T]} \Big| \int_0^t \int_{\mathbb{R}}    [  \langle \pi_s^n, \iota_{\varepsilon}^{0^+} \rangle  -  \langle \pi_s^n, \iota_{\varepsilon}^{0^-} \rangle ]  du ds \Big| > \frac{\delta}{12}  \Big) = 0.
\end{align*}
For   $\ell \geq 1$we define  the empirical averages on a box of size $\ell$ around $0$ as 
\begin{equation} \label{medemp}
\eta^{\rightarrow \ell}(0):=\frac{1}{\ell} \sum_{y=1}^{\ell} \eta(y)   \; \; \text{and} \; \; \eta^{\leftarrow \ell}(0):=\frac{1}{\ell} \sum_{y=-\ell}^{-1} \eta(y).
\end{equation}
Hereinafter we interpret $\varepsilon n$ as $\lfloor \varepsilon n\rfloor$ and with this notation, the last double limit can be rewritten as
\begin{align*}
\lim_{\varepsilon \rightarrow 0^{+}} \liminf_{n \rightarrow \infty} \mathbb{P}_{\mu_n} \Big( \sup_{t \in [0,T]} \Big| \int_0^t  [ \eta_s^{\rightarrow \varepsilon n}(0) - \eta_s^{\leftarrow \varepsilon n}(0)]  ds \Big| > \frac{\delta}{12}  \Big) = 0.
\end{align*}
From $\beta \in [0, \gamma-1)$, Markov's inequality and taking $F(s)=1$ for every $s \in [0,T]$ in  Lemma \ref{replemma}, we get the result.
\end{proof}
Now we prove that in the case 
\begin{prop} \label{caraclimsembarlen}
Assume \eqref{defsigmas}. Then, the limit point is concentrated on trajectories of measures satisfying the first condition of Definition \ref{eq:dif}.
\begin{align*}
\mathbb{Q} \Big(\pi_{ \cdot} \in \mcb D([0,T], \mathcal{M}^+(\mathbb{R})): F_{FrDif}(t,\rho,G, \mcb g) = 0, \forall t \in [0,T], \forall G \in \mcb S_{Dif}   \Big) = 1.
\end{align*}
\end{prop}
\begin{proof}
In order to prove the proposition, it is enough to verify for any $\delta >0$ and $G \in \mcb S_{Dif}$ that
\begin{align} \label{cond1}
\mathbb{Q} \Big(\sup_{t \in [0,T]} |F_{FrDif}(t, \rho, G, g)| > \delta  \Big) = 0.
\end{align}
 We can bound the probability in \eqref{cond1} by the sum of the following two terms:
\begin{align}
\mathbb Q\Big( \sup_{t \in [0,T]}  \Big| & \int_{\mathbb{R}} \varrho(t,u) G(t,u) du -\varrho(0,u) G(0,u) du 
-  \int_0^t \int_{\mathbb{R}} \varrho(s,u) \big[ [ -(- \Delta)^{\frac{\gamma}{2}} ]  + \partial_s \big] G(s,u) du ds   \Big| > \dfrac{\delta}{2} \Big) \label{f1term1sdif}
\end{align}
and
\begin{equation*} \label{f1term2sdif}
\mathbb{Q} \Big( \sup_{t \in [0,T]}   \Big| \int_{\mathbb{R}} \left[ \varrho(0,u) - \mcb g(u) \right] G(0,u) du \Big| > \dfrac{\delta}{2}  \Big).
\end{equation*}
Since $\mathbb{Q}$ is a limit point of $(\mathbb{Q}_n)_{n \geq 1 }$, which is induced by  $(\mu_n)_{n  \geq 1}$ associated to the profile $\mcb g$, then  the last display is equal to zero. Therefore, it remains only to treat \eqref{f1term1sdif}. We first  observe that $ [ -(- \Delta)^{\frac{\gamma}{2}}  G (s, \cdot)]$ is not in $C_c^{\infty}(\mathbb{R})$. However, since  $ [ -(- \Delta)^{\frac{\gamma}{2}} G (s, \cdot)]  \in L^1(\mathbb{R})$, we can approximate $ [ -(- \Delta)^{\frac{\gamma}{2}} G (s, \cdot)]$ in $L^1(\mathbb{R})$ by a sequence in $C_c^{\infty}(\mathbb{R})$ in order to make use of Portmanteau's Theorem as  in the  previous proof.  Doing so, it is enough to treat
\begin{align*}
\liminf_{n \rightarrow \infty} \mathbb{Q}_n \Big(\sup_{t \in [0,T]}  \Big| & \int_{\mathbb{R}} \varrho(t,u) G(t,u) du -\int_{\mathbb{R}} \varrho(0,u) G(0,u) du  
\\
-&  \int_0^t \int_{\mathbb{R}} \varrho(s,u) \big[ [ -(- \Delta)^{\frac{\gamma}{2}} ]  + \partial_s \big] G(s,u) du ds   \Big| > \dfrac{\delta}{8} \Big). 
\end{align*}
From \eqref{defMnt} and Proposition \ref{gensdif}, we can bound the last probability by the sum of the next three terms
\begin{equation} \label{f1term1bdif}
\liminf_{n \rightarrow \infty} \mathbb{P}_{\mu_n}  \Big(  \sup_{t \in [0,T]}  |\mcb M_{t}^{n}(G)| > \frac{\delta}{24} \Big),
\end{equation}
\begin{align}
& \liminf_{n \rightarrow \infty}  \mathbb{P}_{\mu_n} \Big(  \sup_{t \in [0,T]}   \Big| \int_0^t \Big\{ \frac{1}{n} \sum _{x} n^{\gamma} \mcb{K}_n G \left(s,\tfrac{x}{n} \right)  \eta_{s}^{n}(x) -  \langle \pi_s^n, -(- \Delta)^{\frac{\gamma}{2}} G(s, \cdot) \rangle \Big\} ds \Big| > \frac{\delta}{24} \Big) \label{f1term1cdif}
\end{align}
and
\begin{align}
& \liminf_{n \rightarrow \infty}  \mathbb{P}_{\mu_n} \Big(  \sup_{t \in [0,T]} \Big| \int_{0}^{t}   \frac{n^{\gamma-1}}{2} \Big( 1 - \frac{\alpha}{n^{\beta}} \Big)   \sum_{\{x,y\} \in \mcb S}  [G(s,\tfrac{y}{n}) - G(s,\tfrac{x}{n})] p(y-x)  [\eta_s^n(y)-\eta_s^n(x)] ds \Big| > \frac{\delta}{24} \Big). \label{f1term1ddif}
\end{align}
From Doob's inequality and Proposition \ref{tightcond2dif}, the limit in \eqref{f1term1bdif} is equal to zero. From Corollary \ref{pdbgeq}, Corollary \ref{convbound} and Markov's inequality,  then \eqref{f1term1cdif} and \eqref{f1term1ddif}, respectively, are also equal to zero. 
\end{proof}

\begin{prop} \label{caraclimbetazero}
Assume $\mcb S = \mcb S_0$,  $\beta=0$ and $\alpha \neq 1$. Then
\begin{align*}
\mathbb{Q} \Big(\pi_{ \cdot} \in \mcb  D([0,T], \mcb{M}^+(\mathbb{R})): F_{FrDif2}(t,\rho,G, \mcb g, \alpha) = 0, \forall t \in [0,T], \forall G \in \mcb S_{\gamma}   \Big) = 1,
\end{align*}
where $\mcb S_{\gamma}= \mcb S_{ \textrm{Dif} }$ if $\gamma \in (1,2)$ and $\mcb S_{\gamma}= \mcb S_{ \textrm{Neu} }$ if $\gamma \in (0,1]$.
\end{prop}
\begin{proof}
As in the previous proof,  it is enough to verify for any $\delta >0$ and $G \in \mcb S_{\gamma}$ that
\begin{align} \label{cond2}
\mathbb{Q} \Big(\sup_{t \in [0,T]} |F_{FrDif2}(t, \rho, G, g, \alpha)| > \delta  \Big) = 0.
\end{align}
 Now we observe that $ [ -(- \Delta)^{\frac{\gamma}{2}}  G (s, \cdot)]$ and $ [ -(- \Delta)_{\mathbb{R}^{*}}^{\frac{\gamma}{2}}  G (s, \cdot)]$ are not in $C_c^{\infty}(\mathbb{R})$. However, since $ [ -(- \Delta)^{\frac{\gamma}{2}} G (s, \cdot)] $ and $ [ -(- \Delta)_{\mathbb{R}^{*}}^{\frac{\gamma}{2}}  G (s, \cdot)]$ are in $L^1(\mathbb{R})$, we can approximate them in $L^1(\mathbb{R})$ by sequences in $C_c^{\infty}(\mathbb{R})$ so that we can use  Portmanteau's Theorem as we did in the previous proof. Since $\mathbb{Q}$ is a limit point of $(\mathbb{Q}_n)_{n \geq 1 }$,  induced by  $(\mu_n)_{n  \geq 1}$ associated to the profile $\mcb g$, to conclude, it is enough to prove that 
\begin{align*}
\liminf_{n \rightarrow \infty} \mathbb{Q}_n \Big( \sup_{t \in [0,T]} \Big| & \int_{\mathbb{R}} \varrho(t,u) G(t,u) du -\int_{\mathbb{R}} \varrho(0,u) G(0,u) du  -  \int_0^t \int_{\mathbb{R}} \varrho(s,u) \partial_s G(s,u) du ds \nonumber
\\
-&  \int_0^t \int_{\mathbb{R}} \varrho(s,u) \big[ \alpha [ -(- \Delta)^{\frac{\gamma}{2}} ]  + ( 1 - \alpha ) [ -(- \Delta)_{\mathbb{R}^{*}}^{\frac{\gamma}{2}} ] \big] G(s,u) du ds   \Big| > \dfrac{\delta}{8} \Big)=0.
\end{align*}
From \eqref{defMnt} and Proposition \ref{gensdif}, we can bound the limit in the last display by the sum of the next three terms
\begin{equation} \label{f2term1bsdif}
\liminf_{n \rightarrow \infty} \mathbb{P}_{\mu_n}  \Big(  \sup_{t \in [0,T]} |\mcb M_{t}^{n}(G)| > \frac{\delta}{24} \Big),
\end{equation}
\begin{align}
 \liminf_{n \rightarrow \infty}  \mathbb{P}_{\mu_n} \Big(   \sup_{t \in [0,T]}  \Big| \int_0^t \Big\{ &  n^{\gamma-1}  \!\! \!\!\!\! \!\!\sum_{ \{x,y\} \in \mcb F } \!\! \!\![ G(s, \tfrac{y}{n} ) - G(s, \tfrac{x}{n} ) ] p(y-x) \eta_s^n(x) -   \langle \pi_s^n, [-(- \Delta)^{\frac{\gamma}{2}}   G] (s, \cdot) \rangle  \Big\} ds \Big| > \frac{\delta}{24} \Big) \label{f2term1csdif},
\end{align}
\begin{align}
 \liminf_{n \rightarrow \infty}  \mathbb{P}_{\mu_n} \Big(   \sup_{t \in [0,T]}  \Big| \int_0^t \Big\{ & \alpha n^{\gamma-1-\beta}  \sum_{ \{x,y\} \in \mcb S } [ G(s, \tfrac{y}{n} ) - G(s, \tfrac{x}{n} ) ] p(y-x) \eta_s^n(x)  \nonumber  \\
- & \alpha \big( \langle \pi_s^n, [-(- \Delta)^{\frac{\gamma}{2}}   G] (s, \cdot) \rangle - \langle \pi_s^n, [-(- \Delta)_{\mathbb{R}^{*}}^{\frac{\gamma}{2}}   G] (s, \cdot) \rangle  \big) \Big\} ds \Big| > \frac{\delta}{24} \Big). \label{f2term1dsdif}
\end{align}
From Doob's inequality and Proposition \ref{tightcond2dif}, the limit in  \eqref{f2term1bsdif} is equal to zero. From Corollary \ref{corconvfast}, Corollary \ref{corconvslow} and Markov's inequality, then \eqref{f2term1csdif} and \eqref{f2term1dsdif}, respectively, are also equal to zero. This ends the proof.
\end{proof}
\begin{prop} \label{caraclimcombarlen}
Assume that $\mcb S = \mcb S_0$ and $\beta > 0$. Then
\begin{align*}
\mathbb{Q} \Big(\pi_{ \cdot} \in \mcb D([0,T], \mcb{M}^+(\mathbb{R})): F_{FrRob}(t, \rho ,G, g, m \alpha \mathbbm{1}_{\gamma \in (1,2)} \mathbbm{1}_{\beta = \gamma-1}  ) = 0, \forall t \in [0,T], \forall G \in \mcb S_{\gamma} \Big) = 1,
\end{align*}
where $\mcb S_{\gamma}=\mcb S_{Neu}$ when $\gamma \in (0,1]$ and $\mcb S_{\gamma}=\mcb S_{Rob0}$ when $\gamma \in (1,2)$.
\end{prop}
\begin{proof}
As above, we are left to verify, for any $\delta >0$ and $G \in \mcb S_{\gamma}$, that
\begin{align*}
\mathbb{Q} \Big( \sup_{t \in [0,T]} |F_{FrRob}(t, \rho,G, g, m \alpha \mathbbm{1}_{\gamma \in (1,2)} \mathbbm{1}_{\beta = \gamma-1})| > \delta  \Big) = 0.
\end{align*}
Observe that, in the case $(\gamma, \beta) \in (1,2) \times \{ \gamma - 1 \}$, due to the boundary terms  $\varrho(s,0^-)$ and $\varrho(s,0^+)$ that appear in last display, we deal with sets which are not open in the Skorohod topology and, for that reason, we are not  able to use directly Portmanteau's Theorem. In order to avoid this problem, we will proceed in the same way we did in the proof of Proposition \ref{dircond}, making use of the two approximations of the identity $\iota_{\varepsilon}^{0^+}$ and $ \iota_{\varepsilon}^{0^-}$ defined in \eqref{aproxiden}.  Moreover, since $ [ -(- \Delta)_{\mathbb{R}^{*}}^{\frac{\gamma}{2}}  G (s, \cdot)]$ is not in $C_c^{\infty}(\mathbb{R})$, we will approximate it  in $L^1(\mathbb{R})$ by  a sequence in $C_c^{\infty}(\mathbb{R})$ in order to make use of Portmanteau's Theorem, exactly as we did in the proof of Proposition \ref{dircond}. Moreover, summing and subtracting to $\varrho(s,0^+)$ (resp., $\varrho(s,0^-)$) the mean $\langle \pi_s, \iota_{\varepsilon}^{0^+} \rangle$ (resp., $\langle \pi_s, \iota_{\varepsilon}^{0^-} \rangle$); recalling that $ \mathbb{Q}$ is a limit point of $(\mathbb{Q}_n)_{n \geq 1 }$,  induced by  $(\mu_n)_{n  \geq 1}$ associated to the profile $~\mcb g$;  and approximating $\iota_{\varepsilon}^{0^+}$ and $\iota_{\varepsilon}^{0^-}$ by continuous functions in such a way that the error vanishes as $\varepsilon \rightarrow 0^+$, it is enough to prove that 
\begin{align}
& \liminf_{n \rightarrow \infty} \mathbb{Q}_n \Big(  \sup_{t \in [0,T]} \Big| \int_{\mathbb{R}} \rho(t,u) G(t,u) du -\int_{\mathbb{R}} \rho(0,u) G(0,u) du  \nonumber  \\
- & \int_0^t \Big[ \int_{\mathbb{R}} \rho(s,u) \left(  \big[-(- \Delta)_{\mathbb{R}_{*}}^{\frac{\gamma}{2}} G \big] (s,u) + \partial_s G(s,u) \right)  du \Big] ds \nonumber  \\
+ & m \alpha \mathbbm{1}_{\gamma \in (1,2)} \mathbbm{1}_{\beta = \gamma-1} \int_0^t  [    \langle \pi_s, \iota_{\varepsilon}^{0^+} \rangle  -  \langle \pi_s, \iota_{\varepsilon}^{0^-} \rangle  ]  [  G(s,0^{+})    -  G(s,0^{-})  ] ds \Big| > \frac{\delta}{12} \Big)=0. \label{f2term1asdif}
\end{align}
From \eqref{defMnt} and Proposition \ref{gensdif}, we can bound the limit in  last display by the sum of the next three terms
\begin{equation} \label{f3term1bsdif}
\liminf_{n \rightarrow \infty} \mathbb{P}_{\mu_n}   \Big(  \sup_{t \in [0,T]} |\mcb M_{t}^{n}(G)| > \frac{\delta}{36} \Big),
\end{equation}
\begin{align}
 \liminf_{n \rightarrow \infty}  \mathbb{P}_{\mu_n} \Big(  \sup_{t \in [0,T]}  \Big| \int_0^t & \!\!\!\!\sum_{\{x,z\} \in \mcb F} \!\!\!\!n^{\gamma-1} [G(s,\tfrac{y}{n}) - G(s,\tfrac{x}{n}) ]  p(y-x)  \eta_s^n(x) 
 -  \langle \pi_s^n, -(- \Delta)_{\mathbb{R}_{*}}^{\frac{\gamma}{2}} G(s, \cdot) \rangle ds \Big| > \frac{\delta}{36} \Big), \label{f3term1dsdif}
\end{align}
\begin{align}
 \limsup_{\varepsilon \rightarrow 0^{+}} \liminf_{n \rightarrow \infty}  \mathbb{P}_{\mu_n} \Big( & \sup_{t \in [0,T]} \Big|\int_{0}^{t} \Big\{  \alpha n^{\gamma-1-\beta}  \sum_{\{x,z\} \in \mcb S}  [G(s,\tfrac{y}{n}) - G(s,\tfrac{x}{n}) ]  p(y-x)  \eta_s^n(x) \nonumber  \\
-&   m \alpha \mathbbm{1}_{\gamma \in (1,2)} \mathbbm{1}_{\beta = \gamma-1} \int_0^t  [  G(s,0^{-})    -  G(s,0^{+})  ]  [   \eta_s^{\rightarrow n \varepsilon}(0)  - \eta_s^{\leftarrow n \varepsilon}(0)  ] \Big\} ds \Big| > \frac{\delta}{36} \Big). \label{f3term1csdif}
\end{align}
From Doob's inequality and Proposition \ref{tightcond2dif}, we conclude that  \eqref{f3term1bsdif} is equal to zero. From Corollary \ref{corconvfast} and Markov's inequality, we conclude that \eqref{f3term1dsdif} is also equal to zero. Finally, from Corollary \ref{corconvslow}, Proposition \ref{convrob} and Markov's inequality, we conclude that \eqref{f3term1csdif} is also equal to zero. 
\end{proof}

\section{Energy estimates}  \label{secenerest}
In this section, our goal is to prove that $\varrho$ belongs to some fractional Sobolev space, i.e. that   $\varrho$ satisfies the second condition of weak solutions of our hydrodynamical equations. Hereinafter, we fix $C_a >0$ such that $H( \mu_n | \nu_a) \leq  C_a n, \forall n \geq 1$, where $H( \mu_n | \nu_a)$ is the entropy bound in the statement of Theorem \ref{hydlim}. Similarly to \cite{casodif}, we begin stating an important result which  does not depend on the dynamics but only on $H( \mu_n | \nu_a)$. 
Recall that $\bar{\varrho}:=\varrho-a$. We do not present the  proof  of next result since it is given in Section 5.1 of \cite{casodif}.

\begin{prop} \label{estenergstat0}
For all $\mcb S \subset \mcb S_0$, all $\gamma\in(0,2)$, for any $ \beta\geq 0$ and any  $\alpha>0$, we have that
\begin{align*} 
\mathbb{Q} \Big( \pi_{ \cdot} \in \mcb D \big([0,T], \mcb{M}^+(\mathbb{R}) \big): \int_0^T \int_{\mathbb{R}} [\bar{\varrho}(t,u) ]^2   du dt < \infty \Big) = 1.
\end{align*}
\end{prop}

Now we assume that  either \eqref{defsigmas} holds or $\beta=0$ and our goal is to prove that $\bar\rho$ belongs to the fractional  Sobolev space on the full line.  Our proof is strongly inspired by the strategy presented in Subsection 3.3 of \cite{byronsdif}.

\begin{prop} \label{estenergsembarlenforsdif}
Assume that \eqref{defsigmas} holds or $\beta=0$. Then
\begin{align*}
\mathbb{Q} \Big (\pi_{ \cdot} \in \mcb D([0,T], \mcb{M}^+(\mathbb{R})): \bar{\varrho} \in L^2 \big(0, T ; \mcb{H}^{\frac{\gamma}{2}}  ( \mathbb{R} ) \big)  \Big) = 1.
\end{align*}
\end{prop}
\begin{proof}
Let $\varepsilon >0$ and fix $F \in C_{c}^{0,\infty} \big( (0,T) \times \mathbb{R}^2 \big)$. We denote the antisymmetric (with respect to the space variable) part of $F$ as $F^a$, i.e, 
\begin{align*}
F^a(t,u,v) := \frac{ F(t,u,v) - F(t,v,u)}{2} , \forall t \in (0,T), \forall u,v \in \mathbb{R}.
\end{align*}
\textbf{I.)} First we assume \eqref{defsigmas}.  From this and the fact that there exists $C_{\delta}$ (independent of $t$) such that $|F^a(t,u,v)| \leq C_{\delta} |u-v|^{\delta}$ for every $u,v \in \mathbb{R}$, for every $t \in [0,T]$,  we get that
\begin{align*}
 \mathbb{E}_{\mu_n} \Big[ \Big|    \int_0^T \frac{n^\gamma}{n}  \sum_{\substack{\{x,y\} \in \mcb S\\ |x-y| \geq \varepsilon n}}   F^{a} \left(t, \tfrac{x}{n}, \tfrac{y}{n} \right) p(y-x) [\eta_t^n(y) -  \eta_t^n(x) ] dt  \Big|  \Big] \lesssim n^{\gamma-1-\delta} \sum _{\{y,z\} \in \mcb S}  p(y-z)  |y-x|^{\delta}.   
\end{align*}
Since $\delta > \gamma - 1$, taking $n \rightarrow \infty$ the left-hand side of the last display goes to zero and we have
\begin{align}
& \lim_{n \rightarrow \infty} \mathbb{E}_{\mu_n} \Big[   \int_0^T \dfrac{n^{\gamma} }{n}   \sum_{|x-y| \geq \varepsilon n}  F^{a} \left(t, \tfrac{x}{n}, \tfrac{y}{n} \right) p(y-x) [\eta_t^n(y) -  \eta_t^n(x) ] dt  \Big] \nonumber \\
=& \lim_{n \rightarrow \infty} \mathbb{E}_{\mu_n} \Big[    \int_0^T \dfrac{n^{\gamma} }{n}   \sum_{\substack{\{x,y\} \in \mcb F\\ |x-y| \geq \varepsilon n}}  F^{a} \left(t, \tfrac{x}{n}, \tfrac{y}{n} \right) p(y-x) [\eta_t^n(y) -  \eta_t^n(x) ] dt  \Big]. \label{eqlim}
\end{align}
Now we follow closely the proof of item i) of Theorem 3.2 in \cite{byronsdif}. By the entropy inequality, Jensen's inequality and Feynman-Kac's formula, we have
\begin{align*}
& \dfrac{1}{n} \mathbb{E}_{\mu_n} \Big[  \int_0^T   n^{\gamma}  \sum_{\substack{\{x,y\} \in \mcb F\\ |x-y| \geq \varepsilon n}}  F^a \left(t, \tfrac{x}{n}, \tfrac{y}{n} \right) p(y-x) [\eta_t^n(y) -  \eta_t^n(x) ] dt \Big] \\
\leq& C_a + \int_0^T    \sup_f \Big\{  n^{\gamma-1} \Big[  \sum_{\substack{\{x,y\} \in \mcb F\\ |x-y| \geq \varepsilon n}}  F^a \left(t, \tfrac{x}{n}, \tfrac{y}{n} \right) p(y-x)  \int [ \eta(y) - \eta(x) ]  d \nu_a  +   \langle \mcb L_n \sqrt{f} , \sqrt{f} \rangle_{\nu_a} \Big] \Big\}   dt,
\end{align*}
where the supremum is taken over all the densities $f$ on $\Omega$ with respect to $\nu_a$. Above, $\langle f, g\rangle_{\nu_a}$ is the scalar product between $f$ and $g$ in $L^2(\Omega, \nu_a)$, that is,
$\langle f, g\rangle _{\nu_a} := \int f(\eta) g(\eta) d \nu_a
$ 
and $\langle\mcb {L}_{n} \sqrt{f} , \sqrt{f} \rangle_{\nu_a}$ is the Dirichlet form. A simple computation shows that 
 \begin{align}\label{bound}
\langle \mcb {L}_{n} \sqrt{f} , \sqrt{f} \rangle_{\nu_a} = - \dfrac{1}{2}   D_n (\sqrt{f}, \nu_{a} ),
\end{align}
where 
$D_n (\sqrt{f}, \nu_a ) := D_n^{\mcb F} (\sqrt{f}, \nu_a ) + D_n^{\mcb S} (\sqrt{f}, \nu_a )$,
with
\begin{align} \label{DnF}
D_n^{\mcb F}  (\sqrt{f}, \nu_a ) := \frac{1}{2} \sum_{\{ x, y \} \in \mcb F} p(y-x) I_{x,y}  (\sqrt{f}, \nu_a ), \quad 
 D_n^{\mcb S}  (\sqrt{f}, \nu_a ) :=  \frac{\alpha}{2n^{\beta}}  \sum_{\{ x, y \} \in \mcb S} p(y-x) I_{x,y}  (\sqrt{f}, \nu_a ),
\end{align} and 
$I_{x,y}  (\sqrt{f}, \nu_a ) := \int [ \sqrt{f \left( \eta^{x,y} \right) } - \sqrt{f \left( \eta \right) } ]^2 d \nu_a$. From this, we get for every $t \in (0,T)$:
\begin{align}
&   \sum_{\substack{\{x,y\} \in \mcb F\\ |x-y| \geq \varepsilon n}}  F^a \left(t, \tfrac{x}{n}, \tfrac{y}{n} \right) p(y-x)  \int [\eta(y) - \eta(x) ]  d \nu_a  +    \langle \mcb L_n \sqrt{f} , \sqrt{f} \rangle_{\nu_a} \nonumber \\
\leq &  \sum_{\substack{\{x,y\} \in \mcb F\\ |x-y| \geq \varepsilon n}} p(y-x) \Big[ \big|  F^a \left(t, \tfrac{x}{n}, \tfrac{y}{n} \right) \big|  \Big|  \int  [\eta(y) - \eta(x) ]  d \nu_a  \Big|- \dfrac{  I_{x,y}   (\sqrt{f}, \nu_{a} )}{4} \Big]. \label{Axyestenerg}
\end{align}
Now observe that from a change of variables $\eta $ to $\eta^{x,y}$ and from  Young's inequality, there exists a positive constant $A_{x,y}$  such that
\begin{align} \label{young}
\Big| \int [\eta(x) - \eta(y)] f(\eta) d \nu_a \Big| \leq  \frac{ I_{x,y}  (\sqrt{f}, \nu_a )}{2A_{x,y}} + 2 A_{x,y}.
\end{align}
 Choosing $A_{x,y}:= 2  |  F^a \left(t, \tfrac{x}{n}, \tfrac{y}{n} \right) |$ in \eqref{young}, we can bound the expression in \eqref{Axyestenerg} from above by
\begin{align*}
\frac{ 4 c_{\gamma} }{n^2}  \sum_{\substack{\{x,y\} \in \mcb F\\ |x-y| \geq \varepsilon n}} [   F^a \left(t, \tfrac{x}{n}, \tfrac{y}{n} \right) ]^2  	\Big| \tfrac{x}{n} - \tfrac{y}{n} \Big|^{-1-\gamma},
\end{align*}
from where we conclude that
\begin{align*}
& \mathbb{E}_{\mu_n} \Big[    \int_0^T \frac{1}{n}  n^{\gamma}  \sum_{\substack{\{x,y\} \in \mcb F\\ |x-y| \geq \varepsilon n}}  F^a \left(t, \tfrac{x}{n}, \tfrac{y}{n} \right) p(y-x) [\eta_t^n(y) -  \eta_t^n(x) ] dt  \Big] \\
\leq&  C_a + \int_0^T  4 c_{\gamma} \frac{1}{n^2}  \sum_{\substack{\{x,y\} \in \mcb F\\ |x-y| \geq \varepsilon n}} [   F^a \left(t, \tfrac{x}{n}, \tfrac{y}{n} \right) ]^2  	\Big| \tfrac{x}{n} - \tfrac{y}{n} \Big|^{-1-\gamma}  dt, \forall n \geq 1.
\end{align*}
Let $Q_{\varepsilon} = \{ (u,v) \in \mathbb{R}^2 : |u-v| \geq \varepsilon \}$. Taking the limit when $n \rightarrow \infty$ and recalling \eqref{eqlim} we have
\begin{align*}
 &c_{\gamma} \mathbb{E}_{\mathbb{Q}} \Big[   \int_0^T \iint_{Q_{\varepsilon}} F^a(t,u,v) |u-v|^{-1-\gamma} [ \rho(t,v) - \rho(t,u) ] du dv dt \Big]  \\
=& \lim_{n \rightarrow \infty}  \mathbb{E}_{\mu_n} \Big[     \int_0^T \frac{1}{n}  n^{\gamma}  \sum_{\substack{\{x,y\} \in \mcb F\\ |x-y| \geq \varepsilon n}}  F^a \left(t, \tfrac{x}{n}, \tfrac{y}{n} \right) p(y-x) [\eta_t^n(y) -  \eta_t^n(x) ] dt \Big]  \\
\leq &C_a + \lim_{n \rightarrow \infty} \Big\{     \int_0^T  4 c_{\gamma} \frac{1}{n^2}   \sum_{\substack{\{x,y\} \in \mcb F\\ |x-y| \geq \varepsilon n}} [   F^a \left(t, \tfrac{x}{n}, \tfrac{y}{n} \right)]^2  	\Big| \tfrac{x}{n} - \tfrac{y}{n} \Big|^{-1-\gamma}  dt \Big\}  \\
\leq&  C_a + c_{\gamma}\mathbb{E}_{\mathbb{Q}} \Big[  \int_0^T  \iint_{Q_{\varepsilon}} 4  [ F^a(t,u,v) ]^2 |u-v|^{-1-\gamma} du dv dt \Big], 
\end{align*}
which is the same as
\begin{align} \label{ineqF}
\mathbb{E}_{\mathbb{Q}} \Big[  \int_0^T \iint_{Q_{\varepsilon}} \Big\{  \dfrac{ [ \rho(t,v) - \rho(t,u) ]F^a(t,u,v)}{|u-v|^{1 + \gamma}}   - C_{0}  \frac{ [ F^a(t,u,v) ]^2}{|u-v|^{1+ \gamma}} \Big\}  du dv dt \Big] \leq \frac{C_a}{c_{\gamma}},
\end{align}
for $C_{0}=4$. 

\textbf{II.)} Now assume  $\beta=0$. Repeating the previous steps we are lead to 
\begin{align*}
& \mathbb{E}_{\mu_n} \Big[    \int_0^T \frac{1}{n}  n^{\gamma}  \sum_{|x-y| \geq \varepsilon n}  F^a \left(t, \tfrac{x}{n}, \tfrac{y}{n} \right) p(y-x) [\eta_t^n(y) -  \eta_t^n(x) ] dt  \Big] \\
\leq&  C_a + \int_0^T  4 (1 + \alpha^{-1} ) c_{\gamma} \frac{1}{n^2}  \sum_{ |x-y| \geq \varepsilon n} [   F^a \left(t, \tfrac{x}{n}, \tfrac{y}{n} \right) ]^2  	\Big| \tfrac{x}{n} - \tfrac{y}{n} \Big|^{-1-\gamma}  dt, \forall n \geq 1,
\end{align*}
and we conclude that \eqref{ineqF} holds with $C_0 = 4(1+\alpha^{-1})$.  For every $F$ on $C_{c}^{0,\infty} \big( (0,T) \times \mathbb{R}^2 \big)$, we have that 
\begin{align*}
&\iint_{Q_{\varepsilon}} \Big\{  \frac{[ \rho(t,v) - \rho(t,u) ]F(t,u,v)}{|u-v|^{1 + \gamma}}   - C_0  \frac{ [ F(t,u,v) ]^2}{|u-v|^{1 + \gamma}} \Big\} du dv \\
\leq &\iint_{Q_{\varepsilon}} \Big\{  \frac{[ \rho(t,v) - \rho(t,u) ]F^a(t,u,v)}{|u-v|^{1 + \gamma}}   - C_0 \frac{ [ F^a(t,u,v) ]^2}{|u-v|^{1 + \gamma}} \Big\} du dv,
\end{align*}
which leads to
\begin{align} \label{ineqvar}
& \mathbb{E}_{\mathbb{Q}} \Big[ \sup_F \Big\{  \int_0^T \iint_{Q_{\varepsilon}} \Big\{  \frac{ [ \rho(t,v) - \rho(t,u) ]F(t,u,v)}{|u-v|^{1 + \gamma}}   - C_0  \frac{ [ F(t,u,v) ]^2}{|u-v|^{1 + \gamma}} \Big\}  du dv dt \Big\}  \Big] \leq \frac{C_a}{c_{\gamma}}.
\end{align}
Above we made use of Lemma 7.5 in \cite{supinsexp} to insert the supremum (which is carried over $F$ on $C_{c}^{0,\infty} \big( (0,T) \times \mathbb{R}^2 \big)$ inside the expectation. Now we consider the Hilbert space $L^2 \left( \mathbb{R}^2, d \mu_{\varepsilon} \right)$, where $\mu_{\varepsilon}$ is the measure whose density with respect to the Lebesgue measure is given by $(u,v) \in \mathbb{R}^2 \rightarrow \mathbbm{1}_{ \{|u-v| \geq \varepsilon \} }|u-v|^{-1-\gamma}$. We can define $ \Pi: (0,T) \times \mathbb{R}^2$ by $
\Pi (t,u,v) = \pi(t,v) - \pi(t,u), \forall t \in I, \forall u,v \in \mathbb{R}$.

From \eqref{ineqvar}, the density of $C_{c}^{0,\infty} \big( (0,T) \times \mathbb{R}^2 \big)$ in $L^2 \left( \mathbb{R}^2, d \mu_{\varepsilon} \right)$ and Riesz's Representation Theorem, there exists $C_1>0$ such that
\begin{align*}
&  \mathbb{E}_{\mathbb{Q}} \Big[  \int_0^T \iint_{Q_{\varepsilon}}  \dfrac{ [ \rho(t,v) - \rho(t,u) ]^2 }{|u-v|^{1 + \gamma}}   du dv dt  \Big]  \leq C_1.
\end{align*}
 Letting $\varepsilon \rightarrow 0^+$, we conclude, from the Monotone Convergence Theorem, that
\begin{align*}
&\mathbb{E}_{\mathbb{Q}} \Big[  \int_0^T \iint_{\mathbb{R}^2} \frac{[ \bar{\rho}(t,v) - \bar{\rho}(t,u) ]^2}{|u-v|^{1+\gamma}} du dv dt  \Big]  = \mathbb{E}_{\mathbb{Q}} \Big[  \int_0^T \iint_{\mathbb{R}^2} \frac{[ \rho(t,v) - \rho(t,u) ]^2}{|u-v|^{1+\gamma}} du dv dt  \Big] < \infty.
\end{align*}
Combining this with Proposition \ref{estenergstat0}, we have the desired result.
\end{proof}
The proof of the next result is completely analogous  to the previous one and for that reason it will be omitted.
\begin{prop}  \label{estenergcombarlenfor} 
For all $\mcb S \subset \mcb S_0$, all $\gamma\in(0,2)$, for any $ \beta\geq 0$ and any  $\alpha>0$, we have that
\begin{align*}
\mathbb{Q} \Big( \bar{\rho}|_{[0,T] \times \mathbb{R}_{+}^{*}} \in L^2 \big(0, T ; \mcb{H}^{\frac{\gamma}{2}}( \mathbb{R}_{+}^{*})  \big) \Big) = \mathbb{Q} \Big(  \bar{\rho}|_{[0,T] \times \mathbb{R}_{-}^{*}} \in L^2 \big(0, T ; \mcb {H}^{\frac{\gamma}{2}}( \mathbb{R}_{-}^{*})  \big) \Big) = 1.
\end{align*}
In particular,
$
\mathbb{Q} \Big(  \bar{\rho} \in L^2 \big(0, T ; \mcb {H}^{\frac{\gamma}{2}}( \mathbb{R}^{*})  \big) \Big) = 1.
$
\end{prop}

\section{Useful $L^1(\mathbb{P}_{\mu_n})$ estimates} \label{secheurwithout}

 In this section, we assume that $\mcb S = \mcb S_0$ and $\gamma \in (1,2)$ and show some convergences in $L^1(\mathbb{P}_{\mu_n}$) that were used along the article. Recall \eqref{medemp}. The proof of the next result is analogous to the proof of Lemma 6.1 and Lemma 6.2 in \cite{casodif} and for that reason it will be omitted.
\begin{lem} \textbf{(One-block estimate)} \label{obe}
Let $\gamma\in(1,2)$.
For  $ \varepsilon >0$ and $ n \geq 1$, let $\ell_0=\ell_0(\varepsilon,n):=  \varepsilon n^{\gamma-1}$. Let $F \in L^{\infty}([0,T])$ and $\theta \in L^1(\mathbb{Z})$. Then, for every $t \in [0,T]$, 
\begin{equation} \label{obeleft}
\limsup_{\varepsilon \rightarrow 0^+} \limsup_{n \rightarrow \infty} \mathbb{E}_{\mu_n} \Big[\sup_{t \in [0,T]} \Big|  \int_0^t F(s) \sum_{z =-\infty}^{-1}  \theta(z) [\eta_{s}^{n}(z) - \eta_s^{\leftarrow\ell_0 }(0) ] ds \Big|   \Big] = 0,
\end{equation}
\begin{equation} \label{oberight}
\limsup_{\varepsilon \rightarrow 0^+} \limsup_{n \rightarrow \infty} \mathbb{E}_{\mu_n} \Big[ \sup_{t \in [0,T]} \Big| \int_0^t F(s)  \sum_{z =0}^{\infty}  \theta(z) [\eta_{s}^{n}(z) - \eta_s^{\rightarrow  \ell_0}(0) ] ds  \Big|    \Big] = 0
\end{equation}
and
\begin{equation} \label{obepright} 
\limsup_{\varepsilon \rightarrow 0^+} \limsup_{n \rightarrow \infty} \mathbb{E}_{\mu_n} \Big[ \sup_{t \in [0,T]} \Big| \int_0^t F(s)   [ \eta_s^{\rightarrow \ell_0 }(0) - \eta_{s}^{n}(0)] ds   \Big|  \Big] = 0.
\end{equation}
Moreover if $\beta \in [0, \gamma-1)$, we have
\begin{equation} \label{obepleft} 
\limsup_{\varepsilon \rightarrow 0^+} \limsup_{n \rightarrow \infty} \mathbb{E}_{\mu_n} \Big[\sup_{t \in [0,T]}  \Big|  \int_0^t F(s)   [ \eta_s^{\leftarrow \ell_0 }(0) - \eta_{s}^{n}(0)] ds  \Big|   \Big] = 0.
\end{equation}
\end{lem}
At a first glance the reader might be asking about the restriction on the parameter  $\beta$ appearing in last display. This restriction comes from the fact that  \eqref{obepleft}  involves the exchange of particles from  a subset of $\mathbb Z_-^{*}$ to the site $x=0$ and to do that one has to use the slow bonds. For that reason the error depends on the value of $\beta$ and for $\beta<\gamma-1$ this exchange can still be done. 

Recall the definition of $D_n$ given in \eqref{bound} and of $I_{x,y}$ given in\eqref{DnF}. The following lemma can be proved similarly to the proof of  Lemma 5.8 in \cite{stefano}.  
\begin{lem} \textbf{(Moving particle lemma)} \label{movpartlem}\\
Fix $n \geq 1$, $\ell_0 < n-1$ and $M \geq 1$ such that $2^{M} \ell_0 < n-1$. For $ i \in \{1, \ldots, M\}$, let $\ell_{i}:=2^{i} \ell_0$. Let $f$ be a density with respect to $\nu_{a}$. There exists $C_{mpl} >0$ such that
\begin{equation} \label{mpl1}
\sum_{i=1}^{M} \sum_{y=1}^{\ell_{i-1}} \frac{I_{y,y+\ell_{i-1}  }  (\sqrt{f}, \nu_{a} )}{(\ell_{i-1})^{\gamma}} \leq C_{mpl}  D_n  (\sqrt{f}, \nu_{a} )
\end{equation}
and
\begin{equation} \label{mpl2}
\sum_{i=1}^{M} \sum_{y=-\ell_{i-1}}^{-1} \frac{I_{y,y-\ell_{i-1}  }  (\sqrt{f}, \nu_{a} )}{(\ell_{i-1})^{\gamma}}  \leq C_{mpl}  D_n  (\sqrt{f}, \nu_{a} ).
\end{equation}
\end{lem}
\begin{proof}
We prove only \eqref{mpl1}, but we observe that the proof of \eqref{mpl2} is analogous. We can assume without loss of generality that $\ell_0$ is even (the argument is easy to extend to  $\ell_0$ odd) and as a consequence  $\ell_{i-1}$ is an even number for any $i \in \{1, \ldots,M\}$. Fix $i \in \{1, \ldots, M\}$. For every $y \in \{1, \ldots, \ell_{i-1}\}$ consider the $\frac{\ell_{i-1}}{2}$ possibilities for a  particle to jump from $y$ to $y + \ell_{i-1}$ with at most two steps. Hence for any $j \in \{1, \ldots, \frac{\ell_{i-1}}{2}\}$, define 
\begin{equation}
\label{eqzdefinition}
z_{0,j}:=z_{0,j}^i (y)=y,\quad z_{1,j}:=z_{1,j}^i (y) = y + \frac{\ell_{i-1}}{2}+j\quad z_{2,j}:=z_{2,j}^i (y)=y + \ell_{i-1}.
\end{equation} 
In order to simplify the notation, we will omit the index $i$ in $z_{0,j}$, $z_{1,j}$ and $z_{2,j}$. The sites defined in \eqref{eqzdefinition}
 correspond to one jump of length $\frac{\ell_{i-1}}{2}+j$ from $z_{0,j}=y$ to $z_{1,j}$  and one jump of length $\frac{\ell_{i-1}}{2}-j$ from $z_{1,j}$ to $z_{2,j}=y+\ell_{i-1}$. Observe that for $j=\frac{\ell_{i-1}}{2}$, there  is only one jump  from $y$ to $y+\ell_{i-1}$ so that $z_{1,j} = z_{2,j}$. 
Observe that
\begin{equation}\label{problem_cedric}
\sqrt{f \left( \eta^{y,y+\ell_{i-1}}  \right) } - \sqrt{f \left( \eta \right) } = \sqrt{f \left( \eta^{z_{0,j},z_{2,j}}  \right) } - \sqrt{f \left( \eta \right) }  
\end{equation}
is non zero if, and only if, $\eta(z_{0,j})\neq \eta(z_{2,j})$. We want to rewrite \eqref{problem_cedric} using the intermediate point $z_{1,j}$. To do that we consider separately the possible combinations of values of the $z_{q,j}$ with $p\in \{q,1,2\}$.  First, assume that $\eta(z_{0,j})=1$ and $\eta(z_{2,j})=0$. In this case we have two possibilities:
\begin{enumerate}
	\item[a)] when $\eta(z_{1,j})=0$, we observe that $\eta^{z_{0,j},z_{2,j}} = (\eta^{z_{0,j},z_{1,j}} )^{z_{1,j},z_{2,j}}$. So, in this particular case we can write \eqref{problem_cedric} as 
	\begin{equation}\label{list1}
 \Big[ \sqrt{f \big( (\eta^{z_{0,j},z_{1,j}} )^{z_{1,j},z_{2,j}} \big) } - \sqrt{f \left( \eta^{z_{0,j},z_{1,j}} \right) } \Big] + \Big[\sqrt{f (  \eta^{z_{0,j},z_{1,j}} ) } - \sqrt{f ( \eta ) } \Big].\end{equation}
	\item[b)] when $\eta(z_{1,j})=1$, we observe that $\eta^{z_{0,j},z_{2,j}} = (\eta^{z_{1,j},z_{2,j}} )^{z_{0,j},z_{1,j}}$. So, in this particular case we can write \eqref{problem_cedric} as
	 \begin{equation}\label{list2}
	 \Big[\sqrt{f \big( (\eta^{z_{1,j},z_{2,j}} )^{z_{0,j},z_{1,j}} \big) } - \sqrt{f \left(\eta^{z_{1,j},z_{2,j}} \right) } \Big] + \Big[\sqrt{f (  \eta^{z_{1,j},z_{2,j}} ) } - \sqrt{f ( \eta ) } \Big].\end{equation}
\end{enumerate} 
 We also have to consider the case $\eta(z_{1,j})=0$ and $\eta(z_{2,j})=1$. Reasoning similarly to what we did above, if  $\eta(z_{1,j})=0$ we can rewrite \eqref{problem_cedric} as \eqref{list2}, otherwise, if $\eta(z_{1,j})=1$ we can rewrite \eqref{problem_cedric} as \eqref{list1}. Let us now consider the following sets of  configurations:
\begin{equation*}
\begin{split}
\tilde \Omega_{i,y,j}^1 = \{ \eta \in \Omega: &\eta(z_{0,j})=1,\eta(z_{1,j})=0,\eta(z_{2,j})=0 
\text{ or } \eta(z_{0,j})=0,\eta(z_{1,j})=1,\eta(z_{2,j})=1\}
\end{split}
\end{equation*}
\begin{equation*}
\begin{split}
\tilde \Omega_{i,y,j}^2= \{\eta \in \Omega: &\eta(z_{0,j})=1,\eta(z_{1,j})=1,\eta(z_{2,j})=0 \text{ or } \eta(z_{0,j})=0,\eta(z_{1,j})=0,\eta(z_{2,j})=1\}.
\end{split}
\end{equation*}
 Observe that $\tilde \Omega_{i,y,j}^1$ and $ \tilde \Omega_{i,y,j}^2$ are disjoint sets. Now, thanks to the reasoning that we did above and using the inequality $(a+b)^2 \le 2(a^2+b^2)$, we can write
\begin{align*}
&I_{y,y+\ell_{i-1}}  (\sqrt{f}, \nu_{a} )= \int \big(\sqrt{f \left( \eta^{y,y+\ell_{i-1}} \right) } - \sqrt{f \left( \eta \right) } \big)^2 d \nu_{a}  \\
& \lesssim  \int_{\tilde \Omega_{i,y,j}^1}  \Big( \sqrt{f \big( (\eta^{z_{0,j},z_{1,j}} )^{z_{1,j},z_{2,j}} \big) } - \sqrt{f \left( \eta^{z_{0,j},z_{1,j}} \right) } \Big)^2 d\nu_{a}  + \int_{\tilde \Omega_{i,y,j}^1}\left(\sqrt{f (  \eta^{z_{0,j},z_{1,j}} ) } - \sqrt{f ( \eta ) } \right)^2 d\nu_{a}  \\
 &+  \int_{\tilde \Omega_{i,y,j}^2}  \Big( \sqrt{f \big( (\eta^{z_{1,j},z_{2,j}} )^{z_{0,j},z_{1,j}} \big) } - \sqrt{f \left( \eta^{z_{1,j},z_{2,j}} \right) } \Big)^2 d\nu_{a}  + \int_{\tilde \Omega_{i,y,j}^2}\left(\sqrt{f (  \eta^{z_{1,j},z_{2,j}} ) } - \sqrt{f ( \eta ) } \right)^2 d\nu_{a}  \\
 & \leq  \int  \Big( \sqrt{f \big( (\eta^{z_{0,j},z_{1,j}} )^{z_{1,j},z_{2,j}} \big) } - \sqrt{f \left( \eta^{z_{0,j},z_{1,j}} \right) } \Big)^2 d\nu_{a}  + \int\left(\sqrt{f (  \eta^{z_{0,j},z_{1,j}} ) } - \sqrt{f ( \eta ) } \right)^2 d\nu_{a}  \\
 &+  \int  \Big( \sqrt{f \big( (\eta^{z_{1,j},z_{2,j}} )^{z_{0,j},z_{1,j}} \big) } - \sqrt{f \left( \eta^{z_{1,j},z_{2,j}} \right) } \Big)^2 d\nu_{a}  + \int\left(\sqrt{f (  \eta^{z_{1,j},z_{2,j}} ) } - \sqrt{f ( \eta ) } \right)^2 d\nu_{a}. 
 \end{align*}
Since $\nu_a(\eta^{x,y})=\nu_a(\eta)$ for every $x,y \in \mathbb{Z}$ and every $\eta \in \Omega$, last display can be bounded from above by a constant times
 \begin{equation*}
I_{z_{1,j},z_{2,j}}  (\sqrt{f}, \nu_{a})+I_{z_{0,j},z_{1,j}}  (\sqrt{f}, \nu_{a} )+ I_{z_{0,j},z_{1,j}}  (\sqrt{f}, \nu_{a} )+ I_{z_{1,j},z_{2,j}}  (\sqrt{f}, \nu_{a}).
 \end{equation*}
 Observe now that by construction we have $[p (z_{k,j} - z_{k-1,j})]^{-1}  \lesssim \ell_{i-1}^{1+\gamma}$ (the longest possible jump has size at most  $\ell_{i-1}$). Hence, we have
\begin{equation*}
 I_{y,y+\ell_{i-1}}  (\sqrt{f}, \nu_{a} )  \lesssim \ell_{i-1}^{1+\gamma} \sum_{k=1}^2   p (z_{k,j} - z_{k-1,j})I_{z_{k-1,j},z_{k,j}}  (\sqrt{f}, \nu_{a}).
\end{equation*}
Since last inequality  is true for any $j \in \{1,\dots, \tfrac{\ell_{i-1}}{2}\}$,  we can write
\begin{equation*}
\ell_{i-1}I_{y,y+\ell_{i-1}}  (\sqrt{f}, \nu_{a} )	\lesssim  \ell_{i-1}^{1+\gamma}  \sum_{j=1}^{\ell_{i-1}/2} \sum_{k=1}^2   p (z_{k,j} - z_{k-1,j})I_{z_{k-1,j},z_{k,j}}  (\sqrt{f}, \nu_{a} ),
\end{equation*}
which implies that
\begin{equation}
\label{Iprev}
\sum_{i=1}^M \sum_{y=1}^{\ell_{i-1}}  \frac{I_{y,y+\ell_{i-1}}(\sqrt{f}, \nu_{a} )}{\ell_{i-1}^{\gamma}} \lesssim  \sum_{i=1}^M \sum_{y=1}^{\ell_{i-1}} \sum_{j=1}^{\ell_{i-1}/2} \sum_{k=1}^2   p (z_{k,j} - z_{k-1,j})I_{z_{k-1,j},z_{k,j}}  (\sqrt{f}, \nu_{a} ).
\end{equation}
Recall  that the $z_{k,j}$'s depend in fact on $i$ and $y$. We claim that when $i,y,j,k$ describe the sets involved in last sum, the pairs $(z_{k-1,j} , z_{k,j}):=(z_{k-1, j}^i (y), z_{k,j}^i (y) )$ are all different, i.e.
\begin{equation}
\label{eq:phi}
\Phi: (i,y,j,k) \to (z_{k-1, j}^i (y), z_{k,j}^i (y) ) \in \mathbb{Z}_{-}^{*}   \times \mathbb{Z}_{-}^{*} \quad \text{is injective.}
\end{equation}
Therefore, recalling \eqref{DnF}, we can bound from above the term on the right-hand side of \eqref{Iprev} by
\begin{equation}\label{64}
\underset{v\leq w}{ \sum_{v,w \in \mathbb{Z}_{-}^{*}}} p(w-v) I_{v,w} (\sqrt{f}, \nu_{\rho(\cdot)}^N ) \lesssim   D_n (\sqrt{f},\nu_{a}).
\end{equation}
Putting together \eqref{64} and \eqref{Iprev} we get the statement. We still have to prove \eqref{eq:phi} to conclude the proof. Let us assume that
$\Phi (i,y,j,k) = \Phi (i',y',j',k')
$
and let us prove that $(i,y,j,k)=(i',y',j',k')$. We distinguish four cases according to the values of $k$ and $k'$.

\begin{itemize}
\item [i)]$k=k'=1$: then $z_{0,j}^i (y)=z_{0,j'}^{i'} (y')$ and $z_{1,j}^i (y)=z_{1,j'}^{i'} (y')$ imply that 
$y=y', \quad \tfrac{\ell_{i-1}}{2} +j = \tfrac{\ell_{i'-1}}{2} +j' .$
Since $1 \le j \le \ell_{i-1} /2$ and $1 \le j' \le \ell_{i'-1}/2$ we have that
\begin{equation}
\label{eq:k11}
1 + \tfrac{\ell_{i-1}}{2} \le \tfrac{\ell_{i-1}}{2} +j \le \ell_{i-1}\quad \text{and}\quad 1 + \tfrac{\ell_{i'-1}}{2} \le \tfrac{\ell_{i'-1}}{2} +j' \le \ell_{i'-1}.
\end{equation}
If $i\le i' -1$ then $\ell_{i-1} \le \tfrac{\ell_{i' -1}}{2}<  \tfrac{\ell_{i' -1}}{2} +1$ and the equality $\tfrac{\ell_{i-1}}{2} +j = \tfrac{\ell_{i'-1}}{2} +j'$ is then in contradiction with \eqref{eq:k11}. If $i'\le i -1$ then $\ell_{i'-1} \le \tfrac{\ell_{i -1}}{2}<  1+ \tfrac{\ell_{i -1}}{2}$ the equality $\tfrac{\ell_{i-1}}{2} +j = \tfrac{\ell_{i'-1}}{2} +j'$ is then again in contradiction with \eqref{eq:k11}. Hence $i=i'$ and consequently $j=j'$ and we are done.

\item  [ii)] $k=1$ and $k'=2$: then $z_{0,j}^i (y)=z_{1,j'}^{i'} (y')$ and $z_{1,j}^i (y)=z_{2,j'}^{i'} (y')$ imply that 
$
y= y'+ \tfrac{\ell_{i'-1}}{2} +j', \quad  y+ \tfrac{\ell_{i-1}}{2} +j =y'+ \ell_{i'-1} ,
$
and hence by replacing $y$ in the second equality by $ y'+ \tfrac{\ell_{i'-1}}{2} +j'$, we get
$y= y'+ \tfrac{\ell_{i'-1}}{2} +j', \quad  \tfrac{\ell_{i-1}}{2} +j =\tfrac{\ell_{i'-1}}{2} -j'. $
Since $1\le y \le \ell_{i-1}$, $1\le y'\le \ell_{i'-1}$, $1 \le j \le \ell_{i-1} /2$ and $1 \le j' \le \ell_{i'-1}/2$ we have that
\begin{equation}
\label{eq:k12}
\begin{split}
& 1\le y \le \ell_{i-1} \quad \text{and} \quad 2 + \tfrac{\ell_{i'-1}}{2} \le y \le 2 \ell_{i'-1} , \\
& 1 + \tfrac{\ell_{i-1}}{2} \le \tfrac{\ell_{i-1}}{2} +j \le \ell_{i-1} \quad \text{and} \quad 0 \le  \tfrac{\ell_{i-1}}{2} +j \le \tfrac{\ell_{i'-1}}{2} -1 .
\end{split}
\end{equation}
If $i\le i'-1$ then $\ell_{i-1} \le \tfrac{\ell_{i' -1}}{2}<2 + \tfrac{\ell_{i'-1}}{2}$ and there is a contradiction with the  first line of \eqref{eq:k12}. If $i'\le i$ then $\ell_{i'-1} \le {\ell_{i -1}}$, hence $\tfrac{\ell_{i'-1}}{2} -1 < 1 + \tfrac{\ell_{i-1}}{2}$, which is in contradiction with the second line of \eqref{eq:k12}. Hence this case is not possible and we are done.

 \item [iii)] $k=2$ and $k'=1$: by symmetry this case is equivalent to the previous one.

 \item [iv)] $k'=2$ and $k=2$:  then $z_{1,j}^i (y)=z_{1,j'}^{i'} (y')$ and $z_{2,j}^i (y)=z_{2,j'}^{i'} (y')$ imply that 
$$\tfrac{\ell_{i-1}}{2} +j +y  = \tfrac{\ell_{i'-1}}{2} +j' +y', \quad {\ell_{i-1}}+y  = {\ell_{i'-1}} +y' .$$ 
The second equality and the fact that $1 \le y \le \ell_{i-1}$, resp. $1\le y' \le \ell_{i'-1}$ implies that
\begin{equation}
\label{eq:k22}
1+\ell_{i-1} \le y +\ell_{i-1} \le 2\ell_{i-1} \quad \text{and}\quad 1+\ell_{i'-1} \le y +\ell_{i-1} \le 2\ell_{i'-1}.
\end{equation}
If $i\le i'-1$ then $2\ell_{i-1} \le {\ell_{i' -1}}< {\ell_{i' -1}}+1 $ which is in contradiction with \eqref{eq:k22}. Similarly if $i'\le i-1$ then $2\ell_{i'-1} \le {\ell_{i -1}}< {\ell_{i -1}}+1 $ is in contradiction with \eqref{eq:k22}. Hence $i=i'$ and consequently we deduce that $y=y'$ and $j=j'$.
\end{itemize}
This concludes the proof of the lemma.
\end{proof}

\begin{lem} \textbf{(Two-blocks estimate)} \label{tbe} Let $\gamma>1$.
For $\varepsilon >0$ and for $ n \geq 1$ let $\ell_0=\ell_0(\varepsilon,n):=  \varepsilon n^{\gamma-1},$  and $F \in L^{\infty}([0,T])$. For every $t \in [0,T]$, 
\begin{equation} \label{rlleft2}
\limsup_{\varepsilon \rightarrow 0^+} \limsup_{n \rightarrow \infty} \mathbb{E}_{\mu_n} \Big[ \sup_{t \in [0,T]}  \Big|  \int_0^t F(s)  [\eta_s^{\leftarrow\ell_0 }(0) - \eta_s^{\leftarrow \varepsilon n }(0)  ] ds  \Big|   \Big] = 0,
\end{equation}
\begin{equation} \label{rlright2}
\limsup_{\varepsilon \rightarrow 0^+} \limsup_{n \rightarrow \infty} \mathbb{E}_{\mu_n} \Big[ \sup_{t \in [0,T]}  \Big| \int_0^t F(s)   [ \eta_s^{\rightarrow  \ell_0}(0) - \eta_s^{\rightarrow  \varepsilon n}(0)  ] ds  \Big|    \Big] = 0.
\end{equation}
	\end{lem}
\begin{proof}
We present the proof of \eqref{rlright2}, but 
 the proof of \eqref{rlleft2} is analogous. Recall that $C_a >0$ is such that $H( \mu_n | \nu_a) \leq  C_a n, \forall n \geq 1$. 
By the entropy and Jensen's inequalities, and Feynman-Kac's formula, we  bound the expectation in \eqref{rlright2} from above by
\begin{align} \label{suptbe}
& \frac{C_a}{D} +  T \sup_{f} \Big\{  \| F \|_{\infty}  | \langle  \eta^{\rightarrow  \ell_0}(0) -  \eta^{\rightarrow  \varepsilon n}(0) , f \rangle_{\nu_{a}} |  + \frac{n^{\gamma-1}}{D}  \langle \mcb L_n \sqrt{f}, \sqrt{f} \rangle_{\nu_{a}}  \Big\},  
\end{align}
where $D>0$ and the supremum is carried over all the densities $f$ with respect to $\nu_{a}$. Let  $\ell_i:=2^i \ell_0$ and $M=(2-\gamma) \frac{\log(n)}{\log(2)}$, so that $2^M\ell_0=\epsilon n$. We can write
\begin{align*} 
\eta^{\rightarrow  \ell_0}(0) -  \eta^{\rightarrow  \varepsilon n}(0) =& \sum_{i=1}^{M} \Big[  \frac{1}{\ell_{i-1}} \sum_{y=1}^{\ell_{i-1}} \eta(y) -  \frac{1}{\ell_{i}} \sum_{y=1}^{\ell_{i}} \eta(y)   \Big] = \sum_{i=1}^{M} \frac{1}{\ell_i} \sum_{y=1}^{\ell_{i-1}} [\eta(y) - \eta(y+\ell_{i-1})]. 
\end{align*}
From this we get 
\begin{align} \label{medempgen}
\langle  \eta^{\rightarrow  \ell_0}(0) -  \eta^{\rightarrow  \varepsilon n}(0) , f \rangle_{\nu_{a}} = \sum_{i=1}^{M} \frac{1}{\ell_i} \sum_{y=1}^{\ell_{i-1}} \int    [\eta(y) - \eta(y+\ell_{i-1})]  f(\eta) d \nu_{a}.
\end{align}
Recall \eqref{young}. For every $i=1, \ldots,M$ and every $y=1, \ldots, \ell_{i-1}$, we will choose
\begin{align*}
A_{y,y+\ell_{i-1}}= \frac{C_{mpl} D (\ell_{i-1})^{\gamma} \| F \|_{\infty} }{n^{\gamma-1} \ell_i} >0.
\end{align*}
From \eqref{medempgen}, \eqref{young} and Lemma \ref{movpartlem}, we can bound the expression inside the supremum in \eqref{suptbe} by
\begin{align*}
 & \| F \|_{\infty}   \sum_{i=1}^{M} \frac{1}{\ell_i} \sum_{y=1}^{\ell_{i-1}} \Big| \int    [\eta(y) - \eta(y+\ell_{i-1})]  f(\eta) d \nu_{a} \Big| - \frac{n^{\gamma-1}}{2D}  D_n(\sqrt{f}, \nu_{a} ) \\
\leq &\| F \|_{\infty}  \sum_{i=1}^{M} \frac{1}{\ell_i} \sum_{y=1}^{\ell_{i-1}} \Big[ \frac{n^{\gamma-1} \ell_i I_{y,y+\ell_{i-1}  }  (\sqrt{f}, \nu_{a} )}{2 C_{mpl} D (\ell_{i-1})^{\gamma} \| F \|_{\infty} } + 2 \frac{C_{mpl} D (\ell_{i-1} )^{\gamma} \| F \|_{\infty} }{n^{\gamma-1} \ell_i} \Big] - \frac{n^{\gamma-1}}{2D}  D_n(\sqrt{f}, \nu_{a} ) \\
\leq & \frac{C_{mpl} D \| F \|_{\infty}  (2^{M} \ell_0)^{\gamma-1}}{2^{\gamma} n^{\gamma-1}} =   \frac{C_{mpl} \| F \|_{\infty}  }{2^{\gamma}} D \varepsilon^{\gamma-1}.
\end{align*}
Choosing $D= \varepsilon^{\frac{1-\gamma}{2}} >0$, we can bound \eqref{suptbe} by  $( C_a + T\frac{C_{mpl} \| F \|_{\infty}  }{2^{\gamma}}) \varepsilon^{\frac{\gamma-1}{2}} $, that vanishes as $\epsilon \to 0$ since $\gamma>1$.
\end{proof}

Combining Lemma \ref{obe} and Lemma \ref{tbe} we get the next result. 
\begin{lem} \textbf{(Replacement Lemma)} \label{replemma} Let $\gamma>1$.
Let $F \in L^{\infty}([0,T])$ and $\theta \in L^1(\mathbb{Z})$. Then for every $t \in [0,T]$, 
\begin{equation*} 
\limsup_{\varepsilon \rightarrow 0^+} \limsup_{n \rightarrow \infty} \mathbb{E}_{\mu_n} \Big[ \sup_{t \in [0,T]}  \Big|  \int_0^t F(s) \sum_{z =-\infty}^{-1}  \theta(z) [\eta_{s}^{n}(z) - \eta_s^{\leftarrow\varepsilon n }(0) ] ds  \Big|   \Big] = 0,
\end{equation*}
\begin{equation*} 
\limsup_{\varepsilon \rightarrow 0^+} \limsup_{n \rightarrow \infty} \mathbb{E}_{\mu_n} \Big[ \sup_{t \in [0,T]}  \Big| \int_0^t F(s)  \sum_{z =0}^{\infty}  \theta(z) [\eta_{s}^{n}(z) - \eta_s^{\rightarrow  \varepsilon n}(0) ] ds  \Big|   \Big] = 0
\end{equation*}
and
\begin{equation*}  
\limsup_{\varepsilon \rightarrow 0^+} \limsup_{n \rightarrow \infty} \mathbb{E}_{\mu_n} \Big[ \sup_{t \in [0,T]}  \Big| \int_0^t F(s)  [ \eta_s^{\rightarrow \varepsilon n }(0) - \eta_{s}^{n}(0)] ds  \Big|  \Big] = 0.
\end{equation*}
Moreover if $\beta \in [0, \gamma-1)$, we have
\begin{equation*}  
\limsup_{\varepsilon \rightarrow 0^+} \limsup_{n \rightarrow \infty} \mathbb{E}_{\mu_n} \Big[  \sup_{t \in [0,T]} \Big|  \int_0^t F(s)   [ \eta_s^{\leftarrow \varepsilon n }(0) - \eta_{s}^{n}(0)] ds  \Big|  \Big] = 0.
\end{equation*}
\end{lem}

We end this section with an useful application of the Replacement Lemma, that will be useful to treat \eqref{srobterm}. The proof of the next result is analogous to the proof of Proposition 6.4 in \cite{casodif} and for that reason it will be omitted.
\begin{prop} \label{convrob}
Assume $\mcb S= \mcb S_0$, $\gamma \in (1,2)$ and $ \beta \geq \gamma - 1$. Let $t \in [0,T]$ and  $G \in \mcb S_{\textrm{Rob}}$. Then, 
\begin{align}
 \limsup_{\varepsilon \rightarrow 0^+} \limsup_{n \rightarrow \infty} \mathbb{E}_{\mu_n} \Big[& \sup_{t \in [0,T]} \Big| \int_{0}^{t} \Big\{  \alpha n^{\gamma-1-\beta}  \sum_{\{x,z\} \in \mcb S}  [G(s,\tfrac{x}{n}) - G(s,\tfrac{z}{n}) ]  p(x-z)  \eta_s^n(z) \nonumber \\
-& \mathbbm{1}_{\beta=\gamma-1} m \alpha [G(s,0^{-})- G(s,0^{+})]  [ \eta_s^{\rightarrow n \varepsilon}(0) - \eta_s^{\leftarrow n \varepsilon}(0)] \Big\}   ds \Big| \Big] =0. \label{eqconvrob01}
\end{align}
\end{prop}

	\appendix

	\section{Convergences at the discrete level} \label{secdiscconv}
	\label{secuseres}

In this section we present several propositions which were used along the article and which allowed us  treating  the integral term in Dynkin's martingale given in Proposition \ref{gensdif}.
The first two result we present are useful to treat \eqref{princsdif} and justify the choice  $n^{\gamma}$ for the time scale.
\begin{prop} \label{convdisc}
For every $\gamma \in (0,2)$ and $G \in \mcb S_{\textrm{Dif}}$, it holds
\begin{align*}
\lim_{n \rightarrow \infty} \frac{1}{n}   \sum_{x }  \sup_{s \in [0,T]}|n^{\gamma} \mcb{K}_n G \left(s,\tfrac{x}{n} \right)  -[-(- \Delta)^{\frac{\gamma}{2}}   G]  \left(s,\tfrac{x}{n} \right) | =0.
\end{align*}
\end{prop}
\begin{proof}
Recall \eqref{eq:B_g}. Let $b \geq \max \{1, b_G\}$. We prove the result in three steps. 

\textbf{I).} First we  claim that
\begin{equation} \label{lim1lapfrac}
\lim_{n \rightarrow \infty} \frac{1}{n} \sum_{|x| > 2bn}\sup_{s \in [0,T]} | n^{\gamma} \mcb{K}_n G \left(s, \tfrac{x}{n} \right) -  [-(- \Delta)^{\frac{\gamma}{2}} G ] (s, \tfrac{x}{n}) | =0.
\end{equation}
We prove last result when  the sum is  restricted to $x\geq 2bn+1$, but in the other case it is completely analogous. 
 Since $b \geq b_G$, then $\sup_{s \in [0,T]} G(s, u)=0$ if $|u| \geq b$ and
\begin{align} \label{lapfrac1}
& \frac{1}{n} \sum_{x=2bn+1}^{\infty} \sup_{s \in [0,T]}  | n^{\gamma} \mcb{K}_n G \left(s,\tfrac{x}{n} \right) -  [-(- \Delta)^{\frac{\gamma}{2}} G ] (s,\tfrac{x}{n}) | \nonumber \\
=& \frac{c_{\gamma}}{n} \sum_{x=2bn+1}^{\infty} \sup_{s \in [0,T]}  \Big| n^{\gamma} \sum_{y=-bn}^{bn-1} G  \left(s,\tfrac{y}{n} \right) [x-y]^{-\gamma-1}  - \int_{-b}^{b} G(s, u) [ \tfrac{x}{n} -u ]^{-\gamma-1} du  \Big|  \nonumber \\
=& \frac{c_{\gamma}}{n} \sum_{x=2bn+1}^{\infty} \sup_{s \in [0,T]} \Big|  \sum_{y =-bn}^{bn-1} \int_{\tfrac{y}{n}}^{\frac{y+1}{n}} [f_x^n(s,\tfrac{y}{n}) - f_x^n(s,v)]  dv \Big|,
\end{align}
where for $x \geq 2bn +1$, we define $f^n_x:[0,T]\times (-b-1, b+1) \rightarrow \mathbb{R}$ by
$
f_x^n(s,u):= G(s,u)[\tfrac{x}{n}-u]^{-\gamma-1}.
$
Since $b \geq 1$, for $x \geq 2 bn +1$, $- bn \leq y \leq bn-1$, $v \in ( \tfrac{y}{n} , \tfrac{y+1}{n})$, $s \in [0,T]$ and $u \in ( \tfrac{y}{n} , v)$, we get 
\begin{align*}
|\tfrac{x}{n}-u| \geq |\tfrac{x}{n}-b| \geq |\tfrac{2bn+1}{n}- b| \geq b \geq 1,
\end{align*}
which leads to
\begin{align*}
| \partial_u f_x^n(s,u) |=| \partial_u G(s,u)[\tfrac{x}{n}-u]^{-\gamma-1} - (\gamma+1) G(s,u)[\tfrac{x}{n}-u]^{-\gamma-2} | \lesssim [\tfrac{x}{n}-u]^{-\gamma-1}  \leq [\tfrac{x}{n}-b]^{-\gamma-1}. 
\end{align*}
Then Applying the Mean Value Theorem to $f_x^n$, we have
\begin{align} \label{MVTA1}
\sup_{s \in [0,T]} |f_x(s,\tfrac{y}{n}) - f_x(s,v)| \lesssim (v - \tfrac{y}{n})  [\tfrac{x}{n}-b]^{-\gamma-1},
\end{align}
when $ x \geq 2bn+1$, $-bn \geq y \leq bn-1$ and $v \in ( \tfrac{y}{n} , \tfrac{y+1}{n})$.

With the triangle inequality, we can bound the right-hand side of \eqref{lapfrac1} by a constant times
\begin{align*}
   \frac{1}{n} \sum_{x=2bn+1}^{\infty} [\tfrac{x}{n}-b]^{-\gamma-1}  \sum_{y =-bn}^{bn-1} \frac{1}{2n^2} \lesssim \frac{b}{n} \int_{b}^{\infty} u^{-\gamma-1} du = \frac{b^{1-\gamma}}{\gamma n}
\end{align*}
which goes to zero as $n \rightarrow \infty$, proving the first claim.

\textbf{II).} 
Now we will prove the same result as in \eqref{lim1lapfrac} but with the sum restricted to $|x|\leq 2bn$. This will be done in a two step procedure using  \eqref{limlapfrac} for  $\varepsilon >0$ fixed  and later on we will take $\varepsilon \rightarrow 0^+$. First, we claim that
\begin{equation} \label{limlapfrac2}
\lim_{\varepsilon \rightarrow 0^{+}}  \lim_{n \rightarrow \infty}  \frac{1}{n} \sum_{|x|\leq   2 bn} \sup_{s \in [0,T]} | n^{\gamma} \mcb {K}_{n} G (s,\tfrac{x}{n} ) - [-(- \Delta)^{\frac{\gamma}{2}} G ]_{\varepsilon} (s,\tfrac{x}{n} ) | =0.
\end{equation}
By defining for $u \in  \mathbb{R}$, the function $\theta_u: [0,T] \times \mathbb{R} \rightarrow \mathbb{R}$ 
\begin{align}~\label{eq:function_theta}
\theta_u(s,w) := G(s,u-w) + G(s,u+w) -2 G(s,u), \forall (s,w) \in [0,T] \times \mathbb{R},
\end{align} 
 assuming, without loss of generality,  $n > 2 \varepsilon^{-1}$ and recalling  \eqref{limlapfrac} and \eqref{op_Kn}, we can rewrite
\begin{align}
& [-(- \Delta)^{\frac{\gamma}{2}} G ]_{\varepsilon} (s,u ) =  c_{\gamma} \int_{ \varepsilon}^{\infty} \frac{\theta_u(s,w)}{ w^{\gamma+1}} dw, \label{lapfractheta}
\end{align}
and 
\begin{align*}
  n^{\gamma} \mcb {K}_{n} G (s,\tfrac{x}{n} )  = n^{\gamma} \sum_{t =1}^{\infty} p(t) \theta_{\tfrac{x}{n}} (s, \tfrac{t}{n} ) = n^{\gamma}  \sum_{t =1}^{n \varepsilon -1} p(t) \theta_{\tfrac{x}{n}} (s, \tfrac{t}{n} ) + n^{\gamma}  \sum_{t = n \varepsilon}^{\infty} p(t) \theta_{\tfrac{x}{n}} ( s,\tfrac{t}{n} ). 
\end{align*}
From this we have 
\begin{align}
& \frac{1}{n} \sum_{|x|\leq   2 bn} \sup_{s \in [0,T]} | n^{\gamma} \mcb{K}_n  G (s, \tfrac{x}{n} ) -  [-(- \Delta)^{\frac{\gamma}{2}} G ]_{\varepsilon} (s, \tfrac{x}{n} ) | \nonumber  \\ 
\leq & \frac{c_{\gamma}}{n} \sum_{|x|\leq   2 bn} \sup_{s \in [0,T]} \Big| n^{\gamma}  \sum_{t =1}^{ \varepsilon n -1} t^{-\gamma-1} \theta_{\tfrac{x}{n}} (s, \tfrac{t}{n} )  \Big| + \frac{c_{\gamma}}{n} \sum_{|x|\leq   2 bn}  \sup_{s \in [0,T]} \Big|  n^{\gamma}  \sum_{t =  \varepsilon n}^{\infty} t^{-\gamma-1} \theta_{\tfrac{x}{n}} (s, \tfrac{t}{n} ) -   \int_{ \varepsilon}^{\infty} \frac{\theta_{\tfrac{x}{n}}(s, w)}{ w^{\gamma+1}} dw \Big|. \label{eqlapfrac2}
\end{align}
Performing two Taylor expansions of second order on the function $\theta_u$, we get
\begin{equation} \label{thetabounddelta}
| \theta_z(s, w)| \leq w^2 \| \Delta G \|_{\infty}, \forall z,w \in \mathbb{R}, \forall s \in [0,T].
\end{equation}
The leftmost term in \eqref{eqlapfrac2} can be bounded from above by  a constant times 
$\| \Delta G \|_{\infty}\varepsilon^{2-\gamma} $, which vanishes as $\epsilon\to 0$, since $\gamma<2$.
It remains to  treat the rightmost term in \eqref{eqlapfrac2}. From the triangular inequality, we get
\begin{align}
&  \Big|  n^{\gamma}  \sum_{t =  \varepsilon n}^{\infty} t^{-\gamma-1} \theta_{\tfrac{x}{n}} (s, \tfrac{t}{n} ) -   \int_{ \varepsilon}^{\infty} \frac{\theta_{\tfrac{x}{n}}(s,w)}{ w^{\gamma+1}} dw \Big| 
    =  \Big|    \sum_{t =  \varepsilon n}^{\infty} \int_{ \tfrac{t}{n} }^{\tfrac{t+1}{n}} \Big[  \Big( \tfrac{t}{n} \Big)^{- \gamma -1}  \theta_{\tfrac{x}{n}} ( s,\tfrac{t}{n} )  -    w^{-\gamma-1} \theta_{\tfrac{x}{n}}(s,w) \Big] dw \Big| \nonumber \\
\leq &  \sum_{t =  \varepsilon n}^{\infty} \int_{ \tfrac{t}{n} }^{\tfrac{t+1}{n}} | \theta_{\tfrac{x}{n}} (s, \tfrac{t}{n} ) | \Big[  \Big( \tfrac{t}{n} \Big)^{- \gamma -1}    -    w^{-\gamma-1} \Big] dw   +   \sum_{t =  \varepsilon n}^{\infty} \int_{ \frac{t}{n} }^{\frac{t+1}{n}} w^{-\gamma-1} |   \theta_{\tfrac{x}{n}} (s, \tfrac{t}{n} ) -     \theta_{\tfrac{x}{n}}(s,w) | dw . \label{eqlapfrac3}
\end{align}
Now we observe that $| \theta_z(s,w)| \leq 4 \| G \|_{\infty}, \forall z,w \in \mathbb{R}$. Moreover, from the mean value Theorem we have that
$
| \theta_z(s,y) - \theta_z(s, w)| \leq 2 \| \partial_u G \|_{\infty} |y-w|, \forall z,w \in \mathbb{R}, \forall s \in [0,T],
$
and 
\begin{align} \label{mvt2}
 \int_{ \tfrac{t}{n} }^{\tfrac{t+1}{n}}  \Big[  \Big( \tfrac{t}{n} \Big)^{- \gamma -1}    -    w^{-\gamma-1} \Big] dw \leq ( \gamma + 1) n^{\gamma}  t^{-\gamma-2}, \forall t \in \mathbb{N}.
\end{align}
From this we bound \eqref{eqlapfrac3}  from above by
\begin{align*}
&  \sum_{t =  \varepsilon n}^{\infty} \int_{ \tfrac{t}{n} }^{\tfrac{t+1}{n}} 4 \|G \|_{\infty} \Big[  \Big( \tfrac{t}{n} \Big)^{- \gamma -1}    -    w^{-\gamma-1} \Big] dw   +  2 \| \partial_u G \|_{\infty} \sum_{t =  \varepsilon n}^{\infty} \int_{ \tfrac{t}{n} }^{\tfrac{t+1}{n}} w^{-\gamma-1}  \Big( \tfrac{t+1}{n} - \tfrac{t}{n}  \Big) dw \\
\leq & 4 \|G \|_{\infty}  \sum_{t =  \varepsilon n}^{\infty} ( \gamma + 1) n^{\gamma}  t^{-\gamma-2}+ \frac{2 \| \partial_u G \|_{\infty}}{n \gamma} \varepsilon^{- \gamma} \lesssim \tfrac{\varepsilon^{-\gamma-1}}{n} + \tfrac{\varepsilon^{- \gamma}}{n } .
\end{align*}
Therefore the rightmost term in \eqref{eqlapfrac2} is bounded from above by a constant times $\tfrac{\varepsilon^{-\gamma-1}}{n} + \tfrac{\varepsilon^{- \gamma}}{n } $ and when we take the limit $n\to+\infty$ this term vanishes. 
From  all this we conclude that  \eqref{limlapfrac2} holds, proving our claim. 
Finally, we claim that
\begin{equation} \label{convdisc5sdif}
 \lim_{\varepsilon \rightarrow 0^{+}} \limsup_{n \rightarrow \infty} \frac{1}{n} \sum_{|x|\leq  2 bn} \sup_{s \in [0,T]} | [-(- \Delta)^{\frac{\gamma}{2}} G ]_{\varepsilon} (s, \tfrac{x}{n} ) -  [-(- \Delta)^{\frac{\gamma}{2}} G ] (s, \tfrac{x}{n} ) | =0
\end{equation} and this will end the proof. 
Since $[-(- \Delta)^{\frac{\gamma}{2}} G ] (s, \tfrac{x}{n} )= \lim_{\varepsilon_1 \rightarrow 0^+} [-(- \Delta)^{\frac{\gamma}{2}} G ]_{\varepsilon_1} (s, \tfrac{x}{n} )$, recalling \eqref{lapfractheta} and \eqref{thetabounddelta},  we get
\begin{align*}
& \frac{1}{n} \sum_{|x| \leq 2 bn} \sup_{s \in [0,T]} | [-(- \Delta)^{\frac{\gamma}{2}} G ]_{\varepsilon} (s, \tfrac{x}{n}) -  [-(- \Delta)^{\frac{\gamma}{2}} G ] (s, \tfrac{x}{n} ) |  = \frac{c_{\gamma}}{n} \sum_{|x| \leq 2 bn} \sup_{s \in [0,T]} \Big|   \lim_{\varepsilon_1 \rightarrow 0^+}   \int_{ \varepsilon_1}^{ \varepsilon} \frac{\theta_{\frac{x}{n}}(s, w)}{ w^{\gamma+1}} dw \Big| \\
\leq & \frac{c_{\gamma}}{n} \sum_{|x| \leq 2 bn} \lim_{\varepsilon_1 \rightarrow 0^+}   \int_{ \varepsilon_1}^{ \varepsilon} \frac{\| \Delta G \|_{\infty} w^2 }{ w^{\gamma+1}} dw  \leq  \frac{5b c_{\gamma} \| \Delta G \|_{\infty}}{2 - \gamma} \varepsilon^{2-\gamma}, \forall n \geq 1.
\end{align*}
Since $\gamma < 2$, we get \eqref{convdisc5sdif}. 
\end{proof}
Since for all $s\in[0,T]$ we have $|\eta_s^n(x)| \leq 1$, the next result is a trivial consequence of the previous  one. 
\begin{cor} \label{pdbgeq}
If $G \in  \mcb S_{\textrm{Dif}}$, we have
\begin{align*}
\lim_{n \rightarrow \infty} \mathbb{E}_{\mu_n} \Big[ \sup_{t \in [0,T]} \Big| \int_0^t \Big\{ \frac{1}{n} \sum _{x } n^{\gamma} \mcb{K}_n G \left(s, \tfrac{x}{n} \right)  \eta_{s}^{n}(x) -  \langle \pi_s^n, [-(- \Delta)^{\frac{\gamma}{2}}   G] (s, \cdot) \rangle \Big\} ds \Big|  \Big] =0.
\end{align*}
\end{cor}
With  the next two results we were able to treat \eqref{extrasdif}.  These results   motivated us to impose the condition  \eqref{defsigmas}, simplifying the proofs in some regimes.
\begin{prop} \label{neum1}
Assume \eqref{defsigmas}. For $G \in \mcb S_{\textrm{Dif}}$, it holds
\begin{align*}
 \limsup_{n \rightarrow \infty}    n^{\gamma-1} \sum _{\{y,z\} \in \mcb S} p(y-z)  \sup_{s \in [0,T]}| G(s, \tfrac{y}{n}) - G(s, \tfrac{z}{n}) |  = 0.
\end{align*}
\end{prop}
\begin{proof}
Since $G \in \mcb S_{Dif}$, there exists a constant $C_{\delta}$ (independent of $s$) such that $|G(s,u) - G(s,v)| \leq C_{\delta} |u-v|^{\delta}$ for every $u,v \in \mathbb{R}$, for every $s \in [0,T]$. Since  \eqref{defsigmas} holds, we have $\delta > \gamma - 1$ and taking $n \rightarrow \infty$ we trivially  get the result.
\end{proof}
 Since for all $s\in[0,T]$ we have $|\eta_s^n(x)| \leq 1$, the next result is a direct consequence of the previous one. 
\begin{cor} \label{convbound}
Assume \eqref{defsigmas}. For $G \in \mcb S_{\textrm{Dif}}$, we have
\begin{align}
 \limsup_{\varepsilon \rightarrow 0^+} \limsup_{n \rightarrow \infty} \mathbb{E}_{\mu_n} \Big[& \sup_{t \in [0,T]} \Big| \int_{0}^{t} \Big\{ \frac{n^{\gamma-1}}{2} \sum _{ \{x, y \} \in \mcb S } \left[ G(s, \tfrac{y}{n}) - G(s, \tfrac{x}{n}) \right] p(y-x) [\eta_s^n(y)-\eta_s^n(x)] \Big\}   ds \Big| \Big] =0. \label{eqconvbound}
\end{align}
\end{cor}
Now we present a useful result to treat \eqref{srobterm}. This result is as in Proposition A.3. of \cite{casodif}, where it is stated for $\gamma>2$ but the result holds for $\gamma>1$ and the  proof can be easily  adapted from \cite{casodif}  to include this regime of  $\gamma$. 

\begin{prop} \label{lemconvrob}
Let $\gamma \in (1,2)$ and  $G \in \mcb S_{Rob}$. Then, 
\begin{equation} \label{lemrobzpos1}
\lim_{n \rightarrow \infty} \sup_{s \in [0,T]}  \Big|  \sum_{z=0}^{\infty}  \sum_{x = - \infty}^{-1} p(y-x) \big( [ G(s,\tfrac{y}{n}) - G(s,\tfrac{x}{n}) ] -[G(s,0^{-}) - G(s,0^{+}) ] \big)   \Big|  =0.
\end{equation}
By symmetry the same result is true if we exchange $x$ with $z$.
\end{prop}

Finally we present two results that are useful to treat \eqref{sneuterm}.
\begin{prop} \label{lemconvneum}
Let $\mcb S = \mcb S_0$. Assume  $\gamma \in (1,2)$ and  $G \in \mcb S_{\textrm{Rob}}$ or  that $\gamma \in (0,1]$ and $G \in \mcb S_{\textrm{Neu}}$. It holds 
\begin{equation} \label{lemneupos} 
\lim_{n \rightarrow \infty}  \frac{1}{n}  \sum_{x} \sup_{s \in [0,T]} | n^{\gamma} {\sum_{ y: \{x, y\} \in \mcb F }}[ G(s, \tfrac{y}{n} ) - G(s, \tfrac{x}{n} ) ] p(y-x)  - [-(- \Delta)_{\mathbb{R}^{*}}^{\frac{\gamma}{2}} G ]  \left(s,\tfrac{x}{n} \right) | =0.
\end{equation}
\end{prop}
\begin{proof}
We start by decomposing the sum above taking into account the relative position of $\{x,y\}\in\mcb F$. Since $\mcb S = \mcb S_0$, the previous  limit  is bounded by
\begin{align}
& \lim_{n \rightarrow \infty}\frac{1}{n}  \sum_{x=1}^{\infty} \sup_{s \in [0,T]} | n^{\gamma} \sum_{ y = 1}^{\infty} [ G(s, \tfrac{y}{n} ) - G(s, \tfrac{x}{n} ) ] p(y-x)  - [-(- \Delta)_{\mathbb{R}_{+}^{*}}^{\frac{\gamma}{2}} G ]  \left(s,\tfrac{x}{n} \right) | \label{limlapregpos}  \\
+&\lim_{n \rightarrow \infty} \frac{1}{n}  \sum_{x=-\infty}^{-1} \sup_{s \in [0,T]} | n^{\gamma} \sum_{ y = -\infty}^{-1} [ G(s, \tfrac{y}{n} ) - G(s, \tfrac{x}{n} ) ] p(y-x)  - [-(- \Delta)_{\mathbb{R}_{-}^{*}}^{\frac{\gamma}{2}} G ]  \left(s,\tfrac{x}{n} \right) | \label{limlapregneg} \\
+ & \lim_{n \rightarrow \infty} n^{\gamma-1}    \sum_{y=1}^{\infty}  \sup_{s \in [0,T]} | G(s, \tfrac{y}{n} ) - G(s, \tfrac{0}{n} ) | p(y) \label{limlapregzero}.
\end{align}
We claim that the limits in \eqref{limlapregzero}, \eqref{limlapregpos} and \eqref{limlapregneg} are equal to zero. Let us begin with \eqref{limlapregzero}. Since $\gamma \in (0,2)$, we can choose $\delta \in (\gamma-1, \gamma) \cap [0,1]$. Since $G \in \mcb S_{Rob}$, there exists $C_1$ (independent of $s$) such that
$\sup_{s \in [0,T]}| G(s, \tfrac{y}{n} ) - G(s, \tfrac{0}{n} ) |  \leq C_1 y^{\delta} n^{-\delta}, \forall y \geq 1, \forall n \geq 1.
$
Then the limit in \eqref{limlapregzero} is bounded from above by
$n^{\gamma-1 - \delta}  C_1 c_{\gamma}  \sum_{y=1}^{\infty}  y^{\delta-\gamma-1} $, which vanishes as $n\to+\infty$.
Above, we used the fact that $\gamma-1 - \delta <0$ and that the sum over $y$ is convergent. It remains to deal with \eqref{limlapregpos} and \eqref{limlapregneg}. We will only prove \eqref{limlapregpos},  but the proof of  \eqref{limlapregneg} is analogous. The proof goes in two steps. First we treat the terms in the sum for large values of $x$ and then we treat at remaining terms. We split the proof now in these two cases. 

\textbf{I).} First step: treating terms with large values of $x$. Let $b \geq \max \{1, b_G\}$. We claim that
\begin{equation} \label{lim1lapfracreg}
\lim_{n \rightarrow \infty} \frac{1}{n} \sum_{x=2bn+1}^{\infty} \sup_{s \in [0,T]} | n^{\gamma} \sum_{ y = 1}^{\infty} [ G(s, \tfrac{y}{n} ) - G(s, \tfrac{x}{n} ) ] p(y-x)  - [-(- \Delta)_{\mathbb{R}_{+}^{*}}^{\frac{\gamma}{2}} G ]  \left(s,\tfrac{x}{n} \right) | =0.
\end{equation} The proof is very close to the proof of \eqref{lim1lapfrac} and for that reason we just give a sketch. 
Since $b \geq b_G$, we have $ \sup_{s \in [0,T]} G(s, u)=0$ if $|u| \geq b$ and the limit above can be written as
\begin{align} \label{lapfrac1reg}
\lim_{n \rightarrow \infty}  \frac{c_{\gamma}}{n} \sum_{x=2bn+1}^{\infty} \sup_{s \in [0,T]} \Big|  \sum_{y =1}^{bn} \int_{\tfrac{y-1}{n}}^{\frac{y}{n}} [g_x(s,\tfrac{y}{n}) - g_x(s,v)]  dv \Big|,
\end{align}
where $g^n_x: \in [0,T]\times  (0, b) \rightarrow \mathbb{R}$ is given by 
$
g_x^n(s,u):= G(s,u)[\tfrac{x}{n}-u]^{-\gamma-1}, \forall s \in [0,T], \forall u \in (0,b).
$
Since $b \geq 1$, with the same reasoning we did to produce \eqref{MVTA1}, we get
\begin{align*}
\sup_{s \in [0,T]} |g_x(s,\tfrac{y}{n}) - g_x(s,v)| \lesssim (v - \tfrac{y}{n})  [\tfrac{x}{n}-b]^{-\gamma-1} , \forall x \geq 2bn+1, \forall y: \ \leq y  \leq bn, \forall v \in ( \tfrac{y-1}{n} , \tfrac{y}{n}),
\end{align*}
and by the triangle inequality, we can bound  \eqref{lapfrac1reg} by a constant times
$\frac{b^{1-\gamma}}{n \gamma} $, which vanishes as $n\to+\infty$, proving \eqref{lim1lapfracreg}.

\textbf{II).} Second step: treating  terms with  small values of $x$.
First we fix $\varepsilon >0$ and later we take $\varepsilon \rightarrow 0^+$. We first claim that
\begin{equation} \label{lim2lapfracreg}
\lim_{\varepsilon \rightarrow 0^{+}} \lim_{n \rightarrow \infty} \frac{1}{n}    \sum_{x=1}^{\varepsilon n -1}  \sup_{s \in [0,T]} | n^{\gamma}  \sum_{y=1}^{\infty}    [G (s,\tfrac{y}{n} ) - G (s,\tfrac{x}{n} ) ]p(y-x) | =0.
\end{equation}
First note that  a Taylor expansion of second order in $G$ allows to conclude that 
\begin{align*}
&\frac{1}{n}    \sum_{x=1}^{\varepsilon n -1}  \sup_{s \in [0,T]} | n^{\gamma}  \sum_{y=1}^{2x-1}    [G (s,\tfrac{y}{n} ) - G (s,\tfrac{x}{n} ) ]p(y-x) |\\&\leq  \sum_{x=1}^{\varepsilon n -1} n^{\gamma-3}   \| \Delta G \|_{\infty}  \sum_{t=1}^{x-1} c_{\gamma} t^{1-\gamma}  =  \| \Delta G \|_{\infty}c_{\gamma} \frac{1}{n} \sum_{x=1}^{\varepsilon n -1} \frac{1}{n}  \sum_{t=1}^{x-1} \Big( \frac{t}{n} \Big)^{1-\gamma} 
\lesssim  \| \Delta G \|_{\infty}  \frac{1}{n} \sum_{x=0}^{\varepsilon n }  \Big( \frac{x}{n} \Big)^{2-\gamma} \lesssim 
\| \Delta G \|_{\infty} \varepsilon^{3-\gamma},
\end{align*}
and from this we get  
\begin{equation} \label{lim3lapfracreg}
\lim_{\varepsilon \rightarrow 0^{+}} \lim_{n \rightarrow \infty} \frac{1}{n}    \sum_{x=1}^{\varepsilon n -1}  \sup_{s \in [0,T]} | n^{\gamma}  \sum_{y=1}^{2x-1}    [G (s,\tfrac{y}{n} ) - G (s,\tfrac{x}{n} ) ]p(y-x) | =0.
\end{equation}
The proof of \eqref{lim2lapfracreg} is now a consequence of the next result:
\begin{equation} \label{lim4lapfracreg}
\lim_{\varepsilon \rightarrow 0^{+}} \lim_{n \rightarrow \infty} \frac{1}{n}    \sum_{x=1}^{\varepsilon n -1}  \sup_{s \in [0,T]} | n^{\gamma}  \sum_{y=2x}^{\infty}    [G (s,\tfrac{y}{n} ) - G (s,\tfrac{x}{n} ) ]p(y-x) | =0.
\end{equation} To prove it we split in two cases, either $\gamma \in (0,1]$ or $\gamma \in (1,2)$. We start with the former. 
If $\gamma \in (0,1]$, then $G \in \mcb S_{Neu}$ and $\sup_{s \in [0,T] }|G(s,u)|=0$ when $ u \in (0, \bar{b}_G)$. Then if $\varepsilon < \frac{\bar{b}_G}{2}$ we have
\begin{align*}
& \frac{1}{n}    \sum_{x=1}^{\varepsilon n -1}  \sup_{s \in [0,T]} | n^{\gamma}  \sum_{y=2x}^{\infty}    [G (s,\tfrac{y}{n} ) - G (s,\tfrac{x}{n} ) ]p(y-x) | \\
\leq & \frac{1}{n}    \sum_{x=1}^{\varepsilon n -1}  \sup_{s \in [0,T]} | n^{\gamma}  \sum_{y=2x}^{x+\varepsilon n}    [G (s,\tfrac{y}{n} ) - G (s,\tfrac{x}{n} ) ]p(y-x) |  + \frac{1}{n}    \sum_{x=1}^{\varepsilon n -1}  \sup_{s \in [0,T]} | n^{\gamma}  \sum_{y=x+ \varepsilon n +1}^{\infty}    [G (s,\tfrac{y}{n} ) - G (s,\tfrac{x}{n} ) ]p(y-x) | \\
= & \frac{1}{n}    \sum_{x=1}^{\varepsilon n -1}  \sup_{s \in [0,T]} | n^{\gamma}  \sum_{y=\bar{b}_G n }^{\infty}    [G (s,\tfrac{y}{n} ) ]p(y-x) | 
=\frac{ \| G \|_{\infty}}{n}    \sum_{x=1}^{\varepsilon n -1}  \frac{1}{n}  \sum_{y=\bar{b}_G n }^{\infty}   \Big( \tfrac{y}{n} - \tfrac{x}{n} \Big)^{-\gamma-1}  \\\lesssim & \frac{ \| G \|_{\infty}}{n}    \sum_{x=1}^{\varepsilon n }   \Big( \bar{b}_G - \tfrac{x}{n} \Big)^{-\gamma} \leq \| G \|_{\infty}   \Big( \tfrac{\bar{b}_G}{2}  \Big)^{-\gamma} \varepsilon,
\end{align*}
which shows  \eqref{lim4lapfracreg}. If $\gamma \in (1,2)$, a Taylor expansion of first order leads to
\begin{align*}
& \frac{1}{n}    \sum_{x=1}^{\varepsilon n -1}  \sup_{s \in [0,T]} | n^{\gamma}  \sum_{y=2x}^{\infty}    [G (s,\tfrac{y}{n} ) - G (s,\tfrac{x}{n} ) ]p(y-x) |  \\
\leq  &  n^{\gamma-2} \| \partial_u G  \|_{\infty} \sum_{x=1}^{\varepsilon n -1} \sum_{t=x}^{\varepsilon n}  c_{\gamma} t^{-\gamma} +   n^{\gamma-2}  \| \partial_u G  \|_{\infty} \sum_{x=1}^{\varepsilon n -1} \sum_{t=\varepsilon n+1}^{\infty}  c_{\gamma} t^{-\gamma} \\
\lesssim &  \| \partial_u G  \|_{\infty}   \frac{1}{n} \sum_{x=1}^{\varepsilon n} \frac{2}{\gamma-1}  \Big( \tfrac{x}{n} \Big)^{1-\gamma} +\| \partial_u G  \|_{\infty}    \varepsilon  \frac{1}{n} \sum_{t=\varepsilon n}^{\infty} \Big( \tfrac{t}{n} \Big)^{-\gamma} \lesssim \| \partial_u G  \|_{\infty} \varepsilon^{2-\gamma} + \| \partial_u G  \|_{\infty} \varepsilon^{2-\gamma}, 
\end{align*}
and this shows  $\eqref{lim4lapfracreg}$. 
Now we claim that
\begin{equation} \label{lim5lapfracreg}
\lim_{\varepsilon \rightarrow 0^{+}} \lim_{n \rightarrow \infty} \frac{1}{n}    \sum_{x=1}^{\varepsilon n -1}  \sup_{s \in [0,T]} | [ -(- \Delta)_{\mathbb{R}_{+}}^{\frac{\gamma}{2}} G](s, \tfrac{x}{n}) | =0.
\end{equation} As above, we split the proof into two cases, either $\gamma \in (0,1]$ or $\gamma \in (1,2)$. We start with the former. 
If $\gamma \in (0,1]$, then $G \in \mcb S_{Neu}$ and $\sup_{s \in [0,T] }|G(s,u)|=0$ when $ u \in (0, \bar{b}_G)$. Then, for $\varepsilon < \frac{\bar{b}_G}{2}$ the expression inside the double limit above can be rewritten as
\begin{align*}
& \frac{1}{n}    \sum_{x=1}^{\varepsilon n -1}  \sup_{s \in [0,T]} \Big| c_{\gamma}  \lim_{\varepsilon_1 \rightarrow 0^{+}}  \int_{0}^{\infty} \mathbbm{1}_{\{ |v - \frac{x}{n} | \geq \varepsilon_1 \}} \frac{G(s,u) - G(s,\tfrac{x}{n})}{ |u-\tfrac{x}{n}|^{\gamma+1}} du \Big| \\
 = & \frac{1}{n}    \sum_{x=1}^{\varepsilon n -1}  \sup_{s \in [0,T]} \Big| c_{\gamma}  \int_{\bar{b}_G}^{\infty}  \frac{G(s,u) }{ |u-\tfrac{x}{n}|^{\gamma+1}} du \Big| \lesssim  \frac{ \| G \|_{\infty}}{n}    \sum_{x=1}^{\varepsilon n }   \Big( \bar{b}_G - \frac{x}{n} \Big)^{-\gamma} \leq \| G \|_{\infty}   \Big( \frac{\bar{b}_G}{2}  \Big)^{-\gamma} \varepsilon,
\end{align*}
which leads to \eqref{lim5lapfracreg}. If $\gamma \in (1,2)$, choosing $\varepsilon_1, \varepsilon, n$ such that $\varepsilon_1 < \frac{1}{n} < \varepsilon$, we get
\begin{align*}
& \Big| c_{\gamma} \int_{\mathbb{R}_{+}} \mathbbm{1}_{\{ |\frac{x}{n} - v| \geq \varepsilon_1 \}} \frac{G(s, v) - G(s, \tfrac{x}{n})}{ |\tfrac{x}{n}-v|^{\gamma+1}} duv\Big|=  c_{\gamma}\Big| \int_{- \frac{x}{n}}^{\infty} \mathbbm{1}_{\{ |u | \geq \varepsilon_1 \}} \frac{G(s, u+ \tfrac{x}{n} ) - G(s, \tfrac{x}{n})}{ |u|^{\gamma+1}} du \Big| \\
=& c_{\gamma} \Big| \int_{- \frac{x}{n}}^{-\varepsilon_1}  \frac{G(s, u+ \tfrac{x}{n} ) - G(s, \tfrac{x}{n})}{ (-u)^{\gamma+1}} du  +  \int_{\varepsilon_1}^{1}  \frac{G(s, u+ \tfrac{x}{n} ) - G(s, \tfrac{x}{n})}{ u^{\gamma+1}} du +  \int_{1}^{\infty}  \frac{G(s, u+ \tfrac{x}{n} ) - G(s, \tfrac{x}{n})}{ u^{\gamma+1}} du \Big|.
\end{align*}
We know that if the spatial variable in the argument of $G$ is always positive, then $G$ can be replaced by $G_{+}$ for $\gamma \in (1,2)$. Performing a Taylor expansion of second order inside the first two integrals and a Taylor expansion of first order inside the last integral, we get 
\begin{align*}
&\sum_{x=1}^{\varepsilon n -1} n^{-1} \Big| c_{\gamma} \int_{\mathbb{R}_{+}} \mathbbm{1}_{\{ |\frac{x}{n} - v| \geq \varepsilon_1 \}} \frac{G(s, v) - G(s, \tfrac{x}{n})}{ |\tfrac{x}{n}-v|^{\gamma+1}} du \Big| \\
\leq & \sum_{x=1}^{\varepsilon n -1} n^{-1} c_{\gamma} \Big| \partial_uG(s, \tfrac{x}{n}) \Big[ \int_{\varepsilon_1 }^{\frac{x}{n}} (-u) u^{-\gamma-1} du + \int_{\varepsilon_1}^{1}  u u^{-\gamma-1} du \Big] \Big| + \sum_{x=1}^{\varepsilon n -1} n^{-1} c_{\gamma} \| \partial_u G \|_{\infty}  \int_{1}^{\infty}  u^{-\gamma} du   \\
+&\sum_{x=1}^{\varepsilon n -1} n^{-1} c_{\gamma} \| \Delta G \|_{\infty} \Big[  \int_{- \frac{x}{n}}^{-\varepsilon_1} (-u)^{1-\gamma}du  +  \int_{\varepsilon_1}^{1} u^{1-\gamma}du  \Big| \\
\leq & \sum_{x=1}^{\varepsilon n -1} n^{-1} c_{\gamma} | \partial_uG(s, \tfrac{x}{n})  |\int_{\frac{x}{n}}^{1}  u^{-\gamma} du + \sum_{x=1}^{\varepsilon n -1} n^{-1} c_{\gamma} \| \partial_u G \|_{\infty}  \int_{1}^{\infty}  u^{-\gamma} du + 2 \| \Delta G \|_{\infty} c_{\gamma} \varepsilon \int_{\varepsilon_1}^{1} u^{1-\gamma}du   \\
\leq &  \sum_{x=1}^{\varepsilon n -1} n^{-1} c_{\gamma} \| \partial_u G \|_{\infty}  \int_{\frac{x}{n}}^{\infty}  u^{-\gamma} du + 2 \| \Delta G \|_{\infty} c_{\gamma} \varepsilon \int_{\varepsilon_1}^{1} u^{1-\gamma}du \\
\lesssim & \| \partial_u G \|_{\infty} \frac{1}{n} \sum_{x=1}^{\varepsilon n } \Big( \tfrac{x}{n} \Big)^{1-\gamma} + \| \Delta G \|_{\infty}  \varepsilon \lesssim  \| \partial_u G \|_{\infty}  \varepsilon^{2-\gamma} +  \| \Delta G \|_{\infty}  \varepsilon,
\end{align*}
which leads to \eqref{lim5lapfracreg}. From \eqref{lim2lapfracreg} and \eqref{lim5lapfracreg}, we get
\begin{equation} \label{lim6lapfracreg}
\lim_{\varepsilon \rightarrow 0^{+}} \lim_{n \rightarrow \infty} \frac{1}{n} \sum_{x=1}^{\varepsilon n - 1} \sup_{s \in [0,T]} | n^{\gamma} \sum_{ y = 1}^{\infty} [ G(s, \tfrac{y}{n} ) - G(s, \tfrac{x}{n} ) ] p(y-x)  - [-(- \Delta)_{\mathbb{R}_{+}^{*}}^{\frac{\gamma}{2}} G ]  \left(s,\tfrac{x}{n} \right) | =0.
\end{equation}
\textbf{III).}  Third step: treating the remaining terms. Observe that we can write
\begin{align*}
& \frac{1}{n} \sum_{x=\varepsilon n}^{2bn} \sup_{s \in [0,T]} | n^{\gamma} \sum_{ y = 1}^{\infty} [ G(s, \tfrac{y}{n} ) - G(s, \tfrac{x}{n} ) ] p(y-x)  - [-(- \Delta)_{\mathbb{R}_{+}^{*}}^{\frac{\gamma}{2}} G ]  \left(s,\tfrac{x}{n} \right) | \\
\leq & \frac{1}{n} \sum_{x=\varepsilon n}^{2bn} \sup_{s \in [0,T]}| n^{\gamma} \sum_{ y = 1}^{\infty} [ G(s, \tfrac{y}{n} ) - G(s, \tfrac{x}{n} ) ] p(y-x) - [-(- \Delta)_{\mathbb{R}_{+}}^{\frac{\gamma}{2}} G ]_{\varepsilon} (s,\tfrac{x}{n} ) \Big|  \\
+& \frac{1}{n} \sum_{x=\varepsilon n}^{2bn} \sup_{s \in [0,T]}|[-(- \Delta)_{\mathbb{R}_{+}}^{\frac{\gamma}{2}} G ]_{\varepsilon} (s, \tfrac{x}{n} ) -  [-(- \Delta)^{\frac{\gamma}{2}} G ] (s, \tfrac{x}{n} ) |, 
\end{align*}
where for every $u>0$
\begin{align}
[-(- \Delta)_{\mathbb{R}_{+}}^{\frac{\gamma}{2}} G ]_{\varepsilon} (s, u ) 
= c_{\gamma} \int_{ \varepsilon}^{u} \frac{\theta_u(s, w)}{ w^{\gamma+1}} dw + c_{\gamma} \int_{ u}^{\infty}  \frac{G(s, u+w) - G(s, u)}{ w^{\gamma+1}} dw , \label{lapfracthetapos}
\end{align}
where the function $\theta$ has been defined in \eqref{eq:function_theta}.  Assuming  $n > 2 \varepsilon^{-1}$,  for $x \geq \varepsilon n$, we also can write
\begin{align*}
  n^{\gamma}  \sum_{y =1}^{\infty} p(y-x) [ G (s, \tfrac{y}{n} ) - G (s, \tfrac{x}{n} ) ] = n^{\gamma}  \sum_{t =1}^{x-1} p(t) \theta_{\tfrac{x}{n}} (s, \tfrac{t}{n} ) + n^{\gamma}  \sum_{t =x}^{\infty} p(t) [ G (s, \tfrac{x+t}{n}  ) - G (s, \tfrac{x}{n} ) ].
\end{align*}
From this we have 
\begin{align}
& \frac{1}{n} \sum_{x=\varepsilon n}^{2bn} |   n^{\gamma}  \sum_{y =1}^{\infty} p(y-x) [ G (s, \tfrac{y}{n} ) - G (s, \tfrac{x}{n} ) ] - [-(- \Delta)_{\mathbb{R}_{+}}^{\frac{\gamma}{2}} G ]_{\varepsilon} (s, \tfrac{x}{n} ) \Big|\nonumber  \\ 
\leq  & \frac{c_{\gamma}}{n} \sum_{x=\varepsilon n}^{2bn} \Big| n^{\gamma}  \sum_{t =1}^{\varepsilon n-1} t^{-\gamma-1} \theta_{\tfrac{x}{n}} (s, \tfrac{t}{n} )   \Big| + \frac{c_{\gamma}}{n} \sum_{x=\varepsilon n}^{2bn} \Big| n^{\gamma}  \sum_{t =\varepsilon n}^{x-1} t^{-\gamma-1} \theta_{\tfrac{x}{n}} (s, \tfrac{t}{n} ) -  \int_{ \varepsilon}^{\tfrac{x}{n}} \frac{\theta_{\tfrac{x}{n}}(s,w)}{ w^{\gamma+1}} dw  \Big| \label{eq:imp}\\
+ & \frac{c_{\gamma}}{n} \sum_{x=\varepsilon n}^{2bn} \Big|  n^{\gamma}  \sum_{t =x}^{\infty} t^{-\gamma-1} [ G (s, \tfrac{x+t}{n}  ) - G (s, \tfrac{x}{n} ) ]  -  \int_{ \tfrac{x}{n}}^{\infty}  \frac{G(s,\tfrac{x}{n}+w) - G(s,\tfrac{x}{n})}{ w^{\gamma+1}} dw \Big|\label{eq:imp1}.
\end{align}
From  \eqref{thetabounddelta},  the  leftmost term in \eqref{eq:imp} can be  bounded by
\begin{align*}
& \frac{1}{n} \sum_{x = \varepsilon n}^{2 bn} c_{\gamma}  \|  \Delta G \|_{\infty} \frac{1}{n}  \sum_{t =1}^{ \varepsilon n} \Big( \frac{t}{n} \Big)^{1 - \gamma} \leq  2b c_{\gamma}  \|  \Delta G \|_{\infty} \frac{1}{n}  \sum_{t =1}^{ \varepsilon n} \Big( \frac{t}{n} \Big)^{1 - \gamma} \lesssim  \|  \Delta G \|_{\infty} \varepsilon^{2-\gamma},
\end{align*}
and since $\gamma < 2$, it vanishes as $\epsilon\to 0$. Now we analyse the the  rightmost term in \eqref{eq:imp}. From the triangular inequality, we get
\begin{align}
&  \Big|  n^{\gamma}  \sum_{t =  \varepsilon n}^{x-1} t^{-\gamma-1} \theta_{\tfrac{x}{n}} (s, \tfrac{t}{n} ) -   \int_{ \varepsilon}^{\frac{x}{n}} \frac{\theta_{\tfrac{x}{n}}(s,w)}{ w^{\gamma+1}} dw \Big| \nonumber \\
    =&  \Big|    \sum_{t =  \varepsilon n}^{x-1} \int_{ \frac{t}{n} }^{\frac{t+1}{n}} \Big[  \Big( \frac{t}{n} \Big)^{- \gamma -1}  \theta_{\tfrac{x}{n}} ( s,\tfrac{t}{n} )  -    w^{-\gamma-1} \theta_{\tfrac{x}{n}}(s,w) \Big] dw \Big| \nonumber  \\
\leq &  \sum_{t =  \varepsilon n}^{x-1} \int_{ \frac{t}{n} }^{\frac{t+1}{n}} | \theta_{\tfrac{x}{n}} (s, \tfrac{t}{n} ) | \Big[  \Big( \frac{t}{n} \Big)^{- \gamma -1}    -    w^{-\gamma-1} \Big] dw   +   \sum_{t =  \varepsilon n}^{x-1} \int_{ \frac{t}{n} }^{\frac{t+1}{n}} w^{-\gamma-1} |   \theta_{\tfrac{x}{n}} (s, \tfrac{t}{n} ) -     \theta_{\tfrac{x}{n}}(s,w) | dw . \label{eqlapfrac3b}
\end{align}
Recall \eqref{thetabounddelta} and the application of the mean value theorem above \eqref{mvt2}. Putting this together with  \eqref{mvt2}, we bound  \eqref{eqlapfrac3b} from above by
\begin{align*}
&  \sum_{t =  \varepsilon n}^{x-1} \int_{ \frac{t}{n} }^{\frac{t+1}{n}} 4 \|G \|_{\infty} \Big[  \Big( \frac{t}{n} \Big)^{- \gamma -1}    -    w^{-\gamma-1} \Big] dw   +  2 \| \partial_u G \|_{\infty} \sum_{t =  \varepsilon n}^{x-1} \int_{ \frac{t}{n} }^{\frac{t+1}{n}} w^{-\gamma-1}  \Big( \frac{t+1}{n} - \frac{t}{n}  \Big) dw \\
\leq & 4 \|G \|_{\infty}  \sum_{t =  \varepsilon n}^{\infty} ( \gamma + 1) n^{\gamma}  t^{-\gamma-2}+ \frac{2 \| \partial_u G \|_{\infty}}{n \gamma} \varepsilon^{- \gamma} \lesssim \tfrac{\varepsilon^{-\gamma-1}}{n} + \tfrac{\varepsilon^{- \gamma}}{n }.
\end{align*}
From this we get that the rightmost term in the second line of last display
vanishes as $\epsilon \to 0.$

Now we analyse \eqref{eq:imp1}.
From the triangular inequality, we get
\begin{align}
&  \Big|  n^{\gamma}  \sum_{t =x}^{\infty} t^{-\gamma-1} [ G (s, \tfrac{x+t}{n}  ) - G (s, \tfrac{x}{n} ) ]  -  \int_{ \tfrac{x}{n}}^{\infty}  \frac{G(s, \tfrac{x}{n}+w) - G(s, \tfrac{x}{n})}{ w^{\gamma+1}} dw \Big| \nonumber  \\
   =&  \Big|    \sum_{t =  x}^{\infty} \int_{ \frac{t}{n} }^{\frac{t+1}{n}} \Big[  \Big( \tfrac{t}{n} \Big)^{- \gamma -1}  [ G (s, \tfrac{x+t}{n}  ) - G (s, \tfrac{x}{n} ) ]  -    w^{-\gamma-1} [G(s,\tfrac{x}{n}+w) - G(s,\tfrac{x}{n})] \Big] dw \Big| \nonumber  \\
\leq &  \sum_{t =  \varepsilon n}^{\infty} \int_{ \frac{t}{n} }^{\frac{t+1}{n}} |  G ( s,\tfrac{x+t}{n}  ) - G (s, \tfrac{x}{n} )  | \Big[  \Big( \tfrac{t}{n} \Big)^{- \gamma -1}    -    w^{-\gamma-1} \Big] dw   +   \sum_{t =  \varepsilon n}^{\infty} \int_{ \frac{t}{n} }^{\frac{t+1}{n}} w^{-\gamma-1} |  G (s, \tfrac{x+t}{n}  ) -    G(s,\tfrac{x}{n}+w) | dw . \label{eqlapfrac3c}
\end{align}
In the last line, we used that $x \geq \varepsilon n$. We observe that $| G(s,y) - G(s, w)| \leq  \| \partial_u G \|_{\infty} |y-w|, \forall y,w >0, \forall s \in [0,T]$.
Plugging this with \eqref{mvt2}, the expression in \eqref{eqlapfrac3c} is bounded from above by
\begin{align*}
&  \sum_{t =  \varepsilon n}^{\infty} \int_{ \frac{t}{n} }^{\frac{t+1}{n}} 4 \|G \|_{\infty} \Big[  \Big( \frac{t}{n} \Big)^{- \gamma -1}    -    w^{-\gamma-1} \Big] dw   +   \| \partial_u G \|_{\infty} \sum_{t =  \varepsilon n}^{\infty} \int_{ \frac{t}{n} }^{\frac{t+1}{n}} w^{-\gamma-1}  \Big( \frac{t+1}{n} - \frac{t}{n}  \Big) dw \\
\leq & 4 \|G \|_{\infty}  \sum_{t =  \varepsilon n}^{\infty} ( \gamma + 1) n^{\gamma}  t^{-\gamma-2}+ \frac{ \| \partial_u G \|_{\infty}}{n \gamma} \varepsilon^{- \gamma} \lesssim \tfrac{\varepsilon^{-\gamma-1}}{n} + \tfrac{\varepsilon^{- \gamma}}{n },
\end{align*}
so that \eqref{eqlapfrac3c}  vanishes as $n\to+\infty$.
Putting this all together we get
 \begin{equation} \label{lim10lapfracreg}
 \lim_{\varepsilon \rightarrow 0^{+}}  \limsup_{n \rightarrow \infty} \frac{1}{n} \sum_{x=\varepsilon n}^{2bn} |   n^{\gamma}  \sum_{y =1}^{\infty} p(y-x) [ G (s, \tfrac{y}{n} ) - G (s, \tfrac{x}{n} ) ] - [-(- \Delta)_{\mathbb{R}_{+}}^{\frac{\gamma}{2}} G ]_{\varepsilon} (s, \tfrac{x}{n} ) \Big| =0.
\end{equation}
Since $[-(- \Delta)_{\mathbb{R}_{+}}^{\frac{\gamma}{2}} G ] (s,\tfrac{x}{n} ) = \lim_{\varepsilon_1 \rightarrow 0^+} [-(- \Delta)_{\mathbb{R}_{+}}^{\frac{\gamma}{2}} G ]_{\varepsilon_1} (s,\tfrac{x}{n} )$, recalling \eqref{lapfracthetapos} and \eqref{thetabounddelta},  we get
\begin{align*}
& \frac{1}{n} \sum_{x =  \varepsilon n}^{2 bn} \sup_{s \in [0,T]}\Big| [-(- \Delta)_{\mathbb{R}_{+}^{*}}^{\frac{\gamma}{2}} G ]_{\varepsilon} (s, \tfrac{x}{n}) -  [-(- \Delta)_{\mathbb{R}_{+}^{*}}^{\frac{\gamma}{2}} G ] (s, \tfrac{x}{n} ) \Big|  = \frac{c_{\gamma}}{n} \sum_{x = \varepsilon n}^{2bn} \sup_{s \in [0,T]} \Big|   \lim_{\varepsilon_1 \rightarrow 0^+}   \int_{ \varepsilon_1}^{ \varepsilon} \frac{\theta_{\frac{x}{n}}(s, w)}{ w^{\gamma+1}} dw \Big| \\
\leq & \frac{c_{\gamma}}{n} \sum_{x =\varepsilon n}^{2bn} \lim_{\varepsilon_1 \rightarrow 0^+}   \int_{ \varepsilon_1}^{ \varepsilon} \frac{\| \Delta G \|_{\infty} w^2 }{ w^{\gamma+1}} dw  \leq  \frac{2b c_{\gamma} \| \Delta G \|_{\infty}}{2 - \gamma} \varepsilon^{2-\gamma}, \forall n \geq 1,
\end{align*}
and since  $\gamma < 2$, it vanishes as $\epsilon \to 0$. From this we conclude that 
 \begin{equation*} 
 \lim_{\varepsilon \rightarrow 0^{+}}  \limsup_{n \rightarrow \infty} \frac{1}{n} \sum_{x=\varepsilon n}^{2bn} |  n^{\gamma}  \sum_{y =1}^{\infty} p(y-x) [ G (s, \tfrac{y}{n} ) - G (s, \tfrac{x}{n} ) ] - [-(- \Delta)_{\mathbb{R}_{+}}^{\frac{\gamma}{2}} G ] (s, \tfrac{x}{n} ) \Big| =0.
\end{equation*}
Finally, combining the double limit above with \eqref{lim1lapfracreg} and \eqref{lim6lapfracreg}, we get \eqref{limlapregpos}. This ends the proof.
\end{proof}
Since $|\eta_s^n(x)| \leq 1$, the next result is a direct consequence of the previous one.
\begin{cor} \label{corconvfast}
Let $\mcb S = \mcb S_0$. Assume that $\gamma \in (1,2)$ and $G \in \mcb S_{\textrm{Rob}}$ or that $\gamma \in (0,1]$ and $G \in \mcb S_{\textrm{Neu}}$. It holds 
\begin{equation*}  
\lim_{n \rightarrow \infty} \mathbb{E}_{\mu_n} \Big[ \sup_{t \in [0,T]} \Big| \int_0^t \Big\{  \sum_{ \{x,y\} \in \mcb F } n^{\gamma-1} [ G(s, \tfrac{y}{n} ) - G(s, \tfrac{x}{n} ) ] p(y-x) \eta_s^n(x)   -  \langle \pi_s^n, [-(- \Delta)_{\mathbb{R}^{*}}^{\frac{\gamma}{2}}   G] (s, \cdot) \rangle \Big\} \Big| \Big]  =0.
\end{equation*}
\end{cor}

We end this section with two results which are consequences of Proposition \ref{convdisc} and Proposition \ref{lemconvneum}. For the remainder of this section, we define $\mcb S_{\beta, \gamma}$ as
\begin{align*}
\mcb S_{\beta, \gamma} :=\mathbbm{1}_{\{\gamma \in (0,1]\}}\mcb S_{Neu}+\mathbbm{1}_{\{\gamma \in(1,2),\beta=0\}}\mcb S_{Dir}+\mathbbm{1}_{\{\gamma \in(1,2),\beta>0\}}\mcb S_{Rob0}.
\end{align*}
Now we present the final  results of this section, which were useful to treat \eqref{srobterm}.

\begin{prop} \label{convslow}
Let $\mcb S = \mcb S_0$. For every $G \in \mcb S_{\beta, \gamma}$ we have 
\begin{align*}  
\lim_{n \rightarrow \infty}   \frac{1}{n}  \sum_{x} \sup_{s \in [0,T]} \big|& n^{\gamma} \sum_{ \{x, y\} \in \mcb S } [ G(s, \tfrac{y}{n} ) - G(s, \tfrac{x}{n} ) ] p(y-x)  \\&- \mathbbm{1}_{\beta=0} \big(  [-(- \Delta)^{\frac{\gamma}{2}}   G] (s, \tfrac{x}{n} ) -  \sum_{x} [-(- \Delta)_{\mathbb{R}^{*}}^{\frac{\gamma}{2}}   G] (s, \tfrac{x}{n} )  \big) \big| =0.
\end{align*}
\end{prop}
\begin{proof}
If $\beta=0$, the result  comes directly from Proposition \ref{convdisc} and Proposition \ref{lemconvneum}. If $\beta >0$, there are two possibilities: $\gamma \in (1,2)$ and $\gamma \in (0,1]$.

\textbf{I).}  For $\gamma \in (1,2)$, we can choose $\delta \in [0,1] \cap (\gamma  - \beta - 1, \gamma - 1)$. Since $G \in \mcb S_{\textit{Rob0}}$, there exists $C > 0$ (independent of $s \in [0,T]$) such that $|G(s,u) - G(s,v)| \leq C |u-v|^{\delta}$ for every $u,v \in \mathbb{R}$ and every $s \in [0,T]$. This leads to
\begin{align*}
\frac{1}{n} \sum_{x} \sup_{s \in [0,T]} | n^{\gamma} \sum_{ \{x, y\} \in \mcb S } [ G(s, \tfrac{y}{n} ) - G(s, \tfrac{x}{n} ) ] p(y-x)| 
\lesssim  2 n^{\gamma-\beta-1 - \delta} \sum_{x=0}^{\infty} \sum_{y=1}^{\infty} (x+y)^{-\gamma-1+\delta}.
\end{align*}
Since $\delta \in (\gamma  - \beta - 1, \gamma - 1)$, the sum above is finite and the expression in the last display vanishes as $n \rightarrow \infty$, leading to the desired result.

\textbf{II).}  For $\gamma \in (0,1]$, we consider  $G \in \mcb S_{\textit{Neu}}$ and $\sup_{s \in [0,T]} |G(s,u)|=0$ for every $u: |u| \leq \bar{b}_G$ and every $u: |u| \geq b_G$. This leads to
\begin{align*}
& \frac{1}{n} \sum_{x} \sup_{s \in [0,T]} | n^{\gamma} \sum_{ \{x, y\} \in \mcb S } [ G(s, \tfrac{y}{n} ) - G(s, \tfrac{x}{n} ) ] p(y-x)| \\
\leq&2 \| G \|_{\infty} n^{\gamma-\beta-1}  \big[  \sum_{y= \bar{b}_G n}^{b_G n} \sum_{x= b_G n}^{\infty}p(x+y ) +  \sum_{y= 0}^{b_G n} \sum_{x= \bar{b}_G n}^{b_G n} p(x+y) \big]  \\
\lesssim & n^{-\beta} \big[ \frac{1}{n} \sum_{y= \bar{b}_G n}^{b_G n}  \big( b_G + \frac{y}{n} \big)^{-\gamma} +\frac{1}{n} \sum_{y=0}^{b_G n}  \big( \bar{b}_G  + \frac{y}{n} \big)^{-\gamma}  \big] \lesssim n^{-\beta} \big[  \int_{b_G +\bar{b}_G}^{2 b_G} u^{-\gamma} du + \int_{b\bar{b}_G}^{b_G+\bar{b}_G} u^{-\gamma} du].
\end{align*}
Since both integrals above are finite and $\beta>0$,  last display vanishes as $n \rightarrow \infty$.
\end{proof}
The next result is a trivial  consequence of the previous one by the fact that the variables $\eta(x)$ are bounded. 
\begin{cor} \label{corconvslow}
Let $\mcb S = \mcb S_0$. For every $G \in \mcb S_{\beta, \gamma}$ we have 
\begin{align*}  
\lim_{n \rightarrow \infty} \mathbb{E}_{\mu_n} \Big[ \sup_{t \in [0,T]}  \Big| \int_0^t \Big\{ & \alpha n^{\gamma-1-\beta}  \sum_{ \{x,y\} \in \mcb S } [ G(s, \tfrac{y}{n} ) - G(s, \tfrac{x}{n} ) ] p(y-x) \eta_s^n(x)   \\
- & \alpha \mathbbm{1}_{\beta=0} \big( \langle \pi_s^n, [-(- \Delta)^{\frac{\gamma}{2}}   G] (s, \cdot) \rangle - \langle \pi_s^n, [-(- \Delta)_{\mathbb{R}^{*}}^{\frac{\gamma}{2}}   G] (s, \cdot) \rangle  \big) \Big\} ds \Big| \Big]  =0.
\end{align*}
\end{cor}

\begin{prop} \label{tight2condaux}
Let $G \in \mcb S_{ \textrm{Dif} }$. Then
\begin{align*}
n^{\gamma-2} \sum_{x,y} \sup_{s \in [0,T]} [ G ( s,\tfrac{y}{n} ) - G (s, \tfrac{x}{n} ) ]^2 p(y - x) \lesssim \max\{ n^{\gamma-2}, n^{-1} \}.
\end{align*}
\end{prop}
\begin{proof}
Since $G \in \mcb S_{ \textit{Dif} }$, there exists $C_{\delta}$ such that $\sup_{s \in [0,T]} |G(s,u) - G(s,v) | \leq C_{\delta} |u-v|^{\delta}, \forall u,v \in \mathbb{R}$, where $\delta = \min \{1, \frac{\gamma+1}{2} \} \in [0,1]$. Since $ \sup_{s \in [0,T]} G(s,u)=0, \forall u: |u|\geq b_G$, we have 
\begin{align*}
&n^{\gamma-2} \sum_{x,y} \sup_{s \in [0,T]} [ G ( s,\tfrac{y}{n} ) - G (s, \tfrac{x}{n} ) ]^2 p(y - x)  \\
\lesssim\, & n^{\gamma-2-2\delta}  \Big[ \sum_{\substack{|x| < (b_G+1)n\\ |y| < b_Gn }}|y-x|^{2 \delta - \gamma -1}  +  \sum_{\substack{|x| < b_Gn\\ b_Gn \leq |y| < (b_G+1)n }}|y-x|^{2 \delta - \gamma -1} \Big] 
+  n^{\gamma-2} \sum_{\substack{|x| \geq (b_G+1)n\\ |y| < b_Gn}} |x-y|^{-\gamma-1}  \\
\lesssim \,& n^{\gamma-2-2\delta} n^2    +    n^{\gamma-1}  \sum_{t=n}^{\infty} t^{-\gamma-1} 
\lesssim\,  n^{\gamma-2\delta} + n^{-1}    \lesssim \max\{ n^{\gamma-2}, n^{-1} \},
\end{align*} 
where in the second and the last inequalities we used  the fact that $\delta = \min \{1, \frac{\gamma+1}{2} \}$.
\end{proof}

\section{Analysis tools} \label{secuniq}

In this section our goal is to prove the uniqueness of the weak solutions of \eqref{eqhydfracdifrob}, \eqref{eqhydfracdifreal},  \eqref{eqhydfracbetazero} and \eqref{eqhydfracdifnotuniq} (the last one assuming \eqref{conduniq}). Since we did not find in the literature a proof of uniqueness of our weak solutions we decided to prove it. Before doing so, we prove  some useful results.

\subsection{Well-definiteness of the fractional operators} \label{subsecwell}

In this section, we prove that the fractional operators that we deal with are well defined on our space of test functions. Recall \eqref{eq:B_g}. Let us begin with the fractional Laplacian.
\begin{prop} \label{lapfracglob}
Let $G \in C_c^2 (\mathbb{R})$. For every $u \in \mathbb{R}$, the limit
\begin{equation} \label{limlapfrac2wd}
 \lim_{\varepsilon \rightarrow 0^+}   [-(- \Delta)^{\frac{\gamma}{2}} G ]_{\varepsilon} (u )
\end{equation}
exists and there exists $H \in L^1(\mathbb{R})$ such that $|[-(- \Delta)^{\frac{\gamma}{2}} G ]_{\varepsilon} | \leq H, \forall \varepsilon \in (0,b_G)$.  
\end{prop}
The last result also holds  for the regional fractional Laplacian on  $I = \mathbb{R}_{-}^{*}$ or $I = \mathbb{R}_{+}^{*}$. Recall  \eqref{eq:barB_g}.
\begin{prop} \label{lapfracreg}
Let $I = \mathbb{R}_{-}^{*}$ or $I = \mathbb{R}_{+}^{*}$. Let $G \in C_c^2 (\mathbb{R}^{*})$ if $\gamma \in (1,2)$ and $G \in C_{c0}^2 (\mathbb{R})$ if $\gamma \in (0,1]$. Then for every $u \in I$, the limit
\begin{equation} \label{limlapfracpos2}
 \lim_{\varepsilon \rightarrow 0^+}   [-(- \Delta)_{I}^{\frac{\gamma}{2}} G ]_{\varepsilon} (u )
\end{equation}
exists and there exists $H \in L^1(I)$ such that $|[-(- \Delta)_{I}^{\frac{\gamma}{2}} G ]_{\varepsilon} | \leq H, \forall \varepsilon \in (0,b_G)$.
\end{prop}
\begin{proof}
Let $G$ be fixed and recall  \eqref{eq:B_g}. We assume that $I = \mathbb{R}_{+}^{*}$, but we observe that the case $I = \mathbb{R}_{-}^{*}$ is analogous. Define $H: \mathbb{R}_{+}^{*} \rightarrow \mathbb{R}$ by
\begin{equation*}
H(u):=
\begin{cases}
 \frac{c_{\gamma}  \|G \|_{\infty}}{\gamma} [ (u-b_G)^{-\gamma} - u^{-\gamma} ], \text{if} \; u > 2b_G; \\
c_{\gamma} \big[ \frac{  \|G'' \|_{\infty}}{2-\gamma} (2b_G)^{2-\gamma} + \frac{  \|G' \|_{\infty} }{\gamma-1}  u^{1-\gamma} \big], \text{if} \; \gamma \in (1,2) \; \text{and} \; 0 < u \leq  2b_G; \\
c_{\gamma} \big[ \frac{  \|G'' \|_{\infty}}{2-\gamma} (2b_G)^{2-\gamma} + \frac{ 2  \|G \|_{\infty}}{\gamma}u^{-\gamma}  \big], \text{if} \; \gamma \in (0,1] \; \text{and} \; \frac{\bar{b}_G}{2} < u \leq  2b_G; \\
c_{\gamma} \big[ \frac{ \|G' \|_{\infty}}{2-\gamma} (2b_G)^{2-\gamma} + \frac{   \|G \|_{\infty}}{\gamma} \frac{2^{\gamma}}{\bar{b}_G^{\gamma} }   \big], \text{if} \; \gamma \in (0,1] \; \text{and} \; 0 < u \leq \frac{\bar{b}_G}{2}.
\end{cases}
\end{equation*}
Let $u \in \mathbb{R}_+^{*}$. For every $\varepsilon \in (0, \min\{ \frac{u}{2} ,b_G \})$, it holds
\begin{align*}
[-(- \Delta)_{\mathbb{R}_{+}}^{\frac{\gamma}{2}} G ]_{\varepsilon} (u ) =   c_{\gamma}   \int_{ \varepsilon}^{u}  \frac{[G(u+w) - G(u) ]  +[ G(u-w) - G(u) ]}{ w^{\gamma+1}} dw +  c_{\gamma}   \int_{ u}^{\infty}  \frac{G(u+w) - G(u)}{ w^{\gamma+1}} dw.
\end{align*}
Now observe that there are two possibilities for $u$:  $u > 2b_G$ or  $0 < u \leq 2b_G$. For $u > 2b_G$, since $G(y)=0, \forall y \in  (b_G, \infty)$, we get
\begin{align*}
 &[-(- \Delta)_{\mathbb{R}_{+}}^{\frac{\gamma}{2}} G ]_{\varepsilon} (u ) = c_{\gamma}   \int_{ u-b_G}^{u}  \frac{ G(u-w) }{ w^{\gamma+1}} dw, \forall \varepsilon \in (0, b_G),
\end{align*}
and the limit in \eqref{limlapfracpos2} exists. Moreover,
\begin{align*}
|[-(- \Delta)_{\mathbb{R}_{+}}^{\frac{\gamma}{2}} G ]_{\varepsilon} (u )| = \Big| c_{\gamma}   \int_{ u-b_G}^{u}  \frac{ G(u-w) }{ w^{\gamma+1}} dw \Big| \leq c_{\gamma}   \int_{ u-b_G}^{u}  \frac{  \|G \|_{\infty} }{ w^{\gamma+1}} dw = H(u).
\end{align*}
For $0 < u \leq 2b_G$, since $G \in C^2( \mathbb{R}_{+}^{*} )$, performing two Taylor expansions of second order in $G$, we get
\begin{align*}
\Big| c_{\gamma}   \int_{ \varepsilon}^{u}  \frac{\left(G(u+w) - G(u) \right)  + \left( G(u-w) - G(u) \right)}{ w^{\gamma+1}} dw \Big| =&  c_{\gamma}    \Big| \int_{ \varepsilon}^{u}  \frac{ G''\left( \xi_1(w) \right)   +  G''\left( \xi_2(w) \right) }{2 w^{\gamma-1}} dw \Big| \\
\leq c_{\gamma} \frac{  \|G'' \|_{\infty}}{2-\gamma} (2b_G)^{2-\gamma}  < \infty,
\end{align*}
for some $\xi_1(w) \in (u, u+w)$ and some $\xi_2(w) \in (u - w, u)$. Now we distinguish two cases: $\gamma \in (1,2)$ or $\gamma \in (0,1]$.

\textbf{I).}  For $\gamma \in (1,2)$: we have
\begin{align*}
\Big| c_{\gamma}   \int_{ u}^{\infty}  \frac{G(u+w) - G(u)}{ w^{\gamma+1}} dw \Big| =c_{\gamma} \Big|  \int_{ u}^{\infty}  \frac{ G'\big(\xi_3(w)\big)}{ w^{\gamma}} dw \Big| \leq c_{\gamma}  \frac{ \|G' \|_{\infty} }{\gamma-1}  u^{1-\gamma}  < \infty
\end{align*}
for some $\xi_3(w) \in (u, u+w)$, and from this, the limit in \eqref{limlapfracpos2} exists. Moreover,
\begin{align*}
|[-(- \Delta)_{\mathbb{R}_{+}}^{\frac{\gamma}{2}} G ]_{\varepsilon} (u )| \leq c_{\gamma} \Big[ \frac{ \|G''\|_{\infty}}{2-\gamma} (2b_G)^{2-\gamma} + \frac{  \|G'\|_{\infty}}{\gamma-1}u^{1-\gamma}  \Big]= H(u).
\end{align*}
\textbf{II).} For $\gamma \in (0,1]$ we need to distinguish again two cases: $u \in [ \frac{\bar{b}_G}{2}, 2 b_G)$ or $u \in (0, \frac{\bar{b}_G}{2} )$. In the former case we have
\begin{align*}
\Big| c_{\gamma}   \int_{ u}^{\infty}  \frac{G(u+w) - G(u)}{ w^{\gamma+1}} dw \Big| \leq c_{\gamma}  \int_{u }^{\infty} \frac{2 \|G \|_{\infty} }{ w^{\gamma+1}} dw \Big| \leq c_{\gamma}  \frac{2 \|G \|_{\infty} }{\gamma} u^{-\gamma}  < \infty,
\end{align*}
and we conclude that the limit in \eqref{limlapfracpos2} exists. Moreover,
\begin{align*}
|[-(- \Delta)_{\mathbb{R}_{+}}^{\frac{\gamma}{2}} G ]_{\varepsilon} (u )| \leq c_{\gamma} \Big[ \frac{ \|G''|\|_{\infty}}{2-\gamma} (2b_G)^{2-\gamma} +\frac{2\|G \|_{\infty} }{\gamma} u^{-\gamma}  \Big]= H(u).
\end{align*}
In the later case i.e. $\gamma \in (0, \frac{\bar{b}_G}{2} )$, we have
\begin{align*}
\Big| c_{\gamma}   \int_{ u}^{\infty}  \frac{G(u+w) - G(u)}{ w^{\gamma+1}} dw \Big| =c_{\gamma} \Big|  \int_{ \frac{\bar{b}_G}{2} }^{\infty} \frac{G(u+w) }{ w^{\gamma+1}} dw \Big| \leq c_{\gamma}  \frac{ \|G \|_{\infty} }{\gamma} \frac{2^{\gamma}}{\bar{b}_G^{\gamma}}  < \infty
\end{align*}
and we conclude that the limit in \eqref{limlapfracpos2} exists. Moreover,
\begin{align*}
|[-(- \Delta)_{\mathbb{R}_{+}}^{\frac{\gamma}{2}} G ]_{\varepsilon} (u )| \leq c_{\gamma} \Big[ \frac{ \|G''\|_{\infty}}{2-\gamma} (2b_G)^{2-\gamma} +\frac{ \|G\|_{\infty} }{\gamma} \frac{2^{\gamma}}{\bar{b}_G^{\gamma}}  \Big]= H(u).
\end{align*}
To conclude the proof it remains to prove that $H \in L^1(\mathbb{R}^{*}_{+})$. For $u >2 b_G$, applying the mean value Theorem to  the function $x^{-\gamma}$ we obtain
$
(u-b_G)^{-\gamma} - u^{-\gamma}  \leq b_G \gamma (u-b_G)^{-\gamma-1}. 
$
For $\gamma \in (1,2)$, it holds
\begin{align*}
\int_{\mathbb{R}_{+}^{*} } |H(u)| du \leq \int_0^{2 b_G} c_{\gamma} \Big[ \frac{ \|G''\|_{\infty}}{2-\gamma} (2b_G)^{2-\gamma} + \frac{ \|G'\|_{\infty} }{\gamma-1}  u^{1-\gamma} \Big] du  + \int_{2 b_G}^{\infty} b_G \gamma (u-b_G)^{-\gamma-1} < \infty.
\end{align*}
while for  $\gamma \in (0,1]$, it holds
\begin{align*}
\int_{\mathbb{R}_{+}^{*} } |H(u)| du \leq& \int_0^{\frac{\bar{b}_G}{2}} c_{\gamma} \big[ \frac{ \|G'' \|_{\infty}}{2-\gamma} (2b_G)^{2-\gamma} + \frac{  \|G \|_{\infty}}{\gamma} \frac{2^{\gamma}}{\bar{b}_G^{\gamma} }   \big] du \\
+& \int_{ \frac{\bar{b}_G}{2} }^{2 b_G} c_{\gamma} \Big[ \frac{ \|G'' \|_{\infty}}{2-\gamma} (2b_G)^{2-\gamma} + \frac{ 2 \|G \|_{\infty}}{\gamma}u^{-\gamma}  \Big] du  + \int_{2 b_G}^{\infty} b_G \gamma (u-b_G)^{-\gamma-1} < \infty.
\end{align*}
\end{proof}

Now will extend Theorem 3.3 of \cite{reflected} (which is stated for bounded domains) to the cases $I=\mathbb{R}$, $I=\mathbb{R}_{-}^{*}$ and $I=\mathbb{R}_{+}^{*}$. We recall that we identify $[-(- \Delta)_{\mathbb{R}}^{\frac{\gamma}{2}} G ]_{\varepsilon}$ with $[-(- \Delta)^{\frac{\gamma}{2}} G ]_{\varepsilon}$ and $[-(- \Delta)_{\mathbb{R}}^{\frac{\gamma}{2}} G ]$ with $[-(- \Delta)^{\frac{\gamma}{2}} G ]$. 

\begin{prop} \label{fracintpart}
Let $I \in \{ \mathbb{R}, \mathbb{R}_{-}^{*}, \mathbb{R}_{+}^{*} \}$ and $\rho: \mathbb{R} \rightarrow \mathbb{R}$ a  function  belonging to   $L^\infty(\mathbb{R})$ such that
\begin{align} \label{semnorrho}
 \iint_{ I^{2} } [\rho(u) - \rho(v)]^2 |u-v|^{-1-\gamma} du dv < \infty.
\end{align}
If $I=\mathbb{R}$, let $G \in C_c^{2}(\mathbb{R})$. On the other hand, if $I=\mathbb{R}_{-}^{*}$ or $I=\mathbb{R}_{+}^{*}$, let $G \in C_c^2(\mathbb{R}^{*})$ if $\gamma \in (1,2]$ and let $G \in C_{c0}^2(\mathbb{R}^{*})$ if $\gamma \in (0,1]$. Then  
\begin{align} \label{intpart}
 \int_{I} \rho(u) [-(- \Delta)_{I}^{\frac{\gamma}{2}} G ] (u) du = - \frac{c_{\gamma}}{2}    \iint_{I^2} \frac{[ G(u) - G(v) ] [ \rho(u) - \rho(v) ] }{|u-v|^{1 + \gamma}} du dv.  
\end{align}
\end{prop}
\begin{proof}
 For every $k$ satisfying $k^{-1} < b_G$, it holds
\begin{align} \label{intpartk}
 \int_{I} \rho(u) [-(- \Delta)_{I}^{\frac{\gamma}{2}} G ]_{k^{-1}} (u) du = - \frac{c_{\gamma}}{2} \iint_{I^2} \frac{[ G(u) - G(v) ] [ \rho(u) - \rho(v) ] }{|u-v|^{1 + \gamma}}\mathbbm{1}_{\{ |u - v| \geq k^{-1} \}} du dv. 
\end{align}
From Propositions \ref{lapfracglob} and \ref{lapfracreg}, there exists $H \in L^1(I)$ such that $|\rho \cdot [-(- \Delta)_{I}^{\frac{\gamma}{2}} G ]_{k^{-1}} |$ is bounded from above by $\| \rho\|_{\infty} H$, for every $k^{-1} < b_G$. Moreover, we have
\begin{align*}
\iint_{ I^{2} } [G(u) - G(v)]^2 |u-v|^{-1-\gamma} du dv < \infty.
\end{align*}
Plugging this with \eqref{semnorrho} and H\"older's inequality, we get
\begin{align*}
\iint_{I^2} \frac{| G(u) - G(v) | | \rho(u) - \rho(v) | }{|u-v|^{1 + \gamma}} du dv<\infty.
\end{align*}
From the Dominated Convergence Theorem, making $k \rightarrow \infty$, the left-hand side and right-hand side of \eqref{intpartk} go to the left-hand side and right-hand side of \eqref{intpart}, respectively, leading to the desired result.
\end{proof}

\subsection{Results on fractional Sobolev spaces}

The following result is a particular case of Lemma 5.2 of \cite{hitchhiker}. We refer the interested reader to that article for a proof. 

\begin{prop} \label{rhoposH1}
Let $\gamma \in (1,2)$ and $f \in \mcb {H}^{\frac{\gamma}{2}}(\mathbb{R}_{-}^{*}), g \in \mcb{H}^{\frac{\gamma}{2}}(\mathbb{R}_{+}^{*})$. Let $\tilde{f}_{e}$ and $\tilde{g}_{e}$ be the even extensions of the continuous representatives $\tilde{f}$ and $\tilde{g}$, respectively. Then $\tilde{f}_{e}, \tilde{g}_{e} \in  \mcb{H}^{\frac{\gamma}{2}}(\mathbb{R})$. 
\end{prop}
The next result is useful in order to conclude that condition a) is a consequence of conditions $(2)$ and $(3)$ in the definition of weak solutions of \eqref{eqhydfracdifnotuniq}, see Remark \ref{rem:cons}.
\begin{prop} \label{4condition}
Let $\gamma \in (1,2)$ and $f \in \mcb H^{\frac{\gamma}{2}}( \mathbb{R}^{*} )$. Denote the continuous representatives of $f|_{\mathbb{R}_{-}^{*}}$ and $f|_{\mathbb{R}_{+}^{*}}$ by $\tilde{f}_{-}$ and $\tilde{f}_{+}$, respectively. Assume that $\tilde{f}_{-}(0)=\tilde{f}_{+}(0)$. Then there exists $\tilde{f} \in C^{\frac{\gamma-1}{2}}(\mathbb{R})$ such that $f = \tilde{f}$ almost everywhere on $\mathbb{R}$.
\end{prop}
\begin{proof}
Denote the even extensions of $\tilde{f}_{-}$ and $\tilde{f}_{+}$ by $\tilde{f}_{-,e}$ and $\tilde{f}_{+,e}$, respectively. From last result, we have that $\tilde{f}_{-,e}$, $\tilde{f}_{+,e}$ are in $\mcb{H}^{\frac{\gamma}{2}}(\mathbb{R})$. Then from Proposition \ref{holderrep}, there exists $C>0$ such that 
\begin{align} \label{conda}
|\tilde{f}_{-,e}(u) - \tilde{f}_{-,e}(v)| + |\tilde{f}_{+,e}(u) - \tilde{f}_{+,e}(v)| \leq C |u-v|^{\frac{\gamma-1}{2}}, \forall u,v \in \mathbb{R}.
\end{align}
Now define $\tilde{f}: \mathbb{R} \rightarrow \mathbb{R}$ by $\tilde{f}(u)=\mathbbm{1}_{u<0}\tilde{f}_{-}(u)+\mathbbm{1}_{u\geq 0}\tilde{f}_+(u)$.  Then $f = \tilde{f}$ almost everywhere on $\mathbb{R}$. Moreover, we have
\begin{align*}
&|\tilde{f}(u) - \tilde{f}(v)| = |\tilde{f}_{-,e}(u) - \tilde{f}_{-,e}(v)| \leq C |u-v|^{\frac{\gamma-1}{2}}, \forall u,v \in (-\infty,0);\\
&|\tilde{f}(u) - \tilde{f}(v)| = |\tilde{f}_{+,e}(u) - \tilde{f}_{+,e}(v)| \leq C |u-v|^{\frac{\gamma-1}{2}}, \forall u,v \in [0, \infty).
\end{align*}
Finally, for $u<0$ and $v \geq 0$, we have
\begin{align*}
|\tilde{f}(u) - \tilde{f}(v)| \leq  |\tilde{f}_{-,e}(u) - \tilde{f}_{-,e}(0)| + |\tilde{f}_{+,e}(0) - \tilde{f}_{+,e}(v)| \leq  C ( |u|^{\frac{\gamma-1}{2}} + |v|^{\frac{\gamma-1}{2}} ) \leq 2C |u-v|^{\frac{\gamma-1}{2}},
\end{align*}
where in the last line we made use of the hypothesis $\tilde{f}_{-}(0)=\tilde{f}_{+}(0)$ and of \eqref{conda}.
\end{proof}
The next result is useful in order to conclude that condition b) is a consequence of conditions $(2)$ and a) in the definition of weak solutions of \eqref{eqhydfracdifnotuniq}, see Remark \ref{rem:cons}.
\begin{prop} \label{5condition}
Let $\gamma \in (1,2)$ and $f \in L^{\infty}(\mathbb{R}) \cap  \mcb H^{\frac{\gamma}{2}}( \mathbb{R}^{*} )$ be such that there exists $\tilde{f} \in C^{\frac{\gamma-1}{2}}(\mathbb{R})$ satisfying $f = \tilde{f}$ almost everywhere in $\mathbb{R}$. Then $f \in  \mcb{H}^{\frac{\gamma}{2} - \delta}(\mathbb{R})$, for every $\delta \in ( 0, \frac{\gamma-1}{2})$.
\end{prop}
\begin{proof}
Let $\delta \in ( 0, \frac{\gamma-1}{2})$. Since both $f|_{(-\infty,0)} \in L^2 \big( - \infty,0) \big)$ and $f|_{(0,\infty)} \in L^2 \big( (0, \infty) \big)$, then $f \in L^2(\mathbb{R})$. Moreover, since $f|_{(-\infty,0)} \in \mcb {H}^{\frac{\gamma}{2}} \big( (-\infty, 0) \big) \subset \mcb{H}^{\frac{\gamma}{2} - \delta} \big( (-\infty, 0) \big)$ and $f|_{(0,\infty)} \in \mcb{H}^{\frac{\gamma}{2}} \big( (0,\infty) \big) \subset \mcb {H}^{\frac{\gamma}{2} - \delta} \big( (0,\infty) \big)$, we have
\begin{align*}
\int_{-\infty}^{0} \int_{-\infty}^{0} \frac{[f(u)-f(v)]^2}{|u-v|^{1+\gamma- 2 \delta }} dv du + \int_{0}^{\infty} \int_{0}^{\infty} \frac{[f(u)-f(v)]^2}{|u-v|^{1+\gamma - 2 \delta}} dv du < \infty.
\end{align*}
In order to prove that $f \in  \mcb{H}^{\frac{\gamma}{2} - \delta}(\mathbb{R})$, we need
\begin{align*}
\int_{0}^{\infty} \int_{-\infty}^{0} \frac{[f(u)-f(v)]^2}{|u-v|^{1+\gamma - 2 \delta}} dv du + \int_{-\infty}^{0} \int_{0}^{\infty} \frac{[f(u)-f(v)]^2}{|u-v|^{1+\gamma - 2 \delta}} dv du < \infty.
\end{align*}
We will only prove that the first double integral above is bounded, but we observe that the same holds for the second one with an analogous reasoning. Since $f$ is bounded, we have
\begin{align*}
\int_0^{\infty} \int_{-\infty}^{-1} \frac{[f(u)-f(v)]^2}{|u-v|^{1+\gamma - 2 \delta}} dv du \leq & \frac{(2 \| f \|_{\infty} )^2}{\gamma (\gamma-1)} 1^{1-\gamma + 2 \delta} < \infty.
\end{align*}
In an analogous way, we have
\begin{align*}
\int_1^{\infty} \int_{-1}^{0} \frac{[f(u)-f(v)]^2}{|u-v|^{1+\gamma - 2 \delta}} dv du \leq = \frac{(2 \| f \|_{\infty} )^2 (1-2^{1-\gamma})}{\gamma (\gamma-1)} 1^{1-\gamma + 2 \delta} < \infty.
\end{align*}
It remains to prove that
\begin{equation*} 
\int_0^{1} \int_{-1}^{0} \frac{[f(u)-f(v)]^2}{|u-v|^{1+\gamma + 2 \varepsilon}} dv du = \int_0^{1} \int_{-1}^{0} \frac{|\tilde{f}(u)- \tilde{f}(v)|^2}{|u-v|^{1+\gamma + 2 \varepsilon}} dv du < \infty.
\end{equation*}
By hypothesis, there exists $C >0$  such that
\begin{equation*}  
|\tilde{f}(u) - \bar{f}_{-}(v)|^2 \leq  (C |u-v|^{\frac{\gamma-1}{2}} )^2 = C^2 |u-v|^{\gamma-1}, \forall u,v \in \mathbb{R}. 
\end{equation*}
This leads to
\begin{align*}
& \int_0^{1} \int_{-1}^{0} \frac{[f(u)-f(v)]^2}{|u-v|^{1+\gamma-2 \delta}} dv du \leq  C^2 \int_0^{1} \int_{-1}^{0} |u-v|^{2 \delta -2}  dv du  < \infty,
\end{align*}
which leads to the desired result.
\end{proof}
For $\gamma \in (0,1]$, we will need the following density result, which is a consequence of Theorem 1.4.2.4 in \cite{grisvard} for unbounded intervals.
\begin{prop}
Assume $\gamma \in (0,1]$ and $I=\mathbb{R}_{-}^{*}$ or $I=\mathbb{R}_{+}^{*}$ Then  $C_{c}(I)$ is dense in $\mcb {H}^{\frac{\gamma}{2}}(I)$ with the norm $\| \cdot \|_{\mcb {H}^{\frac{\gamma}{2}}(I)}$. In particular, if $f \in \mcb H^{\frac{\gamma}{2}}(\mathbb{R}^{*})$, there exists a sequence $(g_k)_{k \geq 1}$ in $C_{c0}(\mathbb{R})$ such that $(g_k|_{ \mathbb{R}_{-}^{*} })_{k \geq 1}$ converges to $f|_{ \mathbb{R}_{-}^{*}}$ and $(g_k|_{ \mathbb{R}_{+}^{*} })_{k \geq 1}$ converges to $f|_{\mathbb{R}_{+}^{*}}$ with respect to the norms of $\mcb H^{\frac{\gamma}{2}} ( \mathbb{R}_{-}^{*} )$ and $\mcb H^{\frac{\gamma}{2}} ( \mathbb{R}_{+}^{*} )$, respectively.
\end{prop}
\begin{proof}
Without loss of generality, we can assume that $I=\mathbb{R}_{+}^{*}$. Let $f \in \mcb H^{\frac{\gamma}{2}}(\mathbb{R}_{+}^{*})$ and $\varepsilon >0$. Then there exists $k_0 \geq 1$ such that
\begin{align*}
\int_{\mathbb{R}_{+}^{*}-(0,k_0)} f^2(u) du + \iint_{(\mathbb{R}_{+}^{*})^2 - \big( (0,k_0) \big)^2} \frac{[f(u)-f(v)]^2}{|u-v|^{1+\gamma}} du dv < \varepsilon.
\end{align*}
From Theorem 1.4.2.4 in \cite{grisvard}, there exists $g \in C_c^{\infty} \big( (0, k_0 ) \big) \subset C_{c}(\mathbb{R}_{+}^{*}) \subset C_{c0}(\mathbb{R})$ such that
\begin{align*}
\int_{(0,k_0)} [f-g]^2(u) du + \iint_{ \big((0,k_0)\big)^2} \frac{[(f-g) (u)-(f-g)(v)]^2}{|u-v|^{1+\gamma}} du dv < \varepsilon,
\end{align*}
which leads to
\begin{align*}
\int_{\mathbb{R}_{+}^{*}} [f-g]^2(u) du + \iint_{(\mathbb{R}_{+}^{*})^2} \frac{[(f-g) (u)-(f-g)(v)]^2}{|u-v|^{1+\gamma}} du dv <2 \varepsilon.
\end{align*}
Since $\varepsilon$ is arbitrarily small, we get the desired result.
\end{proof}

 From Theorem 7.38 in \cite{sobolevadams}, we have that $C_c^{\infty}(\mathbb{R})$ is dense in $\mcb{H}^{\frac{\gamma}{2}}(\mathbb{R})$ with the norm $\| \cdot \|_{\mcb{H}^{\frac{\gamma}{2}}(\mathbb{R})}$.  
 
From Proposition 23.2 (d) in \cite{MR1033497}, we have that $P \big([0,T], \mcb{H}^{\frac{\gamma}{2}}(I)\big)$ is dense in $L^2\big(0,T; \mcb{H}^{\frac{\gamma}{2}}(I) \big)$ with the norm $\| \cdot \|_{L^2\big(0,T; \mcb{H}^{\frac{\gamma}{2}}(I) \big)} $. Then we can state a corollary of these results.
\begin{lem} 
$\mcb S_{Dif}$ is dense in $L^2\big(0,T; \mcb{H}^{\frac{\gamma}{2}}(\mathbb{R}) \big)$. Moreover, if $\gamma \in (0,1]$, then $P \big([0,T], C_c^{\infty} (I) \big)$ is dense in $L^2\big(0,T; \mcb{H}^{\frac{\gamma}{2}}(I) \big)$ when $I=\mathbb{R}_{-}^{*}$ or $I=\mathbb{R}_{+}^{*}$.
\end{lem}
 As a consequence of last result  we can prove the next two lemmas. We refer the reader to the proof of Lemma 6.1 in \cite{byronsdif}, which uses the same strategy.
\begin{lem} \label{lemuniqsdif}
Let $\rho \in L^2 \left( 0,T; \mcb{H}^{\frac{\gamma}{2}}(\mathbb{R}) \right)$ and $(H_k)_{k \geq 1}$ be a sequence of functions in $\mcb S_{Dif}$ converging to $\rho$ with respect to the norm of $L^2 \left( 0,T; \mcb{H}^{\frac{\gamma}{2}}(\mathbb{R}) \right)$. We define $G_k \in \mcb S_{Dif}$ by 
\begin{align*}
G_k (t,u) = \int_t^T H_k (s,u) ds, \forall t \in [0,T], \forall u \in \mathbb{R}, \forall k \geq 1.
\end{align*}
Let $I = \mathbb{R}$, $I = \mathbb{R}_{-}^{*}$ or $I = \mathbb{R}_{+}^{*}$. It holds
\begin{equation*} 
\lim_{k \rightarrow \infty} \int_0^T \int_{I} \rho(s,u) \partial_s G_k (s,u) du ds = - \int_0^T \int_{I} [ \rho(s,u) ]^2 duds,
\end{equation*}
and
\begin{equation*} 
\lim_{k \rightarrow \infty}  \int_0^T \int_{I} \rho(s,u) [ -(- \Delta)_{I}^{\frac{\gamma}{2}} G_k ] (s,u) du ds = -\frac{c_{\gamma}}{4}  \iint_{I^2}  \frac{[  \int_0^T  \rho(r,u) dr - \int_0^T  \rho(s,v) ds ]^2  }{|u-v|^{1 + \gamma}} du dv .
\end{equation*}
\end{lem}
If $\gamma \in (0,1]$, we will make use of the result below.
\begin{lem} \label{lemuniqsdif2}
Assume $\gamma \in (0,1]$. Let $\rho \in L^2 \left( 0,T; \mcb{H}^{\frac{\gamma}{2}}(\mathbb{R}^{*}) \right)$, $(H_{k,-})_{k \geq 1}$ be a sequence of functions in $P \big( [0,T],  C_c^{\infty}(\mathbb{R}_{-}^{*} ) \big)$ converging to $\rho|_{[0,T] \times \mathbb{R}_{-}^{*}}$ with respect to the norm of $L^2 \left( 0,T; \mcb{H}^{\frac{\gamma}{2}}(\mathbb{R}_{-}^{*}) \right)$ and $(H_{k,+})_{k \geq 1}$ be a sequence of functions in $P \big( [0,T],  C_c^{\infty}(\mathbb{R}_{+}^{*} ) \big)$ converging to $\rho|_{[0,T] \times \mathbb{R}_{+}^{*}}$ with respect to the norm of $L^2 \left( 0,T; \mcb{H}^{\frac{\gamma}{2}}(\mathbb{R}_{+}^{*}) \right)$. We define $H_k \in \mcb S_{Neu} $ by $H_k(t,u):=H_{k,-}(t,u)$ if $(t,u) \in [0,T] \times (-\infty,0)$ and $H_k(t,u):=H_{k,+}(t,u)$ if $(t,u) \in [0,T] \times [0, \infty)$. Finally, we define  $G_k \in \mcb S_{Neu}$ by 
\begin{align*}
G_k (t,u) = \int_t^T H_k (s,u) ds, \forall t \in [0,T], \forall u \in \mathbb{R}, \forall k \geq 1.
\end{align*}
Let $I = \mathbb{R}_{-}^{*}$ or $I = \mathbb{R}_{+}^{*}$. It holds
\begin{equation*} 
\lim_{k \rightarrow \infty} \int_0^T \int_{I} \rho(s,u) \partial_s G_k (s,u) du ds = - \int_0^T \int_{I} [ \rho(s,u) ]^2 duds,
\end{equation*}
and
\begin{equation*} 
\lim_{k \rightarrow \infty}  \int_0^T \int_{I} \rho(s,u) [ -(- \Delta)_{I}^{\frac{\gamma}{2}} G_k ] (s,u) du ds = -\frac{c_{\gamma}}{4}  \iint_{I^2}  \frac{[  \int_0^T  \rho(r,u) dr - \int_0^T  \rho(s,v) ds ]^2  }{|u-v|^{1 + \gamma}} du dv .
\end{equation*}
\end{lem}
In order to prove the uniqueness of weak solutions of \eqref{eqhydfracdifrob} when $\gamma \in (1,2)$ and $\beta=\gamma-1$, the following lemma will be useful. It is strongly inspired by the arguments in Section 4.4. of \cite{stefano}.
\begin{lem} \label{lemuniqrob}
Let $\gamma \in (1,2)$. Assume that   $\varrho:[0,T] \times \mathbb{R}\to\mathbb R$, $\rho\in L^\infty([0,T]\times\mathbb R)$, $\varrho \in  L^2 \big( 0,T; \mcb{H}^{\frac{\gamma}{2}} (\mathbb{R}) \big)$ and $\varrho(s, \cdot) \in C^0(\mathbb{R})$, for a.e. $s \in [0,T]$. Let $(H_{k})_{k \geq 1}$  be a sequence in $\mcb S_{\textrm{Dif}}$  converging to $\varrho$ with respect to the norm of $L^2 \left( 0,T; \mcb{H}^{\frac{\gamma}{2}} (\mathbb{R})\right)$. Then
\begin{align*}
\lim_{k \rightarrow \infty}  \int_0^T \int_s^T \varrho(s,0) H_k(r,0) dr ds  = \frac{1}{2} \Big[ \int_0^T \varrho(s,0)  \Big]^2 .
\end{align*}  
\end{lem}
\begin{proof}
From H\"older's inequality we have that
\begin{align*}
\Big| \int_0^T \int_s^T \varrho(s,0) H_k(r,0) dr ds - \frac{1}{2} \Big[ \int_0^T \varrho(s,0)  \Big]^2 \Big| =& \Big|  \int_0^T \int_s^T \varrho(s,0) [ H_k(r,0) - \varrho(r,0)] dr ds \Big| \\
\leq&  T^{\frac{3}{2}} \| \varrho \|_{\infty} \sqrt{  \int_0^T  [ f_k(r,0)]^2 dr}, 
\end{align*}
where $f_k:=H_k-\varrho, \forall k \geq 1$. By hypothesis, $(f_k)_{k \geq 1}$ converges to zero in $L^2 \big( 0,T; \mcb{H}^{\frac{\gamma}{2}} (\mathbb{R}) \big)$. Since $\varrho(r, \cdot) \in C^0(\mathbb{R})$, for a.e. $r \in [0,T]$, from Proposition \ref{holderrep}, there exists $C$ independent of $k$ such that 
\begin{align*}
|f_k(r,0)| =& \int_{0}^{1} |f_k(r,0)| du \leq \int_{0}^{1}  [|f_k(r,0)-f_k(r,u)|  + |f_k(r,u)| ] du  \\
\leq& C \|f_k(r, \cdot) \|_{\mcb{H}^{\frac{\gamma}{2}}(\mathbb{R})} + \int_{0}^{1} |f_k(r,u)| du \leq( C +1) \|f_k(r, \cdot) \|_{\mcb{H}^{\frac{\gamma}{2}}(\mathbb{R})}.
\end{align*}
Integrating over time and using the hypothesis, the proof ends. 
\end{proof}

\subsection{Uniqueness of weak solutions}

Recall the definition of $\mcb S_{\gamma}$ in \eqref{eqhydfracdifrob}. We observe that weak solutions of \eqref{eqhydfracdifrob} deal with $\mcb S_{\gamma}$ as the space of test functions and the uniqueness of the weak solutions of \eqref{eqhydfracdifrob} is equivalent to the following result.
\begin{prop} \label{uniqeqhydsdifrob}
Let $\varrho_1, \varrho_2$ be such that  $\varrho_1 -a$, $ \varrho_2 -a  \in L^2 \big( 0,T; \mcb{H}^{\frac{\gamma}{2}} (\mathbb{R}^{*}) \big)$, for some $a \in (0,1)$. If 
\begin{align*}
F_{\textrm{FrRob}}(t, \varrho_1, G,  \mcb g,\kappa  ) = 0 = F_{\textrm{FrRob}}(t, \varrho_2, G, \mcb g,\kappa ) , \forall t \in [0,T], \forall G \in \mcb S_{\gamma},
\end{align*} 
then $\varrho_1 = \varrho_2$ almost everywhere in $[0,T] \times \mathbb{R}$.
\end{prop}
\begin{proof}
Denote $\varrho_3:= \varrho_1 - \varrho_2= [\varrho_1 - a] - [\varrho_2 - a]$. Then $\varrho_3 \in  L^2 \big( 0,T; \mcb{H}^{\frac{\gamma}{2}} (\mathbb{R}^{*}) \big)$ and $\varrho_3(s,u) \in \mcb{H}^{\frac{\gamma}{2}}(\mathbb{R} ^{*})$, for a.e. $ s \in [0,T]$.
First we consider the case $\gamma \in (1,2)$. Let $\tilde{\varrho}_{3,-}$ be the even extension of the continuous representative of $\varrho_3(s,\cdot)|_{\mathbb{R}_{-}^{*}}$ and define $\tilde{\varrho}_{3,+}$ in the same way, replacing $\mathbb{R}_{-}^{*}$ by $\mathbb{R}_{+}^{*}$. From Proposition \ref{rhoposH1}, it follows that $\varrho_{3,-}(s, \cdot),\varrho_{3,+}(s, \cdot)\in \mcb{H}^{\frac{\gamma}{2}}(\mathbb{R})$. Then for every $t \in [0,T]$, for every $G \in \mcb S_{\textit{Rob}}$ we get that
\begin{align}\label{eq:impa}
0= & \int_{\mathbb{R}} \varrho_3(t,u) G(t,u) du -  \int_0^t \int_{\mathbb{R}_{-}^{*}} \tilde{\varrho}_{3,-}(s,u)  \partial_s  G(s,u) du ds -  \int_0^t \int_{\mathbb{R}_{+}^{*}} \tilde{\varrho}_{3,+}(s,u)  \partial_s  G(s,u) du ds  \\
- &  \int_0^t \int_{\mathbb{R}_{-}^{*}} \tilde{\varrho}_{3,-}(s,u) [ - (- \Delta)_{\mathbb{R}_{-}^{*}} G](s,u) du ds -  \int_0^t \int_{\mathbb{R}_{+}^{*}} \tilde{\varrho}_{3,+}(s,u) [ - (- \Delta)_{\mathbb{R}_{+}^{*}} G](s,u) du ds \\
+ &  \kappa  \int_0^t  [    \tilde{ \varrho}_{3+}(s,0)  -  \tilde{ \varrho}_{3,-}(s,0) ]  [  G(s,0^{+})    - G(s,0^{-})  ] ds.
\end{align}
Since $\tilde{\varrho}_{3,-}, \tilde{\varrho}_{3,+}  \in  L^2 \big( 0,T; \mcb{H}^{\frac{\gamma}{2}} (\mathbb{R}) \big)$, there exist two sequences $(H_{k,-})_{k \geq 1}, (H_{k,+})_{k \geq 1}\in\mcb {S}_{\textit{Dif}}$ such that $(H_{k,-})_{k \geq 1}$ (resp. $(H_{k,+})_{k \geq 1}$) converges to $\tilde{\varrho}_{3,-}$ (resp. $\tilde{\varrho}_{3,+}$) with respect to the norm of $L^2 \big( 0,T; \mcb{H}^{\frac{\gamma}{2}} (\mathbb{R}) \big)$. Define $G_{k,-} (t,u) :=  \int_t^T H_{k,-} (s,u) ds$ and $G_{k,+} (t,u) :=  \int_t^T H_{k,+} (s,u) ds$, para todo $(t,u) \in [0,T] \times \mathbb{R}$ e para todo $k \geq 1$. Moreover, define $G_{k} \in \mcb {S}_{\textit{Rob}}$ by
\begin{align*}
G_{k} (t,u) =\mathbbm{1}_{u \in (-\infty,0)} G_{k,-} (t,u) + \mathbbm{1}_{u \in [0, \infty)} G_{k,+} (t,u), \forall (t,u) \in [0,T] \times \mathbb{R}, \forall k \geq 1.
\end{align*}
In particular, $G_{k}  (T,u) = 0, \forall u \in \mathbb{R}, \forall k \geq 1$. Taking in \eqref{eq:impa} $t=T$  and $G=G_{k}$, we get
\begin{align}
0= &   - \int_0^T \int_{\mathbb{R}_{-}^{*}} \tilde{\varrho}_{3,-}(s,u)  \partial_s  G_{k,-}(s,u) du ds -  \int_0^T  \int_{\mathbb{R}_{+}^{*}}   \tilde{\varrho}_{3,+}(s,u)  \partial_s G_{k,+}(s,u) du   ds  \nonumber \\
- &  \int_0^t \int_{\mathbb{R}_{-}^{*}} \tilde{\varrho}_{3,-}(s,u) [ - (- \Delta)_{\mathbb{R}_{-}^{*}} G_{k,-}](s,u) du ds -  \int_0^t \int_{\mathbb{R}_{+}^{*}} \tilde{\varrho}_{3,+}(s,u) [ - (- \Delta)_{\mathbb{R}_{+}^{*}} G_{k,+}](s,u) du ds  \nonumber \\
+&  \kappa   \int_0^T  \int_s^T  [    \tilde{ \varrho}_{3+}(s,0)  -  \tilde{ \varrho}_{3,-}(s,0) ]  [  H_{k,+}(r,0)    - H_{k,-}(r,0)  ] dr ds, \forall k \geq 1. \label{equniqbarfor}
\end{align}
Since $G_{k,-}$, $G_{k,+}$ and $H_{k,+}-H_{k,-}$ are in $\mcb {S}_{\textit{Dif}}$, we can  use Lemma \ref{lemuniqsdif} and Lemma \ref{lemuniqrob}. Taking the limit in \eqref{equniqbarfor} when $k \rightarrow \infty$, we have
\begin{align*}
& \int_0^T \int_{\mathbb{R}_{-}} [ \tilde{\varrho}_{3,-}(s,u) ]^2 du ds + \int_0^T \int_{\mathbb{R}_{+}} [ \tilde{\varrho}_{3,+}(s,u) ]^2 du ds \\
+&  \frac{c_{\gamma}}{4}  \iint_{(\mathbb{R}_{-})^2}  \frac{[  \int_0^T  \tilde{\varrho}_{3,-}(r,u)  -\tilde{\varrho}_{3,-}(r,v) dr ]^2  }{|u-v|^{1 + \gamma}} du dv +  \frac{c_{\gamma}}{4}  \iint_{(\mathbb{R}_{+})^2}  \frac{[  \int_0^T  \tilde{\varrho}_{3,+}(r,u) - \tilde{\varrho}_{3,+}(r,v) dr ]^2  }{|u-v|^{1 + \gamma}} du dv \\
+& \frac{\kappa   }{2}  \Big( \int_0^T  [    \tilde{ \varrho}_{3+}(s,0)  -  \tilde{ \varrho}_{3,-}(s,0) ] ds \Big)^2=0,
\end{align*}
which implies that $\tilde{\varrho}_{3,-}, \varrho_{3,-}, \varrho_3$ are equal to zero almost everywhere on $[0,T] \times \mathbb{R}_{-}$ and $\tilde{\varrho}_{3,+}, \varrho_{3,+},  \varrho_3$ are equal to zero almost everywhere on $[0,T] \times \mathbb{R}_{+}$. Then $\varrho_1 = \varrho_2$ almost everywhere on $[0,T] \times \mathbb{R}$. 

It remains to consider the case $\gamma \in (0,1]$. For every $t \in [0,T]$, for every $G \in \mcb S_{\textit{Neu}}$, it holds
\begin{align}\label{eq:impa_1}
0= & \int_{\mathbb{R}} \varrho_3(t,u) G(t,u) du -  \int_0^t \int_{\mathbb{R}_{-}^{*}} \varrho_3(s,u)  \partial_s  G(s,u) du ds -  \int_0^t \int_{\mathbb{R}_{+}^{*}} \varrho_3(s,u)  \partial_s  G(s,u) du ds  \\
- &  \int_0^t \int_{\mathbb{R}_{-}^{*}} \varrho_3(s,u) [ - (- \Delta)_{\mathbb{R}_{-}^{*}} G](s,u) du ds -  \int_0^t \int_{\mathbb{R}_{+}^{*}} \varrho_3(s,u) [ - (- \Delta)_{\mathbb{R}_{+}^{*}} G](s,u) du ds,
\end{align}
Recalling that $\gamma \in (0,1]$ and $\varrho_3 \in  L^2 \big( 0,T; \mcb{H}^{\frac{\gamma}{2}} (\mathbb{R}) \big)$, there exists a sequence  $(H_{k})_{k \geq 1}\in\mcb {S}_{\textit{Neu}}$ such that $(H_{k})_{k \geq 1}|_{[0,T] \times \mathbb{R}_{-}^{*} }$ (resp. $(H_{k})_{k \geq 1}|_{[0,T] \times \mathbb{R}_{+}^{*} }$) converges to $\varrho_3|_{[0,T] \times \mathbb{R}_{-}^{*} }$ (resp. $\varrho_3|_{[0,T] \times \mathbb{R}_{+}^{*} }$) with respect to the norm of $L^2 \big( 0,T; \mcb{H}^{\frac{\gamma}{2}} (\mathbb{R}_{-}^{*} )\big)$ (resp. $L^2 \big( 0,T; \mcb{H}^{\frac{\gamma}{2}} (\mathbb{R}_{+}^{*} )\big)$). For every $k \geq 1$, define $ G_{k} \in \mcb S_{Neu}$ by $G_{k} (t,u) :=  \int_t^T H_{k} (s,u) ds, \forall (t,u) \in [0,T] \times \mathbb{R}$. In particular, $G_{k}  (T,u) = 0, \forall u \in \mathbb{R}, \forall k \geq 1$. Taking in \eqref{eq:impa_1} $t=T$  and $G=G_{k}$, we get
\begin{align}
0= &   - \int_0^T \int_{\mathbb{R}_{-}^{*}}  \varrho_3(s,u)  \partial_s  G_{k,-}(s,u) du ds -  \int_0^T  \int_{\mathbb{R}_{+}^{*}}   \varrho_3(s,u)  \partial_s G_{k,+}(s,u) du   ds  \nonumber \\
- &  \int_0^t \int_{\mathbb{R}_{-}^{*}} \varrho_3(s,u) [ - (- \Delta)_{\mathbb{R}_{-}^{*}} G_{k,-}](s,u) du ds -  \int_0^t \int_{\mathbb{R}_{+}^{*}} \varrho_3(s,u) [ - (- \Delta)_{\mathbb{R}_{+}^{*}} G_{k,+}](s,u) du ds. \label{equniqbarfor2}
\end{align}
Since $G_{k,-}$, $G_{k,+},H_{k,+}-H_{k,-}\in\mcb {S}_{\textit{Dif}}$, we can  use Lemma \ref{lemuniqsdif2}. Taking the limit in \eqref{equniqbarfor2} when $k \rightarrow \infty$, we have
\begin{align*}
& \int_0^T \int_{\mathbb{R}_{-}} [ \varrho_3(s,u) ]^2 du ds + \int_0^T \int_{\mathbb{R}_{+}} [\varrho_3(s,u) ]^2 du ds 
\\
+ &  \frac{c_{\gamma}}{4} \Big[  \iint_{(\mathbb{R}_{-})^2}  \frac{[  \int_0^T \varrho_3(r,u) - \varrho_3(r,v) dr ]^2  }{|u-v|^{1 + \gamma}} du dv + \iint_{(\mathbb{R}_{+})^2}  \frac{[  \int_0^T  \varrho_3(r,u) - \varrho_3(r,v) dr ]^2  }{|u-v|^{1 + \gamma}} du dv \Big]=0,
\end{align*}
which implies that $\varrho_3$ is equal to zero almost everywhere on $[0,T] \times \mathbb{R}_{-}^{*}$ and $[0,T] \times \mathbb{R}_{+}^{*}$. Then $\varrho_1 = \varrho_2$ almost everywhere on $[0,T] \times \mathbb{R}$.

\end{proof}
The uniqueness of the weak solutions of \eqref{eqhydfracdifreal}, \eqref{eqhydfracbetazero} and \eqref{eqhydfracdifnotuniq} (the last one assuming \eqref{conduniq}) are analogous to the proof given above, so we omit details.

\quad

\thanks{ {\bf{Acknowledgements: }}
P.C. thanks FCT/Portugal for financial support through the project Lisbon Mathematics PhD (LisMath). P.C. and P.G. thank  FCT/Portugal for financial support through the project 
UID/MAT/04459/2013.  B.J.O. thanks  Universidad Nacional de Costa Rica  for sponsoring the participation in  this article. This project has received funding from the European Research Council (ERC) under  the European Union's Horizon 2020 research and innovative programme (grant agreement   n. 715734).} The authors thank Milton Jara and C\'edric Bernardin for discussions around the subject.

\bibliographystyle{plain}
\bibliography{bibliografia}

\end{document}